\documentclass[preprint,12pt]{elsarticle}

\usepackage{graphicx}
\usepackage{hyperref}
\usepackage{amsmath}
\usepackage{amssymb}
\usepackage{tikz}
\usepackage{placeins}
\usepackage{soul,xcolor}
\setstcolor{red}

\newcommand{\Rom}[1]{\uppercase\expandafter{\romannumeral #1\relax}}
\usepackage{amsthm}
\usepackage[shortlabels]{enumitem}
\usepackage{chngcntr}
\usepackage{thmtools}
\usepackage{thm-restate}
\usepackage{cleveref}

\usepackage[top=18truemm,bottom=25truemm,left=22truemm,right=22truemm]{geometry}

\usepackage{lineno,hyperref}

\DeclareMathOperator{\Ima}{Im}

\makeatother

\theoremstyle{definition}
\newtheorem{definition}{Definition}[section]

\newtheorem*{remark}{Remark}

\theoremstyle{theorem}
\newtheorem{theorem}{Theorem}[section] 

\theoremstyle{corollary}
\newtheorem{corollary}{Corollary}[theorem]

\theoremstyle{lemma} 
\newtheorem{lemma}[theorem]{Lemma}

\theoremstyle{Proposition} 
\newtheorem{Proposition}[theorem]{Proposition}

\usepackage{amssymb}

\definecolor{IKblue}{RGB}{0,0,100}

\definecolor{KKgreen}{RGB}{0,100,0}

\newcommand{\calC}{{\mathcal C}}
\newcommand{\calD}{{\mathcal D}}
\newcommand{\calE}{{\mathcal E}}

\newcommand{\calH}{{\mathcal H}}

\newcommand{\calS}{{\mathcal S}}

\newcommand{\Tr}{{\rm Tr}}
\newcommand{\conv}{{\rm conv}}
\newcommand{\ext}{{\rm ext}}

\usepackage{enumitem}
\makeatletter
\def\namedlabel#1#2{\begingroup
	Axiom #2.%
	\def\@currentlabel{$\mathbf{#2}$}%
	\phantomsection\label{#1}\endgroup
}
\makeatother

\definecolor{BSorange}{RGB}{255,127,0}

\journal{Annals of Physics}

\begin{document}
	
\begin{frontmatter}	
	\title{Fusion rules from entanglement}

\author{Bowen Shi}
\address{Department of Physics, The Ohio State University, Columbus, OH 43210, USA}
\author{Kohtaro Kato}
\address{Institute for Quantum Information and Matter, California Institute of Technology, Pasadena, CA, USA}
\author{Isaac H. Kim}
\address{Perimeter Institute for Theoretical Physics, Waterloo ON N2L 2Y5, Canada}
\address{Institute for Quantum Computing, University of Waterloo, Waterloo ON N2L 3G1, Canada}
\address{IBM T. J. Watson Research Center, Yorktown Heights NY, USA}
\address{Stanford University, Stanford CA 94305, USA}

	\date{\today}
	
	\begin{abstract}
We derive some of the axioms of the algebraic theory of anyon [A. Kitaev, Ann. Phys., 321, 2 (2006)] from a conjectured form of entanglement area law for two-dimensional gapped systems. We derive the fusion rules of topological charges and show that the multiplicities of the fusion rules satisfy these axioms. Moreover, even though we make no assumption about the exact value of the constant sub-leading term of the entanglement entropy of a disk-like region, this term is shown to be equal to $\ln \mathcal{D}$, where $\mathcal{D}$ is the total quantum dimension of the underlying anyon theory. These derivations are rigorous and follow from the entanglement area law alone. More precisely, our framework starts from two local entropic constraints which are implied by the area law. From these constraints, we prove what we refer to as the ``isomorphism theorem." The existence of superselection sectors and fusion multiplicities follow from this theorem, even without assuming anything about the parent Hamiltonian. These objects and the axioms of the anyon theory are shown to emerge from the structure and the internal self-consistency relations of the information convex sets.
    \end{abstract}

\end{frontmatter}

\section{Introduction}

One of the outstanding questions in modern physics concerns the classification of quantum phases. Many attempts have been already made to classify quantum phases over the past decade. For instance, gapped free-electron systems are completely classified~\cite{Kitaev2009,Schnyder2008}. For more general short-range entangled states, an approach based on cobordism was proposed~\cite{Kapustin2014}. One-dimensional (1D) gapped systems are completely classified at this point~\cite{Chen2011,Schuch2011,Szehr2016,Fidkowski2011}. A general gapped two-dimensional (2D) systems are expected to be described within the framework topological quantum field theory; see \cite{Barkeshli2014}, for example. 

This whole slew of different approaches raises a natural question. Why are there so many different approaches, and how can we ever be sure that the classification is complete? The main difficulty lies in identifying the correct framework. In the presence of interaction, one often needs to make a nontrivial assumption. The only exception so far is the one-dimensional (1D) gapped system. Hastings' theorem~\cite{Hastings2007} implies that any gapped 1D system obeys an area law. This subsequently implies that a matrix product state can approximate the ground state with a moderate bond dimension. It is this result from which a classification of quantum phases of 1D gapped system~\cite{Chen2011,Schuch2011,Szehr2016} follows.

However, in higher dimensions, an analog of Hastings' theorem is unknown. This is mainly because proving area law in 2D gapped systems remains challenging. Furthermore, even if area law turns out to be correct, the states that satisfy area law may not be well-approximated by an efficient tensor network~\cite{Ge2016}. These facts suggest that a classification program in 2D cannot merely mimic the classification program for 1D gapped systems. In fact, in any classification proposal based on tensor networks, there will always be a lingering question on whether we are not missing any unknown phases.

While it is widely believed at this point that topological quantum field theory (TQFT) describes all possible gapped phases in 2D, there is currently no rigorous argument that supports this belief. The existence of a three-dimensional (3D) gapped phase  outside of the TQFT framework~\cite{Haah2011} shows that there may be gapped phases of matter that lie outside of the TQFT framework. Even if TQFT turns out to be the correct framework in 2D, understanding of where this framework comes from remains as an important fundamental problem.

Motivated by this state of affairs, we initiate a program in which a familiar set of axioms of TQFT can be derived from a seemingly innocuous assumption about entanglement. We show that some of the basic concepts of the algebraic theory of anyon~\cite{2006AnPhy.321....2K}, i.e., superselection sectors and fusion multiplicities, emerge from a familiar form of entanglement area law~\cite{Kitaev2006,Levin2006}:
\begin{equation}
S(A) = \alpha \ell - \gamma, \label{eq:area_law}
\end{equation}
where $S(A)$ is the von Neumann entropy of a simply connected region $A$, $\ell$ is the perimeter of $A$, and $\gamma$ is a constant correction term\footnote{Nontrivial sectors, e.g., non-Abelian anyons and topological defects, can modify the constant term. It is a widely-adopted notation that $\gamma$ is the constant term for a disk in the absence of nontrivial sectors.} that only depends on the topology of $A$. The sub-leading correction, which vanishes in the $\ell \to \infty$ limit, is suppressed here.  

We then show that our definition of the fusion multiplicities satisfies all the properties one would have expected from the algebraic theory of anyon. Again, these properties are derived from Eq.~\eqref{eq:area_law}. Moreover, we further derive the following well-known formula:
\begin{equation}
\gamma = \ln \mathcal{D}, \nonumber
\end{equation}
where $\mathcal{D}$ is the total quantum dimension of the anyon theory we defined. Our derivation is rigorous under the assumption (Eq.~\eqref{eq:area_law}) and is completely independent from the previous approaches, i.e., an approach based on an effective field theory description ~\cite{Kitaev2006} and explicit calculations in exactly solvable models~\cite{Levin2006}.

While our assumption is not as rigorous as Hastings' proof of the 1D area law, it is something that is widely accepted at this point. Therefore, we believe this would be a reasonable starting point to obtain a general understanding of gapped phases. A similar, but a markedly different starting point of our work would be the two axioms we have identified. These two axioms are entropic conditions on bounded-radius disks (Axiom \ref{as:A0'} and \ref{as:A1} in Sec.~\ref{sec:axiomsdef}). We can show that these two axioms follow from Eq.~(\ref{eq:area_law}), but after that, we never use Eq.~(\ref{eq:area_law}) explicitly. All of our results follow directly from the axioms.

In other words, the axioms of the anyon fusion theory can be derived from Eq.~\eqref{eq:area_law}. The same conclusion follows from our axioms as well; see Axiom \ref{as:A0'} and \ref{as:A1}. While the conclusion would be the same either way, we would like to advocate for the use of the axioms over Eq.~\eqref{eq:area_law} for the following reasons. The first reason is that Axiom \ref{as:A0'} and \ref{as:A1} are assumed to hold on patches whose size is independent of the system size. Therefore, in principle, one can verify these axioms in time that scales linearly with the system size. Under a promise that the state is translation-invariant, the time can be reduced to a constant. In contrast, Eq.~\eqref{eq:area_law} is defined over length scales that are comparable to the system size. Verifying this assumption will incur an exponential computational cost. Secondly, in the continuum limit, the leading term of Eq.~\eqref{eq:area_law} depends on the ultraviolet-cutoff. On the other hand, the axioms manifestly cancel out this divergent piece.

Our framework is completely Hamiltonian-independent, in the sense that we only require the existence of a global state on the system satisfying the two local entropic constraints. This work is motivated from a number of recent observations: that local reduced density matrices of topological quantum phases often have a quantum Markov chain structure~\cite{2004CMaPh.246..359H,Kim2013,2014arXiv1405.0137K,Kim2015sydney,Kim2015,2015arXiv150807006J,2016PhRvA..93b2317K,2019PhRvB..99c5112S,2019PhRvR...1c3048S}. The key overarching concept is a convex set of density matrices introduced in \cite{Kim2015sydney}, which is later rediscovered and studied under the name ``information convex (set)"~\cite{2019PhRvB..99c5112S,2019PhRvR...1c3048S}. Roughly speaking, this is a set of density matrices which are locally indistinguishable from some reference state. In our context, this reference state would be the ground state of some local Hamiltonian. However, we do not use the fact that the state is a ground state.

Our framework opens up a concrete route to classify gapped quantum phases without resorting to ad-hoc assumptions. 
In addition, we believe our framework is capable of answering a long-standing open question about topological phases. The question is if a single ground state contains all the data necessary to define a topological phase. Given that we can define a notion of topological charges and fusion multiplicities from a single ground state, progress may be made by using our framework.
Our approach can be generalized to a broader context, e.g., to higher dimensions and to setups in which a topological defect~\cite{Bombin2010} or a boundary is present~\cite{Bravyi1998}. We will discuss these applications in our upcoming work.

The rest of this paper is organized as follows. In Sec.~\ref{Sec:summary},  we specify our formal setup and summarize our main results. 
In Sec.~\ref{Sec:Isomorphism}, we prove fundamental properties of the information convex sets, which are the key to obtaining some of the axioms of the algebraic theory of anyon. We shall refer to this part of the full algebraic theory of anyon as the \emph{anyon fusion theory} from now on. In Sec.~\ref{Sec.Fusion_data}, we define the notion of superselection sectors and fusion multiplicities in our framework and prove that the definition satisfies all the axioms of the anyon fusion theory. 
In Sec.~\ref{Sec. TEE}, we show that the constant term $\gamma$ in the area law equals the logarithm of the total quantum dimension. In Sec.~\ref{Sec: discussion}, we conclude with a discussion.

\section{Setup and Summary}\label{Sec:summary}

Let us begin with a general setup and state our physical assumptions. Before we delve into the details, it will be instructive to discuss the physical system we have in mind. We are envisioning a 
gapped system in 2D, which consists of microscopic degrees of freedom, e.g., spins. We would like to coarse-grain these microscopic degrees of freedom so that we can view non-overlapping blocks of spins as gigantic ``supersites," see Fig.~\ref{fig:cells}. We can consider the limit in which the length scale of each block is large compared to the correlation length. We would like to define a sensible notion of distance between the subsystems as well as their topologies. 
\begin{figure}[h]
	\centering
    \includegraphics[scale=1.2]{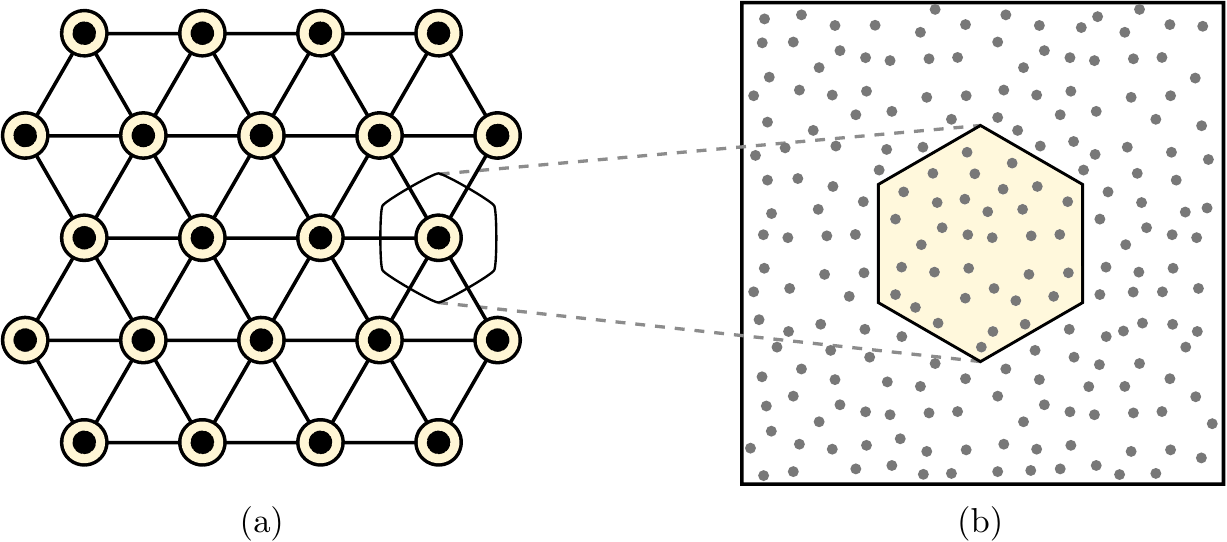}
	\caption{An illustration of the setup. (a) A graph $G=(V,E)$ with each vertex represents a ``supersite" which contains a cluster of microscopic degrees of freedom, e.g. spins in real space. The edges encode the locality of the underlying physical system and it allows us to define a notion of topology for a set of vertices.  (b) A zoomed-in depiction of a supersite.  It contains a block of physical spins. The length scale of each block is a constant that is large compared to the correlation length.}
	\label{fig:cells}
\end{figure}

More concretely, we can consider a quantum many-body spin system  with a tensor product structure $\mathcal{H} = \otimes_{v\in V} \mathcal{H}_v$, where  $\mathcal{H}_v$ is a finite-dimensional Hilbert space and $V$ is a set of vertices of a finite graph $G=(V,E)$ defined on a 2D closed manifold\footnote{Our results work for both orientable and non-orientable manifolds provided that a reference state $\sigma$ exists.}. 
By specifying the set of edges $E$, we can define a natural notion of distance on this graph (the graph distance). We denote the state space of $\calH$ by $\calS(\calH)$, which is the set of all density matrices on $\calH$. We say operator $O$ has support on $X\subseteq V$ if $O=O_X\otimes I_{V\backslash X}$ where $I_{V\backslash X}$ is the identity operator. 

We assume that there is a partition of the manifold into simply connected subsystems so that each $v\in V$ is associated with one of these subsystems.  Furthermore, there is an edge between $v_1, v_2 \in V$ if and only if the subsystems associated with the vertices are adjacent to each other. 
This assumption lets us define a notion of topology for a set of vertices. Without loss of generality, let $U$ be a set of vertices and $\mathcal{U}$ be a union of the subsystems associated with the vertices in $U$. The topology of $U$ is defined as the topology of $\mathcal{U}$. By construction, a single vertex is topologically a disk. However, a more general topology, e.g., an annulus, can be built out of a union of the vertices.

Throughout this paper, we consider a state $\sigma\in\calS(\calH)$ satisfying two axioms shown below. We will call this state the {\it (global) reference state}. We always use $\sigma$ to refer to the same reference state unless specified otherwise.
We use $\sigma_A$ to represent the reduced density matrix of $\sigma$ on a region $A$. Physically interesting examples of the reference state are the ground states of gapped local Hamiltonians. However, our derivations are only based on the properties of the quantum state. Technically, we are allowed to assume the global state to be pure $\sigma=\vert \psi\rangle\langle \psi\vert$ without loss of generality. This is because one can always show the existence of a pure state which has the same local reduced density matrices as the given reference state (see Theorem~\ref{Prop: Sigma(M)}). 

For the readers' convenience, we have summarized the key concepts in Table~\ref{tab:concepts}.\footnote{We thank the referee for this suggestion.} 
\begin{table}[h]
    \centering
    \begin{tabular}{|c|c|c|}
    \hline	
    Notation & Brief description & Reference \\
    \hline
    $\sigma$ & The global reference state that satisfies Axiom \ref{as:A0'} and \ref{as:A1} & N/A\\
    %\hline
    $\sigma_b$ & Reduced density matrix of $\sigma$ over region $b$ & \\
    %\hline
    $\mu(r)$ & A set of $\sigma_b$ over $b$ whose radius is smaller than $r$ &  Eq.~(\ref{eq:mu(r)})\\
    %\hline
    $\Sigma(\Omega)$ & The information convex set of a region $\Omega$ & Definition~\ref{def: info_convex} \\
    \hline
    \end{tabular}
    \caption{A list of notations and their descriptions.}
    \label{tab:concepts}
\end{table}

\subsection{Axioms}\label{sec:axiomsdef}
We start by defining a set of density matrices
\begin{equation}
\mu(r) = \{\sigma_b| b\in \mathcal{B}(r) \}, \label{eq:mu(r)}
\end{equation}
where $\mathcal{B}(r)$ is a set of balls of radius less or equal to $r$ and $\sigma_b$ is the reduced density matrix of the reference state ($\sigma$) on $b$. Because $r$ will be chosen to be a constant independent of the system size, we will simply denote $\mu(r)$ by $\mu$. We will refer to the set of $b\in \mathcal{B}(r)$ as the set of $\mu$-disks.

The axioms of our framework concern two entropic constraints on the set of $\mu$-disks.\footnote{The two local entropic constraints \ref{as:A0'} and \ref{as:A1} are originally proposed by one of us in Ref.~\cite{2014arXiv1405.0137K} The first attempt at deriving the axioms of anyon theory from these conditions was presented in Ref.~\cite{Kim2015sydney}.} Let 
$S(\rho)=-\Tr(\rho\ln\rho) $ be the von Neumann entropy of a state $\rho$. We assume that Axiom \ref{as:A0'} and \ref{as:A1} hold for all $\mu$-disks.

\begin{description}
\item[\namedlabel{as:A0'}{A0}]	For any $\sigma_b\in \mu$, for any configuration of subsystems $BC\subseteq b$ topologically equivalent to the one described in Fig.~\ref{A0},
	\begin{equation}
	S(\sigma_{BC})+S(\sigma_{C})-S(\sigma_{B})=0.
	\label{eq:tau_A0}
	\end{equation}
\end{description}

\begin{figure}[h]
	\centering
    \includegraphics[scale=1.1]{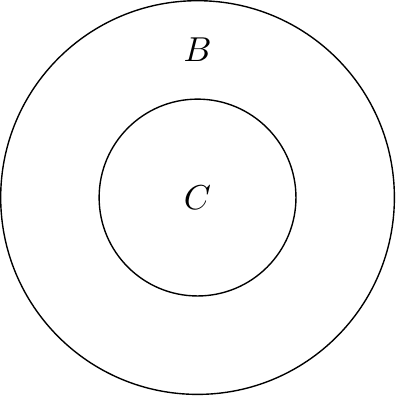}
	\caption{ A disk is divided into its ``core" ($C$) and outer boundary ($B$). The boundary ($B$) is chosen to be thick enough so that correlation between $C$ and the complement of $BC$ is negligible.}
	\label{A0}
\end{figure}

\begin{description}
\item[\namedlabel{as:A1}{A1}]
	For any $\sigma_b \in \mu$, for any configuration of subsystems $BCD\subseteq b$ topologically equivalent to the one described in Fig.~\ref{A1},
	\begin{equation}
	S(\sigma_{BC})+S(\sigma_{CD})-S(\sigma_{B})-S(\sigma_{D})=0.
	\label{eq:tau_A1}
	\end{equation}
\end{description}
\begin{figure}[h]
	\centering
    \includegraphics[scale=1.1]{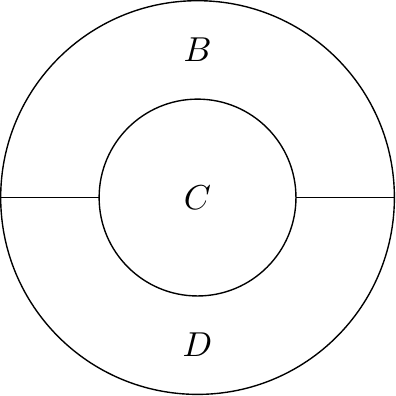}
	\caption{A disk is divided into its ``core" ($C$) and its outer boundary, which is further divided into two pieces ($B$ and $D$).
		\label{A1}}
\end{figure}

To see the physical meaning of these axioms, we observe that Eq.~(\ref{eq:tau_A0}) implies  
\begin{equation}\label{eq:vanishMI}
I(A:C)_\sigma \equiv (S_A+S_C-S_{AC})_{\sigma}=0,
\end{equation}
where $A$ is contained in the complement of $BC$ and $I(A:C)$ is the mutual information. This result follows from the strong subadditivity  (SSA) of von Neumann entropy~\cite{1973JMP....14.1938L}. The mutual information is a measure of the total amount of bipartite correlation. Therefore, Axiom \ref{as:A0'} can be viewed as a formalization of the intuition that long-range two-point correlation vanishes in gapped ground states.
However, note that the assumption is strictly stronger than the vanishing of mutual information itself. For instance, an infinite temperature Gibbs state satisfies Eq.~\eqref{eq:vanishMI} but does not satisfy \ref{as:A0'}.
Eq.~(\ref{eq:tau_A1}) implies that the \emph{quantum conditional mutual information} vanishes:
\begin{equation}\label{eq:vanishCMI}
I(A:C|B)_\sigma\equiv (S_{AB}+S_{BC}-S_{B}-S_{ABC})_{\sigma}=0,
\end{equation}
where $A$ is contained in the complement of $BCD$. Again, \ref{as:A1} is a stronger condition than Eq.~\eqref{eq:vanishCMI}. 

While these axioms can be derived from Eq.~\eqref{eq:area_law},  they are weaker assumptions and may hold in more general settings. Moreover, we expect these axioms to be a more well-defined way of formulating our assumptions compared to Eq.~\eqref{eq:area_law}  because they manifestly get rid of the leading divergent term in the area law. In reality, we expect our assumptions to be satisfied only approximately, up to an error that decays exponentially with $\frac{r}{\xi}$, where $r$ is the length scale of the subsystems and $\xi$ is the correlation length.  
We believe our framework has a natural extension to these cases since every theoretical tool we use can be generalized to such situations (see also Discussion).

How we use these axioms will be explained later in this paper, starting from Section \ref{Sec:Isomorphism}. However, for the hasty readers, we can briefly summarize our intuition as follows. The intuition behind \ref{as:A0'} is that the correlation between two sufficiently separated subsystems is negligible. This fact allows us to decouple two such subsystems without affecting either of them too much. The intuition behind \ref{as:A1} is more profound and subtle. Let us consider the following example. Suppose we have a density matrix over $AB$ and another density matrix over $BC.$ Our goal is to ``merge" these two density matrices; we want to construct a density matrix over $ABC$ whose marginal density matrices on $AB$ and $BC$ are consistent with the given data. It is well-known that one cannot generally do this in quantum systems. For example, if $AB$ is a maximally entangled state and $BC$ is also (the same) maximally entangled state, there cannot be an extension of these density matrices into a single state acting on $ABC$. By assuming Axiom \ref{as:A1}, we can bypass this difficulty and merge the density matrices that are given to us. The majority of our key results follow from this merging process. In particular, we shall extensively use nontrivial identities involving the entropies of the merged subsystems.

\subsection{Main results\label{sec:main_results}}
Our framework employs the notion of \emph{information convex set} \cite{Kim2015sydney,2019PhRvB..99c5112S}, the structure of which has been recently conjectured~\cite{2019PhRvR...1c3048S}. 
Loosely speaking, this is a set of reduced density matrices on a given subsystem; these density matrices are locally indistinguishable from the global reference state, but with an extra structure on this set in order to facilitate our analysis; see Sec.~\ref{Sec.Info_convex_definition} for the details. 
Crucially, information convex sets are defined from a single state. As such, we do not need to invoke any assumption about the parent Hamiltonian.

The following results follow from Axiom \ref{as:A0'} and \ref{as:A1}. Again, we do not assume anything about the parent Hamiltonian.

\begin{enumerate}
	\item \textbf{Isomorphism theorem}\\
	Let $A$ and $B$ be two subsystems which can be smoothly deformed from one to another. We show that the information convex sets associated with $A$ and $B$ are \emph{isomorphic}  (Theorem~\ref{thm: the isomorphism theorem}). These sets can be mapped onto each other by a linear bijective map. Moreover, these maps preserve the distance and the entropy difference between the elements of the information convex set. Concretely, let $\Phi$ be one such map. Then $D(\rho, \rho') = D(\Phi(\rho), \Phi(\rho'))$ for any distance measure $D(\cdot, \cdot)$. Moreover, $S(\rho) - S(\rho') = S(\Phi(\rho)) - S(\Phi(\rho'))$. Therefore, the structure of the information convex set only depends on the topology of the region associated with it.

	\item \textbf{A well-defined notion of topological charges}\\
	We show that the information convex set of an annulus is a simplex whose extreme points are orthogonal to each other. That is, any state $\rho$ in the information convex set of an annulus must have the following form (Theorem~\ref{Prop: structure_2}):
	\begin{equation*}
	\rho = \bigoplus_a p_a \,\sigma^a\,,
	\end{equation*}
	where $\{p_a\}$ is a probability distribution over a finite set, and $\sigma^a$ is a state only depending on the choice of the region. We define the label $a$ as a topological charge/superselection sector of the system. In exactly solvable models (e.g., toric code), $\sigma^a$ corresponds to a reduced density matrix of an annulus that surrounds a topological charge $a$. Different superselection sectors are perfectly distinguishable from each other. Moreover, the isomorphism theorem (see the first main result) implies that the charge is globally well-defined. We furthermore prove that for each sector, there exists a unique anti-sector (Sec.~\ref{Sec. Fusion}).

	\item\textbf{Extracting fusion multiplicities}\\
We completely characterize the information convex set of a $2$-hole disk. We show that any element in this set can be expressed as
\begin{equation*}
    \rho = \bigoplus_{a,b,c \in \mathcal{C}}  p_{ab}^c \, \rho^{abc},
\end{equation*}
where $\mathcal{C}$ is the set of topological charges, $\{p_{ab}^c \}$ is a probability distribution, and $\{ \rho^{abc} \}$ are mutually orthogonal quantum states
labeled by $a,b,$ and $c$. We show that, for each choice of $a$, $b$ and $c$, the set $\{ \rho^{abc} \}$ forms a 
state space of a finite-dimensional Hilbert space; see Theorem~\ref{Thm:}.
The dimensions of these Hilbert spaces are identified as the fusion multiplicities of the underlying anyon theory. In the context of the anyon theory, $a$ and $b$ label the two charges associated with the holes and $c$ represents the total charge of the disk.

	\item\textbf{Axioms of the anyon fusion theory}\\
	Our definition of the fusion multiplicities (see the third main result) satisfies all the axioms of the anyon fusion theory. 
	The proof is based on the merging technique~\cite{2016PhRvA..93b2317K}. We derived several consistency relations by ``merging" two density matrices and comparing the entropy of the reduced density matrices before and after the merge. An equation that relates these entropies leads to the axioms of the anyon fusion theory.
	
	\item\textbf{Topological entanglement entropy}\\
	From our definition of the fusion multiplicities, we further show that the sub-leading constant term $\gamma$ in Eq.~\eqref{eq:area_law} is the logarithm of the total quantum dimension. Unlike in \cite{Kitaev2006}, our derivation makes no assumption about the underlying effective field theory. 
	
\end{enumerate}

As a corollary, many of the anyon data can, in principle, be extracted from the local information of a single ground state. Importantly, we do not need to assume that the low-energy excitations can be described by a unitary modular tensor category. (Instead, we need \ref{as:A0'} and \ref{as:A1}.) In contrast, an oft-used numerical method~\cite{Zhang2012} requires global information of multiple ground states and the assumption that the system is described by a modular tensor category. Also, Haah~\cite{Haah2016} has attempted to extract the topological $S$-matrix from a single ground state. Haah's final argument that his invariant matches the $S$-matrix, however, still relies on the modular tensor category assumptions. 

Finally, we can also define a unitary string operator that creates an anyon pair. The support of the string can deform freely as long as the endpoints are fixed (Appendix~\ref{appendix:String}).

\section{Axiom extension, Information convex set and Isomorphism Theorem}\label{Sec:Isomorphism}
We have proclaimed in Section~\ref{sec:main_results} that we can, among many things, establish a \emph{globally well-defined} notion of topological charge. Because our axioms (Axiom \ref{as:A0'} and \ref{as:A1}) are assumed to hold only on bounded-radius disks, the fact that such a notion can even exist in the first place is not obvious at all. 

In order to explain how this works, we choose the main theme of this section to be ``local to global." Starting from our local axioms (Axiom \ref{as:A0'} and \ref{as:A1}), we will see how we can infer some of the global statements we made in Section~\ref{sec:main_results}. The first step in this direction lies in extending our local axioms to larger regions. Note that our axioms are assumed to hold only on (bounded-sized) $\mu$-disks. In Section \ref{sec:extaxiom}, we will show that the axioms in fact hold on arbitrarily large disks.

Then we move onto a concept that plays the central role in this paper: information convex set. In order to define an information convex set, we need to fix a subsystem and consider a set of states on this subsystem which are locally indistinguishable from the given reference state. Let us refer to this subsystem as $\Omega'$ and the set of states as $\tilde{\Sigma}(\Omega')$.

The set $\tilde{\Sigma}(\Omega')$ is \emph{not} the information convex set, but we are getting close. For every state in $\tilde{\Sigma}(\Omega')$, trace out part of the subsystem that lies at the ``outer edge" of $\Omega'$; see Fig.~\ref{Info_convex_definition}. This way, we get a set of reduced density matrix over a subset of $\Omega'$. Let this subset be $\Omega$. What we have obtained is the information convex set of $\Omega$, which we refer to as $\Sigma(\Omega)$. As one can see, this is quite an involved definition! However, in our opinion, we have a sensible justification. Why we choose $\Sigma(\Omega)$ over $\tilde{\Sigma}(\Omega')$ as a main actor of our story will be the main topic of discussion in Section~\ref{Sec.Info_convex_definition}.

In Section ~\ref{sec:elementary_steps_iso_theorem}, we derive the key technical result of this section: the \emph{isomorphism theorem} (Theorem~\ref{thm: the isomorphism theorem}). This theorem establishes an equivalence between information convex sets for topologically equivalent subsystems connected by a path. 

This section may seem a bit abstract at first, so it will be useful to have a concrete physical picture in mind. Consider a topologically ordered medium \cite{doi:10.1142/S0217979290000139} that can host anyons \cite{Leinaas1977,PhysRevLett.49.957}. It is well-known that, within such a medium, there is a \emph{globally} well-defined notion of superselection sectors and fusion space. This is because one can adiabatically transport anyons from different regions and compare them. For example, suppose we have two anyons that are well-separated from each other. How would we able to decide if they are the same anyon type or not? One can adiabatically bring either of the anyons to some fixed location and perform an Aharonov-Bohm type interference experiment. If the underlying anyon theory is unitary, there must be some experiment that can distinguish two different types of anyons.

In the above illustrative example, we observed that there is a physical process by which we can compare anyons or even a collection of anyons with each other. This comparison is possible if one can transport these objects from one place to another. The isomorphism theorem formalizes this transportation process in a Hamiltonian-independent manner.  

\subsection{Extension of the axioms}\label{sec:extaxiom}
In this Section, we show that our axioms (Axiom \ref{as:A0'} and \ref{as:A1}) hold on arbitrarily large disks. Results in this section have been discussed in Ref.~\cite{2014arXiv1405.0137K}. In order to understand this section, one should become familiar with the notion of \emph{quantum Markov state}~\cite{Petz1987,2004CMaPh.246..359H}. A quantum Markov state is a tripartite state, say over subsystems $A,B,$ and $C$, such that
\begin{equation*}
    I(A:C|B)=S(\rho_{AB}) + S(\rho_{BC}) - S(\rho_B) - S(\rho_{ABC})
\end{equation*}
is equal to $0$. These states enjoy several nontrivial properties, which we make extensive use of.

The following two lemmas will be used frequently.

\begin{lemma} \label{lemma_growth}
	Let $\rho_{ABC}$ and $\sigma_{ABC}$ be density matrices such that (1) $\rho_{AB}=\sigma_{AB}$ and $\rho_{BC}=\sigma_{BC}$; (2) $I(A:C\vert B)_{\rho}= I(A:C\vert B)_{\sigma}=0$. Then $\rho_{ABC} = \sigma_{ABC}$.
\end{lemma}
The proof follows from  Ref.~\cite{2003RvMaP..15...79P}. Suppose we know $\sigma_{AB}$ and $\sigma_{BC}$ for the partition in Fig.~\ref{Grow_a_disk_12_PDF}. Then \ref{as:A1} implies $I(A:C\vert B)_{\sigma}\le (S_{BC} + S_{CD} - S_B - S_D)_{\sigma} =0$. This lemma implies that the state $\sigma_{ABC}$ is uniquely determined by its reduced density matrices. Moreover, there is a CPTP map which can recover $\sigma_{ABC}$ from its reduced density matrices. This map is known as the Petz recovery map.
\begin{lemma}\label{lemma_Petz}
	(Petz recovery map~\cite{2003RvMaP..15...79P}) For any tripartite state $\rho_{ABC}$, $I(A:C|B)_\rho=0$ if and only if 
	\begin{align}
	\rho_{ABC}=\calE_{B\to BC}^\rho(\rho_{AB})\,,
	\end{align}
	where $\calE_{B\to BC}^\rho$ (the Petz recovery map) is defined as
	\begin{align}
	\calE_{B\to BC}^\rho(X_B)= \rho^{\frac{1}{2}}_{BC}\rho_B^{-\frac{1}{2}}X_B\rho_B^{-\frac{1}{2}}\rho_{BC}^{\frac{1}{2}}\,.\nonumber
	\end{align}
\end{lemma}

\begin{figure}[h]
	\centering
	\includegraphics[scale=1]{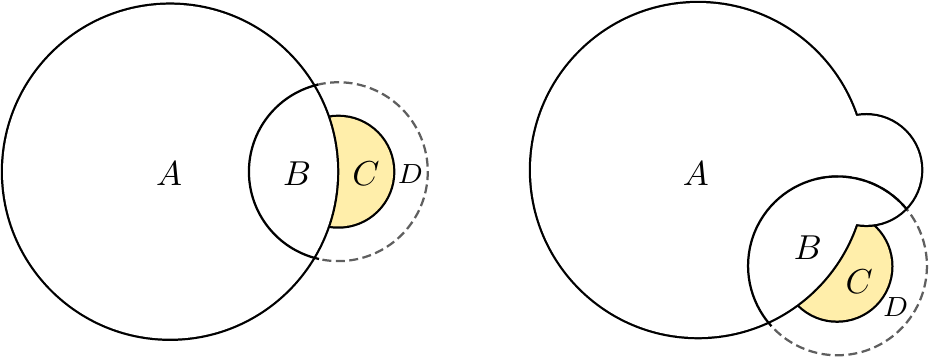}
	\caption{An illustration of the growth procedure of a disk from $AB$ to $ABC$. Here $A$ can be large and $BCD$ is contained in a $\mu$-disk in a manner similar to Fig.~\ref{A1}.}
	\label{Grow_a_disk_12_PDF}
\end{figure}

Now we are in a position to state the main result of this section. In Proposition~\ref{Prop:axiom_large}, we show that Axioms~\ref{as:A0'} and \ref{as:A1} hold for all disks. By assumption, these axioms hold for $\mu$-disks. The nontrivial part of the statement is that the axioms hold at a larger length scale.

\begin{figure}
	\centering
	\includegraphics[scale=1]{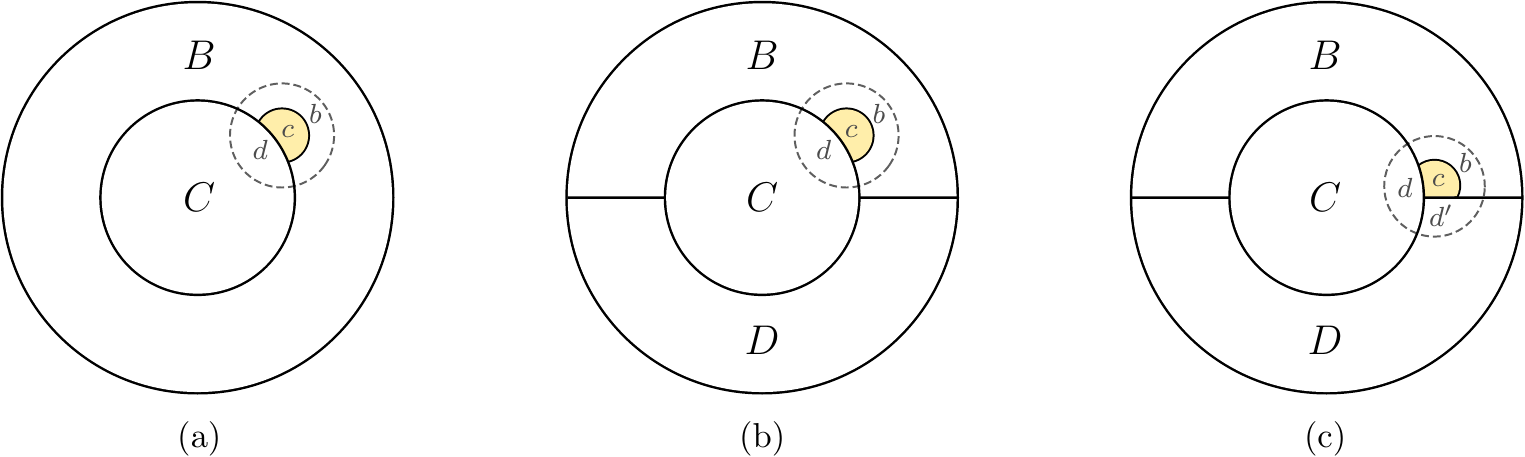}
	\caption{The extension of the axioms. A disk is divided into either $BC$ or $BCD$. A $\mu$-disk is on a smaller length scale, i.e., the small dashed circle surrounding the colored region. These figures represent three ways of enlarging $C$ by a small step. (a)  $bc\subseteq B$ and $d\subseteq C$; (b)  $bc\subseteq B$ and $d\subseteq C$; (c) $bc\subseteq B$, $d\subseteq C$ and $d'\subseteq D$.}\label{enlarge_axioms}
\end{figure}

\begin{Proposition}\label{Prop:axiom_large}
	For a reference state satisfying axioms \ref{as:A0'} and \ref{as:A1}, the entropic conditions Eq.~(\ref{eq:tau_A0}) and  Eq.~(\ref{eq:tau_A1}) are satisfied on all larger disk-like subsystems.
\end{Proposition}
\begin{proof}

We shall extensively use the following two inequalities:
\begin{eqnarray}
S_{BC}+S_{C}-S_{B}&\ge& S_{BB'C}+S_{C}-S_{BB'},\label{eq:enlarge_1}\\
	S_{BC}+S_{CD}-S_{B}-S_{D} &\ge& S_{BB'C}+S_{CDD'}-S_{BB'}-S_{DD'}.\label{eq:enlarge_2}
\end{eqnarray}
Both of them follow from SSA.

Let us first extend Axiom~\ref{as:A0'} to a larger scale. 
Without loss of generality, consider a disk and its subsystem $B$ and $C$, shown in Fig.~\ref{enlarge_axioms}(a). Assume that $BC$ is not contained in any $\mu$-disk. One can consider a sequence of (enlarged) subsystems to obtain this disk ($BC$) from a $\mu$-disk. 

For this purpose, it suffices to consider the following two moves. The first move is to enlarge $B$ while keeping $C$ fixed. The second move is to enlarge $C$ while keeping $BC$ fixed.  Our goal is to show that, for both of these steps, the linear combination of entropy $S_{BC} + S_C - S_B$ is non-increasing. The first move preserves Eq.~(\ref{eq:tau_A0}) because of Eq.~(\ref{eq:enlarge_1}) where we set $B\sqcup B'$ to be the enlarged $B$. To understand why the second move preserves Eq.~(\ref{eq:tau_A0}), consider a deformation depicted in  Fig.~\ref{enlarge_axioms}(a). We need to show $S_{BC} + S_C - S_B$ is non-increasing when we deform $C$ to include the colored region of Fig.~\ref{enlarge_axioms}(a). A variation of Eq.~(\ref{eq:enlarge_2}), which involves the $\mu$-disk $bcd$, is
\begin{equation}
S_{B} + S_{Cc}  - S_{B\backslash c} - S_{C} \le S_{bc} + S_{cd} - S_b - S_d =0. \label{eq:enlarge_3}
\end{equation}
Therefore, both moves preserve Eq.~(\ref{eq:tau_A0}). Because any disk can be enlarged from a $\mu$-disk by applying a sequence of these moves, Axiom~\ref{as:A0'} holds for any disk.

Now, let us move on to Axiom~\ref{as:A1}. Similarly, we can obtain a larger disk, say $BCD$, from a $\mu$-disk by making use of the following moves. As usual, the choice of the subsystems are similar to the one used in Axiom~\ref{as:A1}. The first move is to enlarge $B$ and $D$ while keeping $C$ fixed. The second move is to enlarge $C$ while keeping $BCD$ fixed.

Let us show that the entropic condition (\ref{eq:tau_A1}) holds at every step. 
The first move preserves Eq.~(\ref{eq:tau_A1}).  This is because inequality~(\ref{eq:enlarge_2}) can be applied to enlarged region $BB'$ and $DD'$. Notably, $B'$ and $D'$ can be chosen to touch each other. This allows us to show Eq.~(\ref{eq:tau_A1}) for the arbitrary  $BB'$ and $DD'$ satisfying the topology condition, by suitably choosing $B$ and $D$ in the initial $\mu$-disk. 
For the second move, we can consider the following sequence of small steps shown in Fig.~\ref{enlarge_axioms}(b)(c) as well as the same steps in which the choice of $B$ and $D$ are switched. 
Here is the justification of each small step. It is sufficient to justify the steps shown in Fig.~\ref{enlarge_axioms}(b) and (c). Regarding the enlargement process described in Fig.~\ref{enlarge_axioms}(b), we have
\begin{equation}
	  S_{B} + S_{CDc} - S_{B\backslash c} - S_{CD}\le S_{bc} + S_{cd} - S_{b} - S_{d}=0.  
\end{equation}
For the enlargement process in Fig.~\ref{enlarge_axioms}(c), we have
\begin{equation}
	  S_{B} + S_{CDc} - S_{B\backslash c} - S_{CD}\le S_{bc} + S_{cdd'} - S_{b} - S_{dd'}=0.  
\end{equation}

Therefore, both Eq.~(\ref{eq:tau_A0}) and Eq.~(\ref{eq:tau_A1}) hold at larger length scales. This completes the proof.

\end{proof}

\subsection{Information convex set}\label{Sec.Info_convex_definition}
Care must be taken in reading this section, for we are about to explain the most important concept in this paper: information convex set. Let us begin with some definitions. We say that two density matrices $\rho$ and $\rho'$
are \emph{consistent} with each other if they have identical density matrices on the overlapping support, i.e., $\rho_{A}=\rho'_A$ where $A$ is the intersection of the support of $\rho$ and that of $\rho'$, and denote it by    
\begin{equation*}
\rho \overset{c}{=} \rho'. 
\end{equation*}

For each subsystem $\Omega \subseteq V$, we can define its information convex set by first considering the ``thickening" of $\Omega$. Let $\partial \Omega \subseteq V$ be a set of vertices that are (graph) distance $1$ away from $\Omega$. If the subsystem $\Omega \sqcup \partial \Omega$ is topologically equivalent to $\Omega$, we refer to that subsystem as a thickening of $\Omega$. Let there exists a thickening of $\Omega$; we denote the thickening as $\Omega'$. Intuitively, $\Omega'$ is a subsystem that can be smoothly deformed into $\Omega$ such that $\Omega' \backslash \Omega$ is a boundary of $\Omega$ with a non-vanishing thickness. This thickness must be chosen to be sufficiently large compared to the correlation length of the underlying state.

Now, we are in a position to define the information convex set.
\begin{definition}  \label{def: info_convex}
Let $\Omega \subseteq V$ be a subsystem and let $\Omega'$ be the thickening of $\Omega$. The information convex set of $\Omega$ is defined as
	\begin{equation}
	\Sigma(\Omega)\equiv \{ \rho_{\Omega} \vert \rho_{\Omega}= \Tr_{\Omega'\backslash \Omega} \,\,\rho_{\Omega'}, \rho_{\Omega'}\in\tilde{\Sigma}(\Omega')  \}\,, \label{eq:info_convex}
	\end{equation}
	where $\tilde{\Sigma}(\Omega')$ is defined as
	\begin{equation*}
	\tilde{\Sigma}(\Omega') = \{\rho_{\Omega'}| \rho_{\Omega'} \overset{c}{=} \sigma_b, \quad \forall \sigma_b \in \mu \}.
	\end{equation*}
\end{definition}
This definition was first advocated in \cite{Kim2015sydney}. A related Hamiltonian-based definition was introduced in Ref.~\cite{2019PhRvB..99c5112S}. 

One may have thought that $\tilde{\Sigma}(\Omega)$ would be a more natural definition of information convex set. However, it turns out that the local indistinguishability constraints are insufficient to fix some of the degrees of freedom at the boundaries of $\Omega$. For this reason, it is possible for some of the elements in $\tilde{\Sigma}(\Omega)$ to possess entanglement between two distant $\mu$-disks in the vicinity of the boundary of $\Omega$. Moreover, even a multipartite correlation may arise along the boundary.
These exceptional elements cannot be regarded as the reduced state of a state on a larger support that is also locally indistinguishable from the reference state. In particular, they contain unnecessary extra information that has nothing to do with the anyon theory. The additional partial trace operation in the definition of $\Sigma(\Omega)$ removes such irrelevant correlations around the boundaries.

As it stands, the definition of information convex set is independent of our axioms (Axiom \ref{as:A0'} and \ref{as:A1}). However, once we impose these axioms, one can come up with a more restrictive definition of information convex set that does not involve $\Omega'$. This definition is formulated in Definition~\ref{def: sigma_hat}. Proposition~\ref{Prop: equivalent_info_convex} establishes the equivalence of the two.

\begin{figure}[h]
	\centering
\includegraphics[scale=1.15]{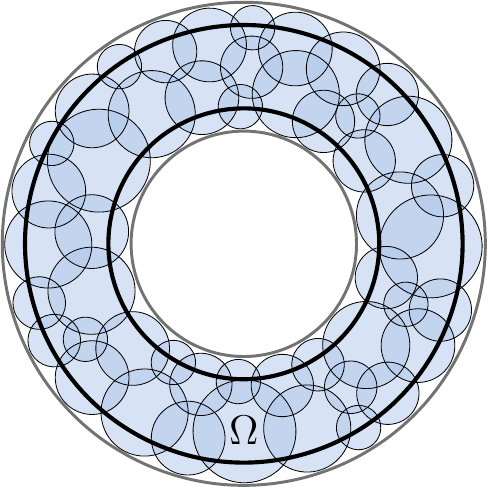}
	\caption{This figure is a schematic depiction of regions involved in the definition of information convex set $\Sigma(\Omega)$. Here $\Omega$ is the annulus between the black circles, annulus $\Omega'\supseteq \Omega$ is the thickening of $\Omega$, i.e., the region between gray circles. Any element in ${\tilde \Sigma}(\Omega')$ is consistent with the reference state $\sigma$ on every $\mu$-disk contained in $\Omega'$ (the blue disks). We chose $\Omega$ to be an annulus for illustration purposes. Other topologies are allowed as well.}
	\label{Info_convex_definition}
\end{figure}

Below, we derive properties of the information convex set. First of all, information convex set is a convex subset of the state space. This follows straightforwardly from the definition.
\begin{Proposition}\label{Prop.Info_convex_basic}
	For any nonempty $\Omega \subseteq V$, $\Sigma(\Omega)$ is a nonempty finite-dimensional compact convex set. Furthermore, if $\Omega\subseteq \Omega'$ and $\rho_{\Omega'}\in \Sigma(\Omega')$, then $\Tr_{\Omega'\backslash\Omega}\,\rho_{\Omega'}\in \Sigma(\Omega)$.
\end{Proposition} 

\begin{proof}
	The state space of a finite-dimensional Hilbert space is a finite-dimensional compact convex set. The partial trace operation $\Tr_{\Omega'\backslash\Omega}$ is linear and bounded. Therefore, the image of the partial trace is compact and convex. The last statement is a direct consequence of Definition~\ref{def: info_convex}.
\end{proof}

Secondly, we show that the information convex set of a disk contains a single element. We use this result throughout this paper, primarily for identifying the uniqueness of the global state on a sphere and for defining the vacuum sector.\footnote{As a side note, let us point out that this result is based only on \ref{as:A1}. Therefore, the uniqueness of the element for the information convex set of a disk can hold more generally, even when \ref{as:A0'} breaks down.}

\begin{Proposition} \label{Prop: structure_1}
	For any disk-like subsystem $\omega$, we have
	\begin{equation}
	\Sigma(\omega) = \{ \sigma_{\omega} \}, \label{eq: sigma_omega}
	\end{equation}
	where $\sigma_{\omega}\equiv \Tr_{\bar{\omega}} \vert \psi\rangle \langle \psi\vert$ and ${\bar{\omega}}$ is the complement of $\omega$.
\end{Proposition}
\begin{proof}
	The idea is illustrated in Fig.~\ref{Cover_a_disk}, which involves a repetition of the ``growth process" described in Fig.~\ref{Grow_a_disk_12_PDF}. Starting from a $\mu$-disk, one can grow the disk until the enlarged disk covers $\omega$. 
	
	Let us get into the details. First, recall that $\Sigma(\omega)$ is nonempty since it contains $\sigma_{\omega}$. Let us pick another element $\sigma'_{\omega}\in\Sigma(\omega)$. According to Definition~\ref{def: info_convex}, $\sigma_{\omega}$ and $\sigma'_{\omega}$ must be identical on a $\mu$-disk. Therefore, we can repeatedly use the conditional independence condition, i.e., $I(A:C\vert B)_{\sigma}=I(A:C\vert B)_{\sigma'}=0$, to show that the reduced density matrices of the two states are identical on increasingly larger disks; see  Lemma~\ref{lemma_growth}. The disk can grow to the point it covers $\omega$. Thus, we can conclude that $\sigma'_{\omega}= \sigma_{\omega}$.
\end{proof}

\begin{figure}[h]
	\centering
\includegraphics[scale=1.1]{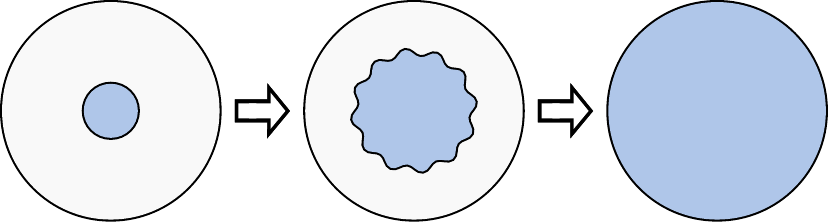}
	\caption{A $\mu$-disk can grow until it covers a larger disk $\omega$.}
	\label{Cover_a_disk}
\end{figure}

More generally, we can consider the information convex set of a subsystem which is \emph{not} topologically a disk. We shall discuss the structure of such set in Section \ref{Sec.Fusion_data}. The following fact will be important for that discussion. Let $\Omega$ be an arbitrary subsystem. We show that, for any element in the information convex set of $\Omega$, its reduced density matrix on a disk-like region is equal to the reduced density matrix of the reference state on the same region.

\begin{Proposition} \label{Prop:disk_in_Omega}
	Any state $\rho_{\Omega}\in\Sigma(\Omega)$ satisfies
	\begin{equation}
	\Tr_{\Omega\backslash\omega}\, \rho_\Omega=\sigma_\omega
	\end{equation}
	on any disk-like subsystem $\omega\subseteq \Omega$.
\end{Proposition}
\begin{proof}
	By Proposition~\ref{Prop.Info_convex_basic}, for any state $\rho_\Omega\in\Sigma(\Omega)$, $\Tr_{\Omega\backslash\omega} \rho_{\Omega}$ is an element of $\Sigma(\omega)$. By Proposition~\ref{Prop: structure_1}, the set $\Sigma(\omega)$ contains only one element.  
\end{proof}

As a side note, we can always choose the reference state to be pure even if the initially given reference state is not. This is because the information convex set of the entire system is the state space of a finite-dimensional Hilbert space (see Theorem~\ref{Prop: Sigma(M)}); one can simply choose one of the pure states in this space to be the reference state. This pure reference state, restricted to the $\mu$-disks, would be consistent with the initially given reference state.

An interesting special case is when the underlying manifold is a sphere ($S^2$). In this case, the global state is unique and thus pure. 
\begin{Proposition}\label{Prop: global pure}
	$\Sigma(S^2)=\{ \vert \psi\rangle \langle \psi\vert \}$.
\end{Proposition}
\begin{proof}
Let us begin by setting up an appropriate set of subsystems. Partition the sphere into three subsystems, $A,B,$ and $C$. We choose $C$ to be a disk and $B$ to be an annulus that surrounds $C$; see Fig.~\ref{A0}. The complement of $BC$, i.e., $A$, is a disk.

Without loss of generality, suppose we have two states $\rho_{ABC},\sigma_{ABC} \in \Sigma(S^2)$. We show that they must be equal. By Proposition~\ref{Prop: structure_1}, their reduced density matrices are identical on $AB$ and $BC$. SSA implies that for both states,
\begin{equation*}
    I(A:C|B) \leq S(BC) + S(C)- S(B),
\end{equation*}
where we suppressed the dependence on $\rho$ and $\sigma$. Either way, the right hand side is $0$ because our axioms hold at any scale; see Proposition~\ref{Prop:axiom_large}. Because $\rho_{ABC}$ and $\sigma_{ABC}$ both have vanishing conditional mutual information (conditioned on $B$) and have identical density matrices (on $AB$ and $BC$), we can use Lemma~\ref{lemma_growth}. Therefore, $\rho_{ABC} = \sigma_{ABC}$. 

By Theorem~\ref{Prop: Sigma(M)}, $\Sigma(S^2)$ is isomorphic to a state space of a finite-dimensional Hilbert space. Because this set has a unique element, the global reference state must be pure.
\end{proof}

\subsection{Elementary steps and isomorphism theorem \label{sec:elementary_steps_iso_theorem}}
In this section, we establish an isomorphism between two information convex sets. This isomorphism exists if the subsystems associated with the two sets are topologically equivalent and can be smoothly deformed into each other. More precisely, in order to establish an isomorphism, we require that the two subsystems to be connected by a \emph{path} (Defintion~\ref{def:path}).

Crucial to this analysis is the concept of state merging. Suppose we have two quantum states $\rho$ and $\sigma$ which share an overlapping support and consistent. The question is whether one can consistently ``sew" them together. Namely, can we find a state which is consistent with both $\rho$ and $\sigma$? This is known as the quantum marginal problem. In general, even deciding whether there is such a state or not is known to be extremely difficult~\cite{Liu2006}. There are several nontrivial necessary conditions~\cite{Kim2013,2014arXiv1405.0137K}, but sufficient conditions are rare.

However, one of us has found a nontrivial sufficient condition~\cite{2016PhRvA..93b2317K}. We restate the result for the reader's convenience. 
\begin{lemma}\label{Prop: merging technique}
	(Merging Lemma~\cite{2016PhRvA..93b2317K})
	Given a set of density matrices $\calS\equiv \{ \rho_{ABC} \}$ and a density matrix $\sigma_{BCD}$ such that $\rho_{BC}=\sigma_{BC}$ and 
	\begin{equation}
	I(A:C\vert B)_{\rho}= I(B:D\vert C)_{\sigma} =0, \quad \forall\rho\in\calS\,,\label{eq: rho sigma}
	\end{equation}
	there exists a unique set of ``merged" states $\{ \tau^\rho_{ABCD}=\calE^\sigma_{C\to CD}(\rho_{ABC}) \}$ which satisfy the following properties. 
	\begin{enumerate}
		\item [(1)]	 $\tau^\rho$ is consistent with $\rho$ and $\sigma$, i.e.
		\begin{equation}
		\tau^\rho_{ABC}=\rho_{ABC} \quad\textrm{and} \quad \tau^\rho_{BCD}= \sigma_{BCD}. \label{eq:marginal}
		\end{equation}
		
		\item [(2)] Vanishing conditional mutual information,
		\begin{equation}\label{eq:marginal2}
		I(A:CD\vert B)_{\tau^\rho}= I(AB:D \vert C)_{\tau^\rho} = 0, \,\,\,\,\forall\,\rho.
		\end{equation}
		\item [(3)] The conservation of von Neumann entropy difference, for arbitrary $\rho,\rho'\in\calS$, 
		\begin{equation}\label{eq:entropypreserve}
		S(\tau^{\rho}_{ABCD})- S(\tau^{\rho'}_{ABCD}) = S(\rho_{ABC})- S(\rho'_{ABC}).
		\end{equation}
	\end{enumerate}
\end{lemma}
The significance of this lemma lies on the fact that one can guarantee the existence of a global state from a (relatively) local information. What is given to us are the density matrices over $ABC$ and $BCD$, together with the conditions that can be verified on $ABC$ and $BCD.$ In particular, these conditions can be directly verified from the given states. Once the conditions are verified, one can guarantee the existence of a state over $ABCD$ that is consistent with the given density matrices.

The merging lemma (Lemma~\ref{Prop: merging technique}), together with our axioms (Axiom \ref{as:A0'} and \ref{as:A1}) underpin the majority of our technical work. The interplay between the two is what allows us to start from strictly local information and conclude something nontrivial at a larger scale. Roughly speaking, such analysis is carried out as follows. Our axioms allow us to upper bound certain conditional mutual information by $0$. We can then apply this fact to Lemma~\ref{Prop: merging technique} repeatedly to merge (many) density matrices. In particular, we can merge elements of multiple information convex sets into an element of yet another information convex set. This not only allows us to smoothly deform the boundary of a subsystem (Fig.~\ref{fig:elementary step}) but also allows us to consider merging processes with nontrivial topology changes; see Sec.~\ref{Sec. Fusion}.

Now, we are in a position to prove the \emph{isomorphism theorem}. This theorem establishes an equivalence between two information convex sets whose underlying subsystems can be smoothly deformed into each other. We first explain a method to establish the equivalence when one subsystem is merely an ``infinitesimal" deformation of the other. Of course, the word infinitesimal should not be taken literally, because we are considering a quantum many-body system on a lattice. What we mean is that one subsystem can be obtained from the other by either attaching or removing a region whose size is comparable to that of the $\mu$-disks.

Imagine zooming into the region in which this deformation occurs. Without loss of generality, we can consider two subsystems $\Omega=ABC$ and $\Omega'=ABCD$ depicted in Fig.~\ref{fig:elementary step}, where $CD$ is contained in a $\mu$-disk. We can show that there exists a bijection between $\Sigma(\Omega)$ and $\Sigma(\Omega')$. This is the content of Proposition~\ref{Prop: iso_ABCD}. The proof is in Appendix~\ref{sec:facts_isomorphism_theorem}.

\begin{figure}[h]
	\centering
\includegraphics[scale=0.9]{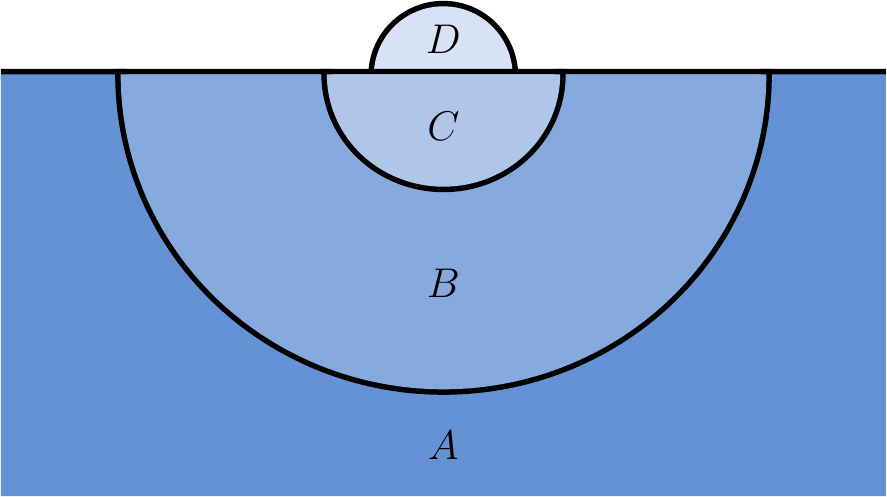}
	\caption{Two slightly different subsystems $\Omega=ABC$ and $\Omega'=ABCD$ (only a part of $A$ is shown).  $BCD$ is a disk and $CD$ is contained in a $\mu$-disk.  The topology of $A$ can be arbitrary. We require $B$ to be thick enough, so that $A$ and $D$ are separated by at least $2r+1$. (Recall that $r$ is the scale used in defining the $\mu$-disks; see Eq.~(\ref{eq:mu(r)}).) }\label{fig:elementary step}
\end{figure}

\begin{restatable}[]{Proposition}{isoabcd}\label{Prop: iso_ABCD}
Consider the partition in Fig.~\ref{fig:elementary step}. (Note:  $A$ and $D$ are assumed to be separated by at least $2r+1$ in Fig.~\ref{fig:elementary step}.) Let the domain of $\Tr_D$ and $\mathcal{E}_{C\to CD}^{\sigma}$ to be $\Sigma(\Omega')$ and $\Sigma(\Omega)$ respectively.
\begin{equation}
    \begin{aligned}
    \Ima \Tr_D &\subseteq \Sigma(\Omega) \\
    \Ima \mathcal{E}_{C\to CD}^{\sigma} &\subseteq \Sigma(\Omega').
    \end{aligned}\label{eq:Tr_and_e}
\end{equation}
Moreover, for all $\rho_{\Omega}\in \Sigma(\Omega)$ and $\rho_{\Omega'} \in \Sigma(\Omega')$,
\begin{eqnarray}
	\Tr_D \circ \calE^\sigma_{C\to CD} (\rho_{\Omega})&=&\rho_{\Omega},  \label{eq:iso_1}\\
	\calE^\sigma_{C\to CD}  \circ\Tr_D(\rho_{\Omega'})&=&\rho_{\Omega'}. \label{eq:iso_2}
\end{eqnarray}
In particular, $\Tr_D$ and $\calE^\sigma_{C\to CD}$ are bijections.
\end{restatable}

As a side note, Proposition~\ref{Prop: iso_ABCD} further implies that the isomorphism $\calE^\sigma_{C\to CD}:\Sigma(\Omega)\to \Sigma(\Omega')$ is independent of the choice of $B$ and $C$. To see why, consider two choices, say $BC\subseteq \Omega$ and $B'C' \subseteq \Omega$. Consider two maps $\calE^\sigma_{C\to CD}$ and $\calE^\sigma_{C'\to C'D}$. Suppose there exists an element of $\Sigma(\Omega)$ which, under these two maps, is mapped into two different elements of $\Sigma(\Omega')$. Upon applying $\Tr_D$ to these two states, by Proposition~\ref{Prop: iso_ABCD}, both states must be mapped back to the same state. This implies that $\Tr_D$ is not injective, which contradicts Proposition~\ref{Prop: iso_ABCD}.

There are two simple corollaries of Proposition~\ref{Prop: iso_ABCD}, which will become handy in the rest of the paper. First, the bijective map $\Tr_D$ and $\calE_{C\to CD}^{\sigma}$ preserves the distance between two states.

\begin{corollary}
\label{coro:distance_preservation}
(Distance preservation)
Let $\rho_{\Omega}, \rho_{\Omega}' \in \Sigma(\Omega)$ and $\rho_{\Omega'}, \rho_{\Omega'}' \in \Sigma(\Omega')$. For any distance measure $D(\cdot, \cdot)$ between quantum states, 
    \begin{equation}
        D(\rho_{\Omega}, \rho_{\Omega}') = D(\calE_{C\to CD}^{\sigma}(\rho_{\Omega}), \calE_{C\to CD}^{\sigma}(\rho_{\Omega}')) 
    \end{equation}
    and
    \begin{equation}
        D(\rho_{\Omega'}, \rho_{\Omega'}') = D(\Tr_D(\rho_{\Omega'}), \Tr_D(\rho_{\Omega'}')).
    \end{equation}
\end{corollary}
\begin{proof}
    For both identities, the proof is practically identical. So we only discuss the proof of the first identity. Because both $\calE_{C\to CD}^{\sigma}$ and $\Tr_D$ are CPTP maps, distance is nonincreasing under these maps. Therefore, 
    \begin{equation}
    \begin{aligned}
        D(\rho_{\Omega}, \rho_{\Omega}') &\geq D(\calE_{C\to CD}^{\sigma}(\rho_{\Omega}), \calE_{C\to CD}^{\sigma}(\rho_{\Omega}')) \\
        &\geq D(\Tr_D\circ\calE_{C\to CD}^{\sigma}(\rho_{\Omega}), \Tr_D\circ\calE_{C\to CD}^{\sigma}(\rho_{\Omega}')) \\
        &= D(\rho_{\Omega}, \rho_{\Omega}'),
    \end{aligned}
    \end{equation}
    where in the last line we used Proposition~\ref{Prop: iso_ABCD}. Therefore, $D(\rho_{\Omega}, \rho_{\Omega}') = D(\calE_{C\to CD}^{\sigma}(\rho_{\Omega}), \calE_{C\to CD}^{\sigma}(\rho_{\Omega}'))$.
\end{proof}

We note that, while we only considered distances between two quantum states, the same proof applies to the preservation of the fidelity $F(\rho,\tau)=(\Tr\sqrt{\sqrt{\rho}\,\tau\sqrt{\rho}})^2$. While fidelity is not a distance measure, its behavior is monotonic under application of CPTP maps. Therefore, the proof of Corollary~\ref{coro:distance_preservation} still applies.

In fact, we can show more. Even the entropy difference is preserved under $\Tr_D$ and $\calE_{C\to CD}^{\sigma}$. The proof of this statement is a simple byproduct of the proof of Proposition~\ref{Prop: iso_ABCD}. The proof follows immediately from property (3) of Lemma~\ref{Prop: merging technique}.
\begin{corollary}
(Entropy difference preservation)
\label{coro:entropy_difference}
    Let $\rho_{\Omega}, \rho_{\Omega}' \in \Sigma(\Omega)$ and $\rho_{\Omega'}, \rho_{\Omega'}' \in \Sigma(\Omega')$. The von Neumann entropies satisfy
    \begin{equation}
        S(\rho_{\Omega}) -S(\rho_{\Omega}') = S(\calE_{C\to CD}^{\sigma}(\rho_{\Omega})) - S( \calE_{C\to CD}^{\sigma}(\rho_{\Omega}'))
    \end{equation}
    and
    \begin{equation}
        S(\rho_{\Omega'}) - S(\rho_{\Omega'}') = S(\Tr_D(\rho_{\Omega'})) -  S(\Tr_D(\rho_{\Omega'}')).
    \end{equation}
\end{corollary}

Therefore, given a subsystem $\Omega$ and its information convex set $\Sigma(\Omega)$ we can establish a bijection between $\Sigma(\Omega)$ and $\Sigma(\Omega')$ where $\Omega'$ is a slight deformation of $\Omega$. In order to apply Proposition~\ref{Prop: iso_ABCD}, $\Omega'$ must be topologically equivalent to $\Omega$. We refer to the process of subtracting/adding a disk-like region to a given subsystem as the \emph{elementary step} of the deformation. 

 The isomorphism between two information convex sets can be established by repeating these elementary steps. However, we have to be careful on two points. First, for two given topologically equivalent subsystems, there can be more than one way to deform one to the other. Second, even if the underlying subsystems are topologically equivalent, there may not be a smooth deformation between the two. As a trivial example, suppose we have two spheres. We can place two subsystems on each of these spheres. Even if these subsystems are topologically equivalent to each other, there is no sequence of subsystems that smoothly deforms one to the other. Even on a connected space, one cannot make such a statement; see Fig.~\ref{Annulus_X0_X1}.   

Therefore, these (potentially different) isomorphisms must be labeled by their paths. Let us formalize this notion below.
\begin{definition} (Path)\label{def:path}
A finite sequence of subsystems $\{\Omega^t\}$ with  $t={i/N}$ and $i=0,1,2,\cdots, N$, ($N$ is a positive integer), is a \emph{path} connecting $\Omega^0$ and $\Omega^1$ if each pair of nearby subsystems in the sequence are related by an elementary step of deformation, illustrated in Fig.~\ref{fig:elementary step}.
\end{definition}

Because a path is built up from elementary steps, we obtain the following theorem.
\begin{theorem} (Isomorphism Theorem) \label{thm: the isomorphism theorem}
	If $\Omega^0$ and $\Omega^1$ are connected by a path $\{ \Omega^t \}$, then there is an isomorphism 
	\begin{equation}
	\Phi_{\{\Omega^t\}}:\quad \Sigma(\Omega^0)\to \Sigma(\Omega^1)
	\end{equation}
	uniquely determined by the path $\{\Omega^t\}$. Moreover, it preserves the distance and the entropy difference between elements
	\begin{align}
	D(\rho, \sigma) &= D\left(\Phi_{\{\Omega^t\}}(\rho), \Phi_{\{\Omega^t\}}(\sigma)\right)\label{eq:preserve_F} \\
	S(\rho) - S(\sigma) &= S\left(\Phi_{\{\Omega^t\}}(\rho)\right)-S\left(\Phi_{\{\Omega^t\}}(\sigma)\right),\label{eq:preserve_S}
	\end{align}	
	where $D(\cdot,\cdot)$ is any distance measure which is non-increasing under CPTP-maps. 
\end{theorem}

We omit the proof since it straightforwardly follows by applying Proposition~\ref{Prop: iso_ABCD} repeatedly. For any path $\{\Omega^t\}$, we can define an inverted path $\{\Omega^{1-t} \}$ which reverses the sequence of subsystems. This leads to the inverse isomorphism  $\Phi_{\{\Omega^{1-t}\}}:\,\, \Sigma(\Omega^1)\to \Sigma(\Omega^0)$.

Generally speaking, different paths may give rise to different isomorphisms. That is, under two different isomorphisms, an element of the information convex set may be mapped to two distinct elements. However, sometimes, we merely need the existence of an isomorphism. In that case, we will use a notation
\begin{equation}
\Sigma(\Omega^0)\cong \Sigma(\Omega^1) \nonumber
\end{equation}
to indicate the existence of such an isomorphism. Under this condition, any distance measure and entropy difference is preserved; see Eqs.~(\ref{eq:preserve_F}) and (\ref{eq:preserve_S}).

\section{Fusion data from information convex sets}\label{Sec.Fusion_data}

The isomorphism theorem (Theorem~\ref{thm: the isomorphism theorem}) guarantees that the structure of the information convex set only depends on the topology, as long as the underlying subsystems can be smoothly deformed from one to another along some path. 

We now focus on how to extract the information of the topological charges and the corresponding fusion rules from the information convex set. We do this by studying how the geometry of the information convex set depends on the topology of the underlying subsystem. We then use the merging technique (Lemma~\ref{Prop: merging technique}) to relate subsystems with different topologies and obtain several consistency equations. We then define the fusion rules and show that they satisfy all the constraints expected from the known algebraic theory of anyon~\cite{2006AnPhy.321....2K}. The result of this study is summarized in Table~\ref{tab:higher}.

\begin{table}[h]
	\centering
	\begin{tabular}{|c|c|}
		\hline
		Physical data  & Number of holes       \\
		\hline
		Superselection sectors   &    $1$      \\
		\hline
		Fusion multiplicities & $2$   \\
		\hline
		Axioms of the fusion theory & $1,2,3,4$ (merging) \\
		\hline
	\end{tabular}
	\caption{Physical data that can be extracted from disks with different number of holes.}
	\label{tab:higher}
\end{table}

\subsection{Superselection sectors/topological charges} \label{Sec. 2 punctures}
Let us define a notion of superselection sectors, which is one of the key ingredients of the algebraic theory of anyon \cite{2006AnPhy.321....2K}. Historically, the notion of superselection sectors was introduced in the context of local field theory; see~\cite{Fredenhagen1989,Haag1996}.
In the context of topologically ordered systems which is most relevant to our discussion, a nontrivial superselection sector corresponds to an anyon type that cannot be created by any local operator. 

There are several recent attempts to rigorously formulate the superselection sectors based on operator algebra assumptions. One approach is based on the operator algebra on an annulus~\cite{Haah2016,2018arXiv181002376K}, and another approach is based on the operator algebra on a cone-like subsystem of an infinite lattice~\cite{naaijkens2017quantum,2019CMaPh.373..219C}. Our approach to characterize the superselection sectors is similar to the one based on the operator algebra on annuli. However, these two approaches differ in their assumptions and their range of validity.

We will identify a well-defined information-theoretic object and find that this object coincides with the conventional notion of superselection sectors in anyon theory. Importantly, we find that the information convex set of an annulus forms a simplex (Theorem~\ref{Prop: structure_2}).  The simplex has a finite number of extreme points. Moreover, these extreme points are orthogonal to each other. See Fig.~\ref{Annulus_LMR}(c) for an illustration.  (See Appendix~\ref{app:convexsets} for the general definition of extreme points.) We will define these extreme points as the superselection sectors. 

\begin{restatable}[]{theorem}{simplex}
(Simplex Theorem) \label{Prop: structure_2}
	For an annulus $X$, the information convex set is the convex hull of a finite number of orthogonal extreme points, $\{ \sigma^a_{X} \}$, i.e. 
	\begin{equation}\label{Eq:directsumconv}
	\Sigma(X)=\left\{ \rho_{X} \left\vert \rho_{X}=\bigoplus_a p_a \sigma^a_{X}  \right.\right\},
	\end{equation}
	where $\{ a \}$ is a finite set of labels and $\{ p_a \}$ is a probability distribution. 
\end{restatable}

\begin{figure}[htbp]
	\centering
	\includegraphics[scale=1.10]{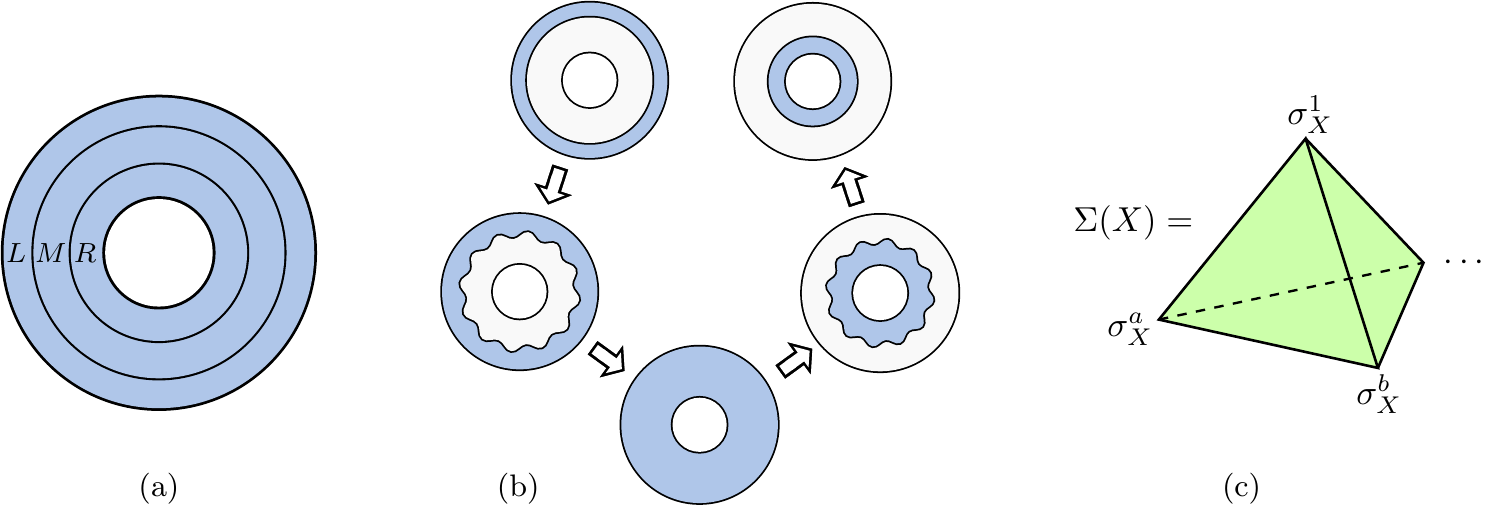}
	\caption{(a) Division of an annulus $X$ into three thinner annuli $L$, $M$, $R$. (b) A path (extensions-extensions-restrictions-restrictions) which generates an isomorphism $\Sigma(L)\cong \Sigma(X)\cong \Sigma(R)$. (c) A schematic depiction of the simplex structure of $\Sigma(X)$. The extreme points are the ``corners" of the simplex. If the annulus $X$ is contained in a disk, then one of the extreme points has the vacuum label ``1".}
	\label{Annulus_LMR}
\end{figure}
Here we show a sketch of the proof of Theorem~\ref{Prop: structure_2}. (See Appendix~\ref{Proof_of_Prop: structure_2} for the full proof). The orthogonality follows from the factorization property of the fidelity $F(\rho,\tau)=(\Tr\sqrt{\sqrt{\rho}\,\tau\sqrt{\rho}})^2$. Let $F_X$ be the fidelity of two extreme points in the information convex set of $X=LMR$ in Fig.~\ref{Annulus_LMR}(a). (We use the same convention for subsystems.) By using the fact that any extreme point has $I(L:R)=0$ (Corollary~\ref{coro:product_X}),  we find
\begin{equation}
F_{LR}= F_L F_R.
\end{equation}
Because the fidelity is non-decreasing under a partial trace, we have $F_{LMR} \leq F_{LR}$. Since $X$ and $L,R$ are annuli connected by paths, see Fig.~\ref{Annulus_LMR}(b), the isomorphism theorem implies $F=F_{L}=F_R=F_{LMR}$ and thus 
\begin{equation}
F \leq F^2. 
\end{equation}
$F\in [0,1]$. Therefore, the two extreme points are either the same ($F=1$) or orthogonal ($F=0$).  This derivation also shows that we can copy the information of the extreme point to $L$ and $R$ simultaneously. 
The finiteness of the label set follows from the orthogonality and the setup that the Hilbert space is finite-dimensional.

Theorem~\ref{Prop: structure_2} implies that $\Sigma(X)$ forms a simplex in the state space. It has a finite number of extreme points $\{\sigma^a_{X} \}$, which can be perfectly distinguished from each other by a projective measurement supported on the annulus. The simplex structure also implies that its elements can only store classical information in the probability distribution $\{p_a\}$. The isomorphism theorem~\ref{thm: the isomorphism theorem} guarantees the universality of the label set, i.e., the fact that the same set of labels applies to all annuli, which could be connected to each other by a path. Note that there could be annuli not connected by any path, e.g., the $X^0$ and $X^1$ in Fig.~\ref{Annulus_X0_X1}. Theorem~\ref{Prop: structure_2} is still applicable for both annuli, but the label sets for them can be different. This is related to the existence of topological defects~\cite{Bombin2010}.

\begin{figure}[h]
	\centering
\includegraphics[scale=1.8]{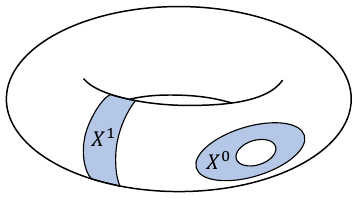}
	\caption{A pair of annuli $X^0$ and $X^1$ on a torus. They cannot be connected by any smooth path because $X^0$ is contractible and $X^1$ is non-contractible. }
	\label{Annulus_X0_X1}
\end{figure}

One of the extreme points is special. Consider a contractible annulus (see $X^0$ in Fig.~\ref{Annulus_X0_X1} for example). The information convex set of such annulus has a special extreme point which we label as ``$1$"; physically, this corresponds to the vacuum sector.
\begin{restatable}[]{Proposition}{propextreme}
 \label{Prop. 1 is extreme}
	Let $\omega$ be a disk. For any annulus $X \subseteq \omega$ 
	\begin{equation}
	\sigma^1_{X} \equiv \Tr_{\omega\backslash X} \, \sigma_{\omega},
	\end{equation}
	is an extreme point of $\Sigma(X)$. 
\end{restatable}
See Appendix~\ref{sec:orthogonality} for the proof. Importantly, the notion of vacuum sector is unambiguous because $\Sigma(\omega)$ has a unique element if $\omega$ is a disk; see Proposition~\ref{Prop: structure_1}.

Now we are ready to define the superselection sectors in our framework. When there is a pair of anyons, the topological charge can be measured by an Aharonov-Bohm type interferometry measurement by braiding the anyons~\cite{BONDERSON20082709}. Indeed, the projective measurement used for distinguishing different $\sigma^a_X$ corresponds to this interferometry measurement for several exactly solvable models. Based on this observation, we identify each label of the extreme points as a superselection sector of the system. 
\begin{definition}
    Let $X$ be a contractible annulus. The set of superselection sectors is a set of extreme points in $\Sigma(X)$. 
\end{definition}
Except for the vacuum sector, we label each extreme points with the lower-case Roman letters:
\begin{equation}
	\calC=\{1, a,b,c,\cdots\}\,.
\end{equation}

Several authors have already made attempts to define superselection sectors in 2D gapped phases. A statement analogous to the simplex theorem was obtained recently in \cite{2018arXiv181002376K} for models with a local commuting parent Hamiltonian. This proof is based on the operator-algebraic framework of Haah~\cite{Haah2016}.

We expect our derivation to hold more generally, because we make no assumption about the parent Hamiltonian. If the area law (Eq.(\ref{eq:area_law})) holds, our results follow. In particular, if we can prove approximate versions of our statements, we may be able to rigorously define a notion of superselection sectors for models with non-zero Hall conductance or non-zero chiral central charge. These models cannot have a commuting projector parent Hamiltonian~\cite{2019CMaPh.tmp..257K,2019arXiv190506488K}.

The isomorphism theorem guarantees the label set of the superselection sector to be independent of the details of the annulus. However, this theorem in itself does not imply there is a well-defined way to compare the topological charges for two annuli. Indeed, in the presence of a topological defect, transporting the same superselection in two different ways may result in different sector~\cite{Bombin2010}. 

In order to show that the isomorphisms associated with different paths are identical, one necessarily has to invoke an extra condition on the paths. Lemma~\ref{prop:topological_deformation} establishes one such condition. Roughly speaking, it is possible to compare two annuli unambiguously independent of the path, as long as both paths lie in a single disk.

\begin{restatable}[]{lemma}{deformation}
 \label{prop:topological_deformation}
	Let $X^0$ and $X^1$ be two annuli contained in a disk $C$; see Fig.~\ref{Topo_deformation} for example. Let $\{ X^t_{(1)} \}$ and $\{ X^t_{(2)} \}$ be two paths connecting $X^0$ and $X^1$ such that $X^0_{(i)} = X^0$, $X^1_{(i)} = X^1$ for $i=1,2$. Moreover, assume that  
	$
	\cup_t X_{(1)}^t \subseteq C$, $ \cup_t X_{(2)}^t \subseteq C.
	$
	Then, the isomorphisms
	\begin{eqnarray}
	&& \Phi_{\{ X^t_{(1)} \}}: \Sigma(X^0)\to \Sigma(X^1)  \nonumber \\
	\textrm{and}&& \Phi_{\{ X^t_{(2)} \}}: \Sigma(X^0)\to \Sigma(X^1) \nonumber
	\end{eqnarray}
	are identical.
\end{restatable}

\begin{figure}[h]
	\centering
\includegraphics[scale=1.2]{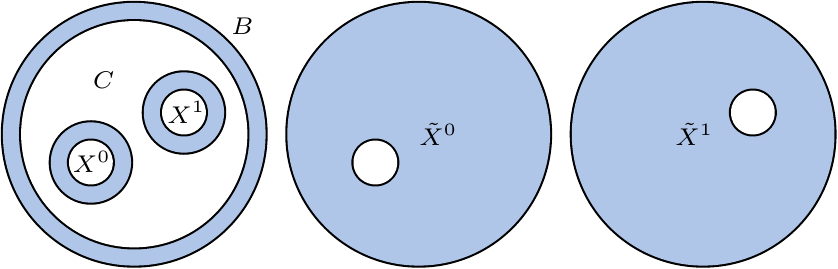}
	\caption{Both $C$ and $BC$ are disks. $B$ is an annulus. $X^0$, $X^1$, $\tilde{X}^0$ and $\tilde{X}^{1}$ are annuli. Note that $X^0$ and $X^1$ are subsets of $C$. In the proof of Lemma~\ref{prop:topological_deformation}, we construct an extension $X^t\to \tilde{X}^t$ with $B\subseteq \tilde{X}^t$. }
	\label{Topo_deformation}
\end{figure}
Lemma~\ref{prop:topological_deformation} implies that we can always treat $\Phi_{\{X^t_{(1)}\}}(\sigma_{X^0}^a)$ as the label $a$ for $X^1$. The key idea behind the proof is that one can copy the information about which superselection sector lies inside an annulus to a common annulus; see Fig.~\ref{Topo_deformation} for illustration. The proof is left in Appendix~\ref{sec:orthogonality}.

Now that we have defined a notion of superselection sectors, we can define their quantum dimensions. We will do so by investigating a contribution to the entanglement entropy that depends on the choice of this sector. We will use the following definition. Later, we will be able to determine their value.
\begin{definition}[Entropy contribution of superselection sector]\label{def:f(a)}
	For a contractible annulus $X$, we define the universal contribution to von Neumann entropy from superselection sector $a$ as 
	\begin{equation}
	f(a)\equiv \frac{S(\sigma^a_{X}) - S(\sigma^1_{X})}{2}.
	\label{eq: entanglement of extreme}
	\end{equation}
\end{definition}
The denominator 2 is introduced to take into account that $X$ has two boundaries. For a connected 2D manifold,  $f(a)$ is a real number that does not depend on the choice of the contractible annulus. This is because the entropy difference is preserved by an isomorphism. Furthermore, $f(1)=0$ by definition. 

Later, we shall study similar contributions for a $n$-hole disk with $n\ge 2$. We will find that $f(a)$ appears generically, even for these more generic subsystems. The repeated appearance of these objects hint at a possibility that there may be nontrivial relations concerning $f(a)$. Indeed, we will later see that $f(a)= \ln d_a$, where $d_a$ is the quantum dimension of an anyon/superselection sector $a$. We will also be able to derive the fusion axioms and an expression for topological entanglement entropy.

\subsection{Fusion rules and fusion spaces} \label{Sec. 3 punctures}
The fusion rules determine the possible choice of the total composite topological charge of two given topological charges.  In the algebraic theory of anyons, the fusion rule for charges $a$ and $b$ is formally written as  
\begin{align}
a\times b=\sum_c N_{ab}^c\,c\,, \nonumber
\end{align}
where $N_{ab}^c\in \mathbb{Z}_{\geq0}$ is the fusion multiplicity.
This is analogous to the fact that two spin-$\frac{1}{2}$ particles can fuse into spin-$0$ or spin-$1$ particle.
Nevertheless, there is a fundamental difference between the fusion of spins and that of anyons. For the definition of particle spins and their fusion rules, rotational symmetry is often needed either in the Hamiltonian or the Lagrangian. In contrast, the notion of topological charges and their fusions are expected to emerge from the collective properties of a many-body quantum system~\cite{anderson1972more,wen2004quantum,2006AnPhy.321....2K}. Indeed, we emphasize that our axioms (Axiom \ref{as:A0'} and \ref{as:A1}) are unrelated to any symmetry.

In our framework, the superselection sectors were identified from annuli. Thus, one may expect the fusion rules to be extracted from $2$-hole disks. In this section, we show that this is indeed the case.  Let us consider a 2-hole disk $Y$, which we depict in Fig.~\ref{Y_BC}. Let $B_1$, $B_2$ and $B_3$ be the three annuli around the boundaries of $Y$. The information convex set of each annulus has the same simplex structure, and we can label the extreme points by the same label set.  
Let $\Sigma^c_{ab}(Y)$ be a convex subset of $\Sigma(Y)$, defined as
\begin{align}
\Sigma^c_{ab}(Y)\equiv\left\{\rho_Y \in\Sigma(Y) \left|   
\begin{array}{l}
\Tr_{Y\backslash B_1} \,\rho_Y= \sigma^a_{B_1} \\
\Tr_{Y\backslash B_2} \,\rho_Y= \sigma^b_{B_2} \\
\Tr_{Y\backslash B_3} \,\rho_Y= \sigma^c_{B_3} 
\end{array} \right. \right\}\,,
\end{align}
where $\sigma^a_{B_i}$ is an extreme point of $\Sigma(B_i)$, $i=1,2,3$. The convention of charge labeling among the different annuli is fixed by Lemma~\ref{prop:topological_deformation}.  $\Sigma_{ab}^c(Y)$ may be empty if there is no state satisfying all the conditions. We call such a combination of $(a,b,c)$ {\it forbidden}.

\begin{figure}[h]
	\centering
\includegraphics[scale=1.1]{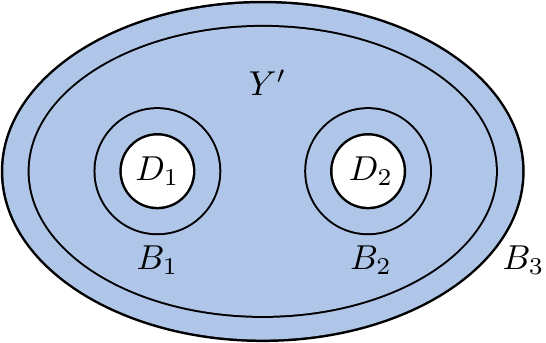}
	\caption{A 2-hole disk $Y=BY'$, with $B=B_1B_2B_3$. $B_1$, $B_2$ and $B_3$ are three annuli surrounding the three boundaries of $Y$.  $YD$ is a disk, where  $D=D_1D_2$. $D_1$ and $D_2$ are the two disks surrounded by annuli $B_1$ and $B_2$. }
	\label{Y_BC}
\end{figure}

We show that every extreme point of $\Sigma(Y)$ is contained in some $\Sigma_{ab}^c(Y)$. Because $\Sigma(Y)$ is a convex set, the entire set can be characterized by $\{\Sigma_{ab}^c(Y)\}$. Each $\Sigma_{ab}^c(Y)$ is isomorphic to the state space of a finite-dimensional Hilbert space $\mathbb{V}_{ab}^c$. These two results are summarized below in  Theorems~\ref{Thm:abc} and \ref{Thm:}.

\begin{theorem}\label{Thm:abc}
	For a 2-hole disk $Y$, the information convex set $\Sigma(Y)$ is the following  convex combination 
	\begin{equation}
	\Sigma(Y)= \left\{\left.\rho_{Y}=\bigoplus_{a,b,c\in\calC}p_{ab}^c\,\,\rho^{abc}_Y\,\right|\; \rho^{abc}_Y\in\Sigma_{ab}^c(Y) \right\}, \label{eq:conv_2-hole}
	\end{equation}
	where $\{p_{ab}^c\}$ is a probability distribution.
\end{theorem}

\begin{proof}
	After taking a partial trace, the reduced density matrix of an extreme point of $\Sigma(Y)$ reduces to an extreme point of $\Sigma(B_1)$, $\Sigma(B_2)$ and $\Sigma(B_3)$. This fact follows from Lemma~\ref{Prop: product_entanglement_cut} in the appendix. Therefore, every extreme point of $\Sigma(Y)$ is in $\Sigma_{ab}^c(Y)$ for some $a, b,$ and $c$. This implies Eq.~(\ref{eq:conv_2-hole}).
\end{proof}

This theorem implies that one can classify the extreme points by a triple of labels $(a,b,c)$. Furthermore, the convex combination in Eq.~(\ref{eq:conv_2-hole}) is orthogonal, since one can perfectly distinguish these labels by projective measurements on the three distinct annuli.

Now we study the geometric structure of each $\Sigma_{ab}^c(Y)$. We should emphasize an important difference between $\Sigma_{ab}^c(Y)$ and the information convex set of an annulus. On an annulus, the information convex set has a classical structure specified by a probability distribution $\{p_a\}_{a\in\calC}$.  
In contrast, $\Sigma_{ab}^c(Y)$ is coherent in the sense that it is isomorphic to the state space $\calS(\mathbb{V}_{ab}^c)$ of a certain finite dimensional Hilbert space $\mathbb{V}_{ab}^c$. If the dimension of $\mathbb{V}_{ab}^c$ is greater or equal to $2$, the structure  of $\Sigma_{ab}^c(Y)$ allows the storage of quantum information. 
This structure is established by the following theorem. 
\begin{restatable}[]{theorem}{twohole}
\label{Thm:}
	Consider a 2-hole disk $Y$.  $\forall a,b,c\in \mathcal{C}$, 
	\begin{equation}\label{eq:simeq_iso}
	\Sigma_{ab}^c(Y) \cong \calS (\mathbb{V}_{ab}^c),
	\end{equation} 
	where $\mathbb{V}_{ab}^c$ is a finite-dimensional Hilbert space. 
\end{restatable}

A particular choice of $(a,b,c)$ is forbidden, when $\dim\mathbb{V}_{ab}^c=0$. See Appendix~\ref{section:fusion_space_proofs} for the proof of Theorem~\ref{Thm:}. The key idea is to show that there is a quantum channel which simultaneously purifies every extreme point of $\Sigma_{ab}^c(Y)$  into a state in Hilbert space $\mathcal{H}_{EY}$ ($E$ is an auxiliary system).  We then show that any  superposition of the purified states reduces to an extreme point of $\Sigma_{ab}^c(Y)$ on $Y$. It follows that the  quantum channel which achieves the purification provides an isomorphism between $\Sigma_{ab}^c(Y)$ and $\calS (\mathbb{V}_{ab}^c)$.

We call the Hilbert space $\mathbb{V}_{ab}^c$ defined in Theorem~\ref{Thm:} as the \emph{fusion space}. Physically, this Hilbert space is nonempty if the superselection sectors $a,b$ has a total charge of $c$. We thus define the fusion rule using the dimension of the corresponding fusion space. 
\begin{definition}\label{def:fusion rule}
	We define the \emph{fusion rule} of labels $a, b$ in $\calC$ by the formal product
	\begin{align}
	a\times b=\sum_{c\in\calC}N_{ab}^c\,c\,,
	\end{align}
	where $N_{ab}^c\equiv \dim\mathbb{V}_{ab}^c$. 
\end{definition} 
The results in Theorem~\ref{Thm:abc} and Theorem~\ref{Thm:} generalize to $n$-hole disks with $n\ge 3$. The same applies to the concepts of fusion spaces and fusion rules.

\subsection{Derivation of the axioms of the fusion rules} \label{Sec. Fusion}

In this section, we show how the axioms of the anyon fusion theory emerge from our axioms. This derivation includes the existence of antiparticles and a set of rules that $\{ N_{ab}^c \}$ has to satisfy. 

We have defined the set of superselection sectors 
\begin{equation}
\mathcal{C}=\{1,a,b,c,\cdots \}\nonumber
\end{equation} 
in terms of the extreme points of $\Sigma(X)$, where $X$ is a contractible annulus. $\mathcal{C}$ is always a finite set and there is a unique sector $1\in \mathcal{C}$ which we refer to as the \emph{vacuum}.  We have also identified a set of non-negative integers $\{ N_{ab}^c \}$ encoded in the structure of $\Sigma(Y)$ with a 2-hole disk $Y$.

The following is a list of the results we are going to prove under our definitions.
\begin{enumerate}
	\item $N_{ab}^c = N_{ba}^c$. (Proposition~\ref{as:2}).
	\item $N_{a1}^c = N_{1a}^c =\delta_{a,c}$. (Proposition~\ref{as:1}.)
	\item The existence of an anti-sector $\bar{a}\in \calC$ for $\forall a\in \calC$ such that $N_{ab}^1=\delta_{b,\bar{a}}$. 
	(Proposition~\ref{Prop: a_bar} and Definition \ref{def: antiparticle}.)
	\item $N_{ab}^c = N_{\bar{b} \bar{a}}^{\bar{c}}$. (Proposition~\ref{as:-2}.)
	\item $\sum_i N_{ab}^i N_{ic}^d = \sum_j N_{aj}^d N^j_{bc}$. (Proposition~\ref{as:-1}.)
\end{enumerate}
Together, these properties form a subset of the axioms of the algebraic theory of anyon outlined in Appendix E of \cite{2006AnPhy.321....2K}, also known as the unitary modular tensor category (UMTC). Concretely, what we derive in this section is the set of axioms of \emph{fusion rule algebra} \cite{frohlich2006quantum} which is also known under the name  \emph{commutative fusion ring} \cite{etingof2016tensor}. It contains  slightly less axioms than a  fusion category because we have not defined the $F$-symbols.\footnote{$F$-symbols will give further constraints to the  fusion multiplicities, the simplest examples are studied in \cite{ostrik2002fusion}. }

In the derivation of the axioms of the fusion rules, we will extensively use the merging technique to relate subsystems of different topologies. From our axioms, we can infer that elements in the information convex sets are quantum Markov states with respect to many relevant partitions of certain subsystems. Because quantum Markov states can be merged together (Lemma~\ref{Prop: merging technique}), we can merge many of the elements of the information convex sets together.  Moreover, this merging process can be repeated many times. With this process, we can generate elements of an information convex set over some region from information convex sets of its subregions. We refer the readers to Proposition~\ref{Prop:merging_Sigma} in the appendix for the technical details.

\begin{Proposition} \label{as:2}
	\begin{equation}
	N_{ab}^c = N_{ba}^c. \label{eq: fusion_6}
	\end{equation}
\end{Proposition}
\begin{proof}
	Let us consider a path which maps a 2-hole disk $Y$ back to itself by exchanging the two internal holes. Associated with this path, there is an automorphism $\Sigma(Y) \cong \Sigma(Y)$. The automorphism permutes the labeling and induces an isomorphism $\Sigma_{ab}^c(Y)\cong \Sigma_{ba}^c(Y)$ for each $a,b,c$.  Thus, $N_{ab}^c = N_{ba}^c$.
\end{proof}

\begin{Proposition} \label{as:1}
	\begin{equation}
	N_{1a}^c  = N_{a1}^c =\delta_{a,c}. \label{eq: fusion_1}
	\end{equation}
\end{Proposition}

\begin{figure}[h]
	\centering
\includegraphics[scale=1.1]{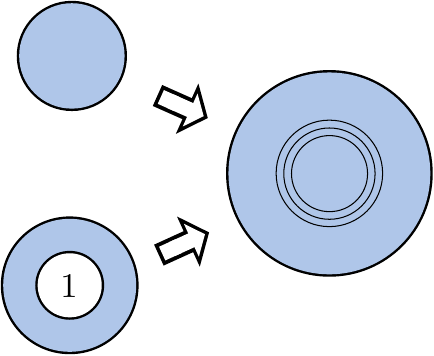}
	\caption{A hole with the vacuum charge can be merged with a disk. The case shown in this diagram involves an annulus and a disk, but the idea works for any $n$-hole disk with $n\ge 1$. The left side shows the topology of the subsystems before they are merged. Also, the number ``$1$" is the vacuum sector.  The merged subsystem is shown on the right. The three concentric lines partition the disk into the four subsystems used in the merging process.}
	\label{Merging_change_patch_III}
\end{figure}

\begin{proof}
	Suppose $\rho_Y \in \Sigma^c_{1a}(Y)$ for 2-hole disk $Y$. Then the hole with the vacuum charge can be merged with a disk, see Fig.~\ref{Merging_change_patch_III}. After the merging process, we obtain an annulus $X$. The density matrix obtained from the merging process belongs to $\Sigma(X)$. The isomorphism theorem implies that the two boundaries of $X$  detect the same topological charge. Therefore, $N_{1a}^c=\delta_{a,c}$. Then, $N_{a1}^c=\delta_{a,c}$ follows from Proposition~\ref{as:2}.
\end{proof}

One implication of this result is that $\Sigma_{11}^1(Y)$ contains a unique element, which we call $\sigma^1_{Y}$. This statement generalizes to $n$-hole disks with $n\ge 3$. 
The following lemma, which is about the universal contribution to the von Neumann entropy, will be useful for the rest of the proofs. Moreover, this lemma will be one of the key results that establish a connection between this contribution and the quantum dimension.
\begin{restatable}[]{lemma}{lemmafabc}
\label{lemma_fabc}
	Let $\rho_Y$ be an extreme point of $\Sigma_{ab}^c(Y)$ and $\sigma^1_Y$ be the unique element of $\Sigma_{11}^1(Y)$, then
	\begin{equation}
	S(\rho_Y)- S(\sigma^1_Y) = f(a) + f(b) + f(c),
	\end{equation}
	where $f(\cdot)$ is the function defined in Definition \ref{def:f(a)}.
\end{restatable}
The proof is in Appendix~\ref{section:fusion_space_proofs}. The key idea is that for an extreme point of $\Sigma_{ab}^c(Y)$, we can prove a condition  similar to that in \ref{as:A0'}, which converts the entropy of a pair of $2$-hole disks into that of the three annuli around the three disjoint boundaries of $Y$. The result generalizes easily to $n$-hole disks for any $n\ge 3$.

Compared to the previous proofs, the proofs of the rest of the properties requires a new technique. The key idea lies in deriving consistency equations of the entropy difference, obtained by the following four steps:

(i) Obtain an element of an information convex set by merging two (or three) extreme points of the  information convex sets associated with the subsystems.

(ii) Compute the entropy of the merged element from the entropy formulas with respect to the pre-merged regions.

(iii) Compute the entropy of the merged element from the entropy formula with respect to the post-merged regions.

(iv) The entropy obtained from these two perspectives must yield the same result. This leads to a set of consistency equations, which leads to a set of nontrivial relations.

\begin{figure}[h]
	\centering
\includegraphics[scale=1.1]{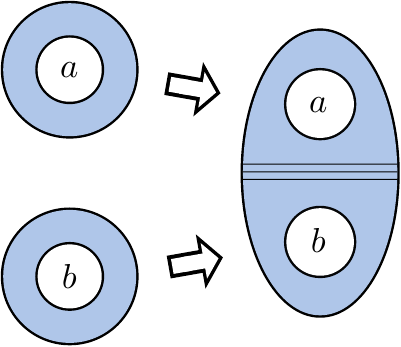}
	\caption{ Merging two annuli and obtain a 2-hole disk.  On the right side, there are two thin disk-like regions in the middle, which are chosen to be the $B$ and $C$ subsystem in the merging lemma (Lemma~\ref{Prop: merging technique}).}
	\label{Merging_a_b}
\end{figure}

For a concrete example of the method, let us study the case shown in Fig.~\ref{Merging_a_b}. The cases in Fig.~\ref{Merging_1_a}, \ref{Last_fusion_1}, \ref{Merging_3_to_4_III}, \ref{Merging_3_annulus_III}, \ref{Merging_change_1_2_III} employ a similar logic. Let us explain the idea, which is broken down into four steps.

(i) We can merge the pair of annuli for any chosen charge pair $a,b\in\calC$. This is possible because the conditions required for merging  are satisfied. Let us call the merged state as $\sigma^{a\times b}_{Y}$. It follows that $\sigma^{a\times b}_{Y}\in \textrm{conv}\big( \bigcup_{c} \,\Sigma_{ab}^c(Y) \big)$. Since the merged state exists, the set $ \textrm{conv}\big( \bigcup_{c} \,\Sigma_{ab}^c(Y) \big)$ is nonempty. This implies that $\sum_c N_{ab}^c \ge 1$, $\forall a,b\in \mathcal{C}$. Moreover, $\sigma^{1\times 1}_Y$ is equal to $\sigma^1_Y$, which is the unique element of $\Sigma^1_{11}(Y)$.

(ii) From the perspective of the two annuli, the von Neumann entropy difference can be expressed as:
\begin{equation}
S(\sigma_Y^{a\times b})-S(\sigma_Y^{1})= 2 f(a) + 2 f(b). \label{eq: perspective_part}
\end{equation}
This result follows from the fact that merging preserves the entropy difference; see property (3) of Lemma~\ref{Prop: merging technique}. More explicitly, this result follows from conditional independence condition ($I(A:D\vert BC)=0$) of the merged state;  see  Fig.~\ref{Merging_a_b}. Here $A$ ($D$) is the upper  (lower) annuli, with charge $a$ ($b$); the two annuli are separated by disk-like region $BC$ in the middle.

(iii) From the perspective of the 2-hole disk $Y$, 
the von Neumann entropy difference is
\begin{equation}
\begin{aligned}
&S(\sigma_Y^{a\times b})-S(\sigma_Y^{1})\\
&= f(a) + f(b) + \ln (\sum_{c} N_{ab}^c e^{f(c)}). \label{eq: perspective_whole}
\end{aligned}
\end{equation}
To derive this result, note that the merged state is the maximal entropy element in $ \textrm{conv}\big( \bigcup_{c} \,\Sigma_{ab}^c(Y) \big)$. This is because the entropy of any state in  $ \textrm{conv}\big( \bigcup_{c} \,\Sigma_{ab}^c(Y) \big)$ can be upper bounded by its marginals by the SSA and the merged state saturates this bound. Given the structure of $\Sigma(Y)$, it is easy to find the maximal entropy in terms of $\{ N_{ab}^c \}$ and $f(\cdot)$. We calculated the maximal entropy and obtained Eq.~(\ref{eq: perspective_whole})~\cite{2019PhRvR...1c3048S}.

(iv) By comparing the two perspectives in Eq.~(\ref{eq: perspective_part}) and Eq.~(\ref{eq: perspective_whole}), we find
\begin{equation}
e^{f(a)} e^{f(b)} = \sum_c N_{ab}^c e^{f(c)}. \label{eq: quantum_dim}
\end{equation}
Readers well-versed in the fusion theory of anyon may have noticed the similarity between $e^{f(a)}$ and the quantum dimension~\cite{2006AnPhy.321....2K}. This is not a coincidence. In Sec.~\ref{Sec. TEE}, we shall see that they are, in fact, the same thing. To establish their equivalence, we will derive a few more identities involving $\{ N_{ab}^c \}$.

Moreover, we can calculate the probability of having charge $c$ on the third boundary
\begin{equation}
P_{(a\times b\to c)} = \frac{N_{ab}^c e^{f(c)}}{e^{f(a)} e^{f(b)}}. \label{eq: fusion_prob}
\end{equation}
Its  physical meaning  is the probability to have an outcome $c$ from the fusion of two independently created charges $a$ and $b$. In terms of the density matrices, $P_{(a\times b\to c)}$ is the coefficient of the element in the center of $\Sigma_{ab}^c(Y)$ when writing $\sigma_Y^{a\times b}$ in terms of a convex combination.

It is worth noting that the same function $f(\cdot)$ appears in the entropy of the annulus and the 2-hole disk. This is crucial for the comparing the two perspectives (Eq.~(\ref{eq: perspective_part}) and (\ref{eq: perspective_whole})). With this equivalence, we are in a position to derive more properties of $\{N_{ab}^c\}$. In deriving these properties, we will curtail our explanation a bit, because the argument is essentially the same.

\begin{figure}[h]
	\centering
	\centering
\includegraphics[scale=1.1]{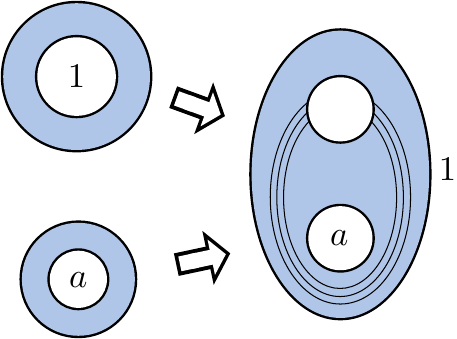}
	\caption{Merging a pair of annuli to obtain a 2-hole disk. We first deform the annulus associated to $1$ so that it becomes ``longer" vertically. Then, the annulus associated to $a$ is merged into the interior of this deformed annulus. The two thin $U$-shaped disk-like regions are chosen to be the subsystem $B$ and $C$ in the merging lemma (Lemma~\ref{Prop: merging technique}).}
	\label{Merging_1_a}
\end{figure}

\begin{Proposition} \label{Prop: a_bar}
	For each charge sector $a\in \calC$, there is a unique sector $\bar{a}\in \calC$ such that
	\begin{equation}
	N_{ab}^1 =\delta_{b,\bar{a}}.     \label{eq: fusion_2} 
	\end{equation}
	It further satisfies the following properties.
	\begin{equation}
	\bar{1}=1,\quad \bar{\bar{a}} =a,\quad f(\bar{a}) = f(a). \label{eq:fa=fbara}
	\end{equation}
\end{Proposition}
\begin{proof}
	From the merging of two annuli with charges $1$ and $a$ shown in Fig.~\ref{Merging_1_a}, we can derive that $\sum_{b}N_{ab}^1\ge 1$, $\forall \,a$ because the merged state always exists. Furthermore, $e^{f(a)} =\sum_{b} N_{ab}^1 e^{f(b)}$.
	
	Let us pick a sector $b$ such that $N_{ab}^1\ge 1$. We can see that $e^{f(a)}\ge e^{f(b)}$. However, since $N^1_{ab}= N^1_{ba}$, by repeating the same logic we obtain $e^{f(b)}\ge e^{f(a)}$. For both of them to be true, we must have a unique sector $\bar{a}$ such that $N_{ab}^1 =\delta_{b,\bar{a}}$ and $f(a) = f(\bar{a})$.
	Then it follows from $N_{11}^1=1$ that $\bar{1}=1$. Since $N^1_{\bar{a}a} = N_{a \bar{a}}^1 =1$, we have $\bar{\bar{a}}=a$.
\end{proof}
\begin{definition}[Antiparticle]\label{def: antiparticle}
	We define the antiparticle of $a\in \calC$ as the unique sector $\bar{a}\in \calC$ established in  Proposition~\ref{Prop: a_bar}.
\end{definition}
The definition of $\bar{a}$ is universal and insensitive to the choice of the subsystem. Furthermore, on a sphere, one could alternatively define $\bar{a}$ according to a nontrivial automorphism of $\Sigma(X)$, see Appendix~\ref{Antiparticle_sphere}.

\begin{Proposition}\label{as:-2}
	\begin{equation}
	N_{ab}^c = N_{\bar{b}\bar{a}}^{\bar{c}}. \label{eq: fusion_4} 
	\end{equation}
\end{Proposition}

\begin{figure}[h]
	\centering
\includegraphics[scale=1.0]{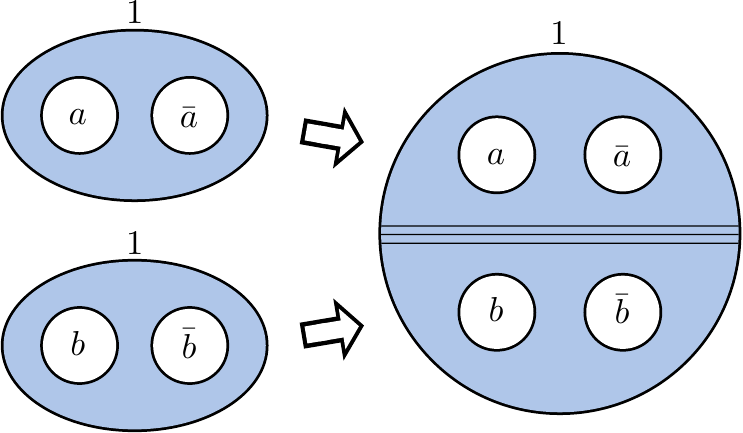}
	\caption{Merging a pair of 2-hole disks to obtain a 4-hole disk.}
	\label{Last_fusion_1}
\end{figure}

\begin{figure}[h]
	\centering
\includegraphics[scale=1.0]{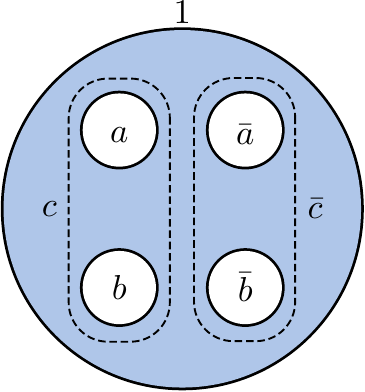}
	\caption{The fusion of $(a,b)$ and  $(\bar{b},\bar{a})$, and matching the fusion probabilities $P_{(a\times b\to c)}$ and $P_{(\bar{b}\times \bar{a}\to \bar{c})}$.  }
	\label{Last_fusion_2}
\end{figure}

\begin{proof}
	We consider the merging process in Fig.~\ref{Last_fusion_1}. Before merging, the density matrices are two extreme points from $\Sigma^1_{a\bar{a}}(Y_u)$ and $\Sigma^1_{b\bar{b}}(Y_d)$, if we call the pair of 2-hole disks as $Y_u$ and $Y_d$. Since $N_{11}^a=\delta_{1,a}$, the outermost boundary of the merged subsystem 
	must have charge $1$. 
	
	Now let us view the merged state in a different way, as depicted in Fig.~\ref{Last_fusion_2}. We have derived that $N^1_{cd}=\delta_{d, \bar{c}}$, which implies that in the merged state, the fusion outcome of $a\times b$ and that of $\bar{b}\times \bar{a}$ are perfectly correlated. Whenever we get the outcome $c$ from the fusion of $a$ and $b$, we must get $\bar{c}$ from the fusion of $\bar{b}$ and $\bar{a}$. 
	
	Furthermore, $a$ and $b$ in this state are ``independently created" in the sense that fusion probability $P_{(a\times b \to c)}$ obeys Eq.~(\ref{eq: fusion_prob}). To see why, consider partial trace operations over (i) a region which connects the hole with charge $\bar{a}$ to the outer boundary and (ii) a region which connects the hole with charge $\bar{b}$ to the outer boundary. These regions are chosen so that the remaining subsystems are topologically equivalent to the ones appearing on the right side of Fig.~\ref{Merging_a_b}. Recalling the general inequality $I(AA':CC'|B) \geq I(A:C|B)$, we observe that the annulus associated with $a$ and the annulus associated with $b$ are independent conditioned on a disk-like region in between them that separates the two annuli. As we have already discussed above, this conditional independence condition implies that $P_{(a\times b \to c)}$ obeys Eq.~(\ref{eq: fusion_prob}). Of course, an analogous argument can be applied to $P_{(\bar{b} \times \bar{a} \to \bar{c})}$.
	
	Because $c$ and $\bar{c}$ are completely correlated,
	\begin{equation}
	P_{(a\times b\to c)} = P_{(\bar{b} \times \bar{a} \to \bar{c})}.
	\end{equation}
	Then, noticing $f(a)= f(\bar{a})$ from Eq.~(\ref{eq:fa=fbara}), we can derive $N_{ab}^c= N_{\bar{b}\bar{a}}^{\bar{c}}$. 
\end{proof}

As mentioned earlier, the results in Theorem~\ref{Thm:abc} and Theorem~\ref{Thm:} generalize to $n$-hole disks with $n\ge 3$, and the same applies to the concepts of fusion space and fusion rules. Let us introduce a few notations for $n=3$ which are useful for the next proof. For a 3-hole disk $Z$, we use $\Sigma(Z)$ to denote its information convex set, $\Sigma_{abc}^d(Z)$ to denote the convex subset of $\Sigma(Z)$ with fixed charges $a,b,c,d$ on the boundaries (Fig.~\ref{Merging_3_to_4_III}). The corresponding fusion space $\mathbb{V}_{abc}^d$ has a finite dimension $N_{abc}^d\in \mathbb{Z}_{\geq 0}$.

\begin{Proposition}\label{as:-1}
	The fusion rules are associative, i.e.,
	\begin{equation}
	N_{abc}^d=\sum_i N_{ab}^i N_{ic}^d = \sum_j N_{aj}^d N_{bc}^j. \label{eq: fusion_5}
	\end{equation}
\end{Proposition}

\begin{proof}	
	The key idea is to obtain a 3-hole disk $Z$ in two different ways; see Fig.~\ref{Merging_3_to_4_III} and \ref{Merging_3_annulus_III}. The first method gives us a lower bound of $N_{abc}^d$ in terms of $N_{ab}^c$, and the second method shows the bound saturates.

	\begin{figure}[h]
	\centering
\includegraphics[scale=1.0]{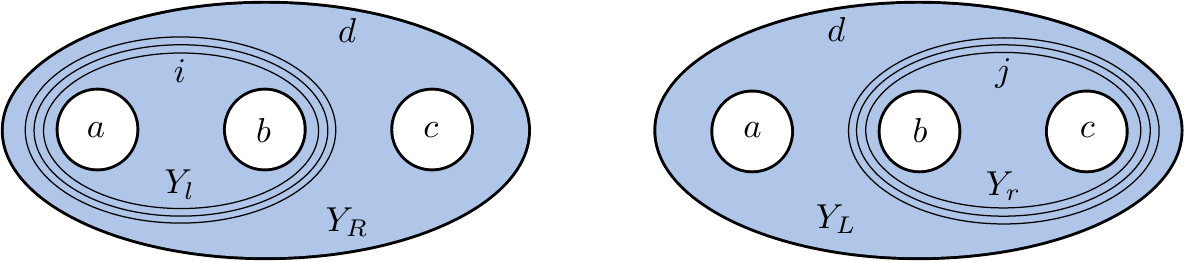}
		\caption{ Merging a pair of 2-hole disks to obtain a 3-hole disk. Here $Z=Y_l \cup Y_R =Y_L \cup Y_r$. Here $a,b,c,d, i,j$ are labels of the topological charges.} 
		\label{Merging_3_to_4_III}
	\end{figure}

	Let us consider the merging of a pair of 2-hole disks to obtain a 3-hole disk shown in Fig.~\ref{Merging_3_to_4_III}. We summarize the logic in a streamlined fashion in (i), (ii), (iii) below. 
	
	(i)  Let us consider the left side of Fig.~\ref{Merging_3_to_4_III}, which describes the merging of $Y_l$ and $Y_R$. We pick an orthonormal basis of $\mathbb{V}_{ab}^{i}$, which can be chosen to be the extreme points of $\Sigma_{ab}^i (Y_l)$.  The number of such extreme points is equal to $N_{ab}^i$, which is the dimension of the Hilbert space $\mathbb{V}_{ab}^{i}$. Applying the same logic to $Y_R$,  we see that the number of these extreme points is $N_{ic}^d$.
	
	(ii) Let us pick two arbitrary extreme points from the sets discussed above (one from $\Sigma_{ab}^i (Y_l)$ and another from $\Sigma_{ic}^d(Y_R)$) and merge them. We get an element in $\Sigma_{abc}^d(Z)$. It is an extreme point. This fact is verified by calculating the von Neumann entropy and making use of the 3-hole version of Lemma~\ref{lemma_fabc}.
	This way, we get $N_{ab}^i N_{ic}^d$ number of extreme points, and any two of them are orthogonal. This follows from the fact that fidelity is nondecreasing under a CPTP map. 
	
	(iii) By applying the merging process for all  $i$, we find $\sum_i N_{ab}^i N_{ic}^d$ mutually orthogonal extreme points of $\Sigma_{abc}^d(Z)$. 
	Therefore, we must have 
	\begin{equation}
	N_{abc}^d\ge \sum_i N_{ab}^i N_{ic}^d. \label{eq: ge1}
	\end{equation}
	The reason is $N_{abc}^d$ is the maximal number of mutually orthogonal extreme points in $\Sigma_{abc}^d(Z)$.
	Similarly, from the right side of Fig.~\ref{Merging_3_to_4_III}, we have
	\begin{equation}
	N_{abc}^d\ge \sum_j N_{aj}^d N_{bc}^j. \label{eq: ge2}
	\end{equation}

	\begin{figure}[h]
	\centering
    \includegraphics[scale=1.1]{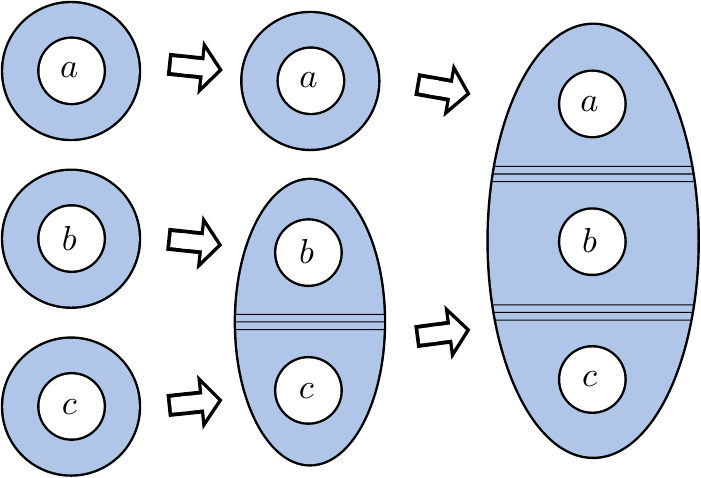}
		\caption{ Merging three annuli to obtain a 3-hole disk.}
		\label{Merging_3_annulus_III}
	\end{figure}
	
	We did not find a way to turn ``$\ge$" into ``$=$" from Fig.~\ref{Merging_3_to_4_III} alone. However, we can show ``=" by considering a different way of merging subsystems; see Fig.~\ref{Merging_3_annulus_III}. 
	The merged element, which we call $\sigma^{a\times b\times c}_Z$, is the maximal entropy element of  $\textrm{conv}\big(\bigcup_{d}\Sigma_{abc}^d(Z) \big)$. Comparing with $\sigma^1_{Z}$, the unique element of $\Sigma_{111}^1(Z)$, it has an extra  $2(f(a)+ f(b) + f(c))$ contribution to the von Neumann entropy. On the other hand, we use the structure of $\Sigma(Z)$  to calculate the maximal entropy in the sector with charge $a$, $b$, $c$ in terms of $\{ N_{abc}^d \}$, we find a contribution equals to $f(a) + f(b) + f(c) +\ln (\sum_d N_{abc}^d e^{f(d)})$. These two perspectives must provide a consistent answer. Thus,
	\begin{equation}
	e^{f(a)}e^{f(b)}e^{f(c)} = \sum_d N_{abc}^d e^{f(d)}.
	\end{equation}
	However, from Eq.~(\ref{eq: quantum_dim}) we know that 
	\begin{equation}
	e^{f(a)}e^{f(b)}e^{f(c)} = \sum_d (\sum_i N_{ab}^i N_{ic}^d) e^{f(d)}
	\end{equation}
	and $e^{f(\cdot)}$ is positive since $f(\cdot)$ is real. So the ``$\ge$" in Eq.~(\ref{eq: ge1}) must be replaced by ``$=$" and the same replacement works for Eq.~(\ref{eq: ge2}). Thus, we conclude that
	Eq.~(\ref{eq: fusion_5}) holds.	
\end{proof}
The result and proof of proposition~\ref{as:-1} generalizes to $n$-hole disks with $n>3$.

\section{Topological entanglement entropy}\label{Sec. TEE}
In this section, we show that the sub-leading term $\gamma$ of the area law~(\ref{eq:area_law}) for a disk is given by the well-known formula
\begin{align}\label{eq:TEE}
\gamma=\ln \mathcal{D}\,,    
\end{align}
where  $\mathcal{D}$ is the total quantum dimension defined from our definition of the fusion multiplicities $\{N_{ab}^c\}$. 
We show this result by calculating two different linear combinations of subsystem entropies\footnote{More precisely speaking, we show that a certain linear combination of entanglement entropy must be $\ln \mathcal{D}$. This result implies that $\gamma =\ln \mathcal{D}$ is the only consistent value of the sub-leading term in the area law formula Eq.~(\ref{eq:area_law}).} respectively proposed by Kitaev-Preskill~\cite{Kitaev2006} and Levin-Wen~\cite{Levin2006}, see Fig.~\ref{TEE_LW_PK}.

\begin{figure}[h]
	\centering
\includegraphics[scale=1.1]{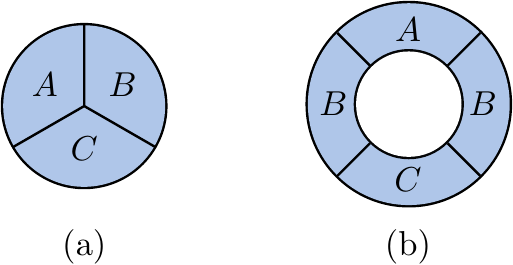}
	\caption{(a) The Kitaev-Preskill partition; (b) the Levin-Wen partition.}
	\label{TEE_LW_PK}
\end{figure}

The sub-leading term $\gamma$ is called the topological entanglement entropy (TEE)~\cite{Kitaev2006}. There are two known methods for deriving TEE: assuming an underlying field theory description or explicitly calculating entropy in an exactly solvable model. Our method, on the other hand, shows that the area law formula itself implies the equivalence of TEE and $\ln\calD$, which may be applicable to a larger class of systems. The ingredients behind this proof are scattered in literature~\cite{Kim2013,Kim2015,2016PhRvA..93b2317K,2019PhRvB..99c5112S}. Recently, one of us showed that the quantum dimension must show up in the von Neumann entropy if the fusion space is coherently encoded in the 2-hole disk~\cite{2019PhRvR...1c3048S}. In this work, we further reduce the assumption to our Axiom \ref{as:A0'} and \ref{as:A1}. The end result is the same.

We begin by defining the quantum dimensions in our framework. 
\begin{definition}
	We define the set of quantum dimensions $\{ d_a \}$ as the unique positive solution of the equation set
	\begin{equation}\label{eq:qddef}
	d_a d_b =\sum_{c} N_{ab}^c \,d_c\,,
	\end{equation}
	where $N_{ab}^c$ is defined in Definition~\ref{def:fusion rule}. We also define the total quantum dimension $\calD$ by $\mathcal{D}=\sqrt{\sum_{a\in\calC} d_{a}^2}$. 
\end{definition}
Note that given the results in Sec.~\ref{Sec. Fusion}, the uniqueness of Eq.~\eqref{eq:qddef} is guaranteed by the Perron-Frobenius theorem, see e.g.  appendix of~\cite{2019PhRvR...1c3048S} for a self-contained derivation. Furthermore,
\begin{equation}
d_{\bar{a}} = d_{a},\quad d_a\ge 1,\quad d_1=1.
\end{equation}
Recall that, from the merging in Fig.~\ref{Merging_a_b}, we have obtained Eq.~(\ref{eq: quantum_dim}), and since $e^{f(a)}$ is positive, we must have
\begin{equation}
f(a)= \ln d_{a}. \label{eq:f(a)}
\end{equation}

In Ref.~\cite{Levin2006}, it is proposed that the conditional mutual information $I(A:C|B)$ for the partition in Fig.~\ref{TEE_LW_PK}(b) matches to $2\ln\calD$. In the paper, it is proven for a class of exactly solvable model called the Levin-Wen model (also known as the string-net model)~\cite{2005PhRvB..71d5110L}. Here we show that the same formula also holds in our framework. 
\begin{Proposition} \label{Prop: LW_TEE}
	For the Levin-Wen partition (Fig.~\ref{TEE_LW_PK}$(b)$), it holds that
	\begin{equation}
	I(A:C\vert B)_{\sigma^1}= 2 \ln \mathcal{D}. \label{eq: TEE-LW}
	\end{equation}
\end{Proposition}

\begin{figure}[h]
	\centering
	\includegraphics[scale=1.1]{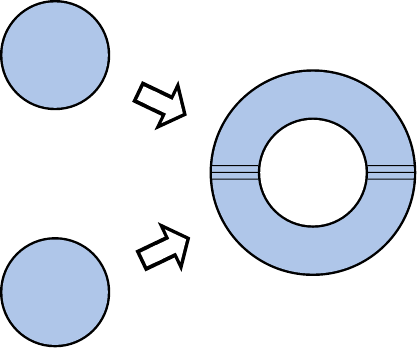}
	\caption{Merging a pair of disks to obtain an annulus. Two disks are deformed so that, once merged together, they form an annulus.}
	\label{Merging_change_1_2_III}
\end{figure}

\begin{proof}
	Let us consider the merging process in Fig.~\ref{Merging_change_1_2_III}, which obtains an annulus $X$ from a pair of disks.
	Let  $\tilde{\sigma}_{X}\in\Sigma(X)$ be the element obtained from merging. It is in the center of $\Sigma(X)$, i.e. the maximal entropy element. Dividing $X$ according to the Levin-Wen partition in Fig.~\ref{TEE_LW_PK}(b), gives
	$I(A:C\vert B)_{\tilde{\sigma}} =0$
	because of the property of the merged state; see Lemma~\ref{Prop: merging technique}. 
	
	Because of the simplex structure of $\Sigma(X)$ (see theorem~\ref{Prop: structure_2}) and the fact that $f(a)$ is equal to $\ln d_{a}$ (see Eq.~(\ref{eq:f(a)})), we can express $\tilde{\sigma}_{X}$ as a convex combination of extreme points
	\begin{equation}
	\tilde{\sigma}_{X} = \sum_{a} \frac{d^2_{a}}{\mathcal{D}^2} \,\sigma^a_{X}. \nonumber
	\end{equation}
	This formula is obtained by maximalizing the von Neumann entropy.
	From it, one derives
	$S(\tilde{\sigma}_X) - S(\sigma^1_{X}) = 2 \ln \mathcal{D}$.
	It follows that Eq.~(\ref{eq: TEE-LW}) is true. 
\end{proof}

\begin{restatable}[]{Proposition}{tee}
\label{Prop: KP_TEE}
	For the Kitaev-Preskill partition,
	\begin{equation}
	\gamma \equiv (S_{AB} + S_{BC} + S_{CA} - S_A - S_{B} - S_{C} - S_{ABC})_{\sigma_\omega} \label{eq: KP_TEE}
	\end{equation}
	where $\omega=ABC$, see Fig.~\ref{TEE_LW_PK}(a), then $\gamma = \ln \mathcal{D}$.
\end{restatable}
The idea of the proof is to relate the Levin-Wen combination with two copies of Kitaev-Preskill combinations. See Appendix~\ref{P_Sec. Proof of {Prop: KP_TEE}} for the proof.

Because there are gapped systems in which these axioms are violated by \emph{spurious} contributions to the area law~\cite{Bravyi2008,2016PhRvB..94g5151Z}, one should not expect our result is applicable to every  gapped system. These violations may be pathological unless certain symmetries are imposed. However, they can persist in certain subsystem symmetry-protected phases \cite{2019PhRvL.122n0506W}. 
Reconciling our framework with these systems remains as an outstanding open problem. Nonetheless, our result does shed some light on a related issue: if we check the quantum state on a finite length scale and verify \ref{as:A0'} and \ref{as:A1}, then it is guaranteed that TEE will not suffer from any spurious contribution on all larger length scales.

\section{Summary and discussions}\label{Sec: discussion}

In this paper, we have initiated a derivation of the axioms of the algebraic theory of anyon from a conjectured form of entanglement area law for the ground states of 2D gapped phases. 
Our framework is based on two entropic constraints (axiom \ref{as:A0'} and \ref{as:A1}), which are implied by the area law formula. We have defined the superselection sectors and the fusion spaces through the geometry of the information convex sets. The axioms of the anyon fusion theory are derived from the internal self-consistency relation of the information convex sets. 
Moreover, we have provided a rigorous derivation of the well-known formula for TEE, $\gamma=\ln\calD$. While our main physical motivation was to consider ground states of 2D gapped phases, 
we only required a single quantum state satisfying our axioms as the input. Our result implies that  many of the anyon data can be extracted from local information of a single ground state alone.

Some of the readers may contest that our exact area law assumption is unrealistic. It would be desirable to relax this assumption to something that is less restrictive. We expect our framework to have a natural extension to the case in which Axiom \ref{as:A0'} and \ref{as:A1} holds approximately. This is because every technical tool we have used in this paper has an analog for such situations. For instance, the merging lemma can be generalized by using the approximate recovery map~\cite{Fawzi2015}. 

It should be noted that there are gapped systems in which Axiom \ref{as:A1} is violated. Such corrections  are known as the \emph{spurious} contributions to the area law~\cite{Bravyi2008,2016PhRvB..94g5151Z}. The existence of the spurious contribution implies that one should not expect our result to hold in every gapped system. While we do not have a solution to this problem, one may hope to take one of the following approaches. First, one may show that the notion of superselection sectors and the fusion rules are stable under a finite-depth quantum circuit when starting from a state that satisfies our axioms. Alternatively, one may attempt to show that there is always a finite-depth quantum circuit that can remove the spurious contribution.

While we have proved a set of axioms pertinent to the anyon fusion theory, further work is necessary to reproduce the anyon theory in its known form. It would be interesting to investigate whether our axioms give rise to a well-defined notion of $R$ and $F$-symbols. Also, could the $S$ and $T$-matrices be extracted from a single ground state? Can we prove that every anyon theory consistent with our axioms are modular? To tackle these questions, one may need to recover a certain $U(1)$ phase that is missing in the density matrix formulation, perhaps with the help of the string operator shown to exist in Appendix~\ref{appendix:String}. 

A more ambitious question is whether we can arrive at a complete classification of two-dimensional gapped phases from our axioms. The current conjecture~\cite{2006AnPhy.321....2K} is that two systems are in the same phase if and only if their underlying anyon theory and the chiral central charges are identical. In Ref.~\cite{2006AnPhy.321....2K}, Kitaev speculated: ``To prove or disprove this statement, a mathematical notion of equivalence between topological phases is necessary. It may be based on local (or quasilocal) isomorphisms between operator algebras."  Our framework seems to be the step in the right direction, given that we have a sensible definition of isomorphism between different subsystems and that we could derive axioms of the anyon fusion theory from a reasonable physical assumption. Such a feat will be a complete and rigorous justification of the point of view that 2D gapped quantum phases can be described by the anyon theory~\cite{2006AnPhy.321....2K,wen2004quantum}.

Compared to the existing approaches that compute entanglement measures on ground states, information convex set leads to a more incisive understanding of the underlying topological phase. Upon calculating entanglement entropy from a given ground state, one often obtains an order parameter that reveals partial information about the underlying quantum phase. An oft-cited example is the total quantum dimension \cite{Kitaev2006,Levin2006}. Remarkably, we found that a much more refined set of data can be extracted by studying the information convex set of the ground state. This includes quantum dimensions and fusion multiplicities. Because two different anyon theories can give rise to the same total quantum dimension but different (individual) quantum dimensions and fusion multiplicities, information convex set clearly contains richer data than entanglement entropy.

Moreover, starting from our entanglement-based assumptions (Axiom \ref{as:A0'} and \ref{as:A1}), we were able to derive the basic concepts of anyon theory from the internal consistency relations between the elements of the information convex set. What is surprising is that the emergent physical laws that appear in these topologically ordered systems were logical consequences of these axioms; we did not need to make any further assumptions. It would be interesting to understand how widely this approach can be applied in other circumstances.

Obvious areas to explore further would be higher dimensions and setups in which a topological defect~\cite{Bombin2010} or a boundary is present~\cite{Bravyi1998}. Such studies may be an ideal framework to classify topological phases in 3D, which remains as an outstanding open problem. We will discuss these applications in our upcoming work.

\section*{Acknowledgments} BS wishes to thank Yilong Wang for a discussion on possible obstruction of the categorification of fusion rings. BS is supported by the National Science Foundation under Grant No. NSF DMR-1653769. IK's research at Perimeter Institute was supported by the Government of Canada through Industry Canada and by the Province of Ontario through the Ministry of Economic Development and Innovation. IK also acknowledges support from IBM T. J. Watson research center as well as Simons foundation. KK is supported by the Institute for Quantum Information and Matter (IQIM) at Caltech, which is a National Science Foundation (NSF) Physics Frontiers Center (NSF Grant PHY-1733907).

\appendix
\renewcommand*{\thesection}{\Alph{section}}
\renewcommand\thefigure{\arabic{figure}}

\section{Notations and useful facts}  
In this appendix, we summarize basic notations of convex analysis and quantum information theory. Well-known facts in quantum information theory will be summarized in a self-contained manner. This section can be skipped for readers who are familiar with convex analysis and strong subadditivity of entropy.

\subsection{Convex sets}\label{app:convexsets}
Here we review facts about convex sets. We consider a subset of a finite-dimensional real space $\mathbb{R}^N$ closed under convex combinations, where $N\in \mathbb{Z}_{\ge0}$. The convex set is compact if it is a compact subset of $\mathbb{R}^N$. For our purpose, for an $N$ dimensional Hilbert space, the real space  $\mathbb{R}^{2N^2}$ could be identified as the $2N^2$ real components of an operator acting on the Hilbert space.

We use $\conv(\mathcal{X})$ to denote
the convex hull of a set $\mathcal{X} \subseteq \mathbb{R}^{N}$, which is the smallest convex set that contains set $\mathcal{X}$. In other words, it is the set of all convex combinations of elements in $\mathcal{X}$.

An \emph{extreme point} of a convex set $\calS$ is a point in $\calS$, which does not lie in any open line segment joining two points of $\calS$. We use $\ext(\calS)$ to denote the set of extreme point of a convex set $\calS$.

Finally, we notice the \emph{Minkowski-Caratheodory theorem}, which states that: 
\emph{Let $\calS$ be a compact convex subset of $\mathbb{R}^{N}$ of dimension $n$. Then any point in $\calS$ is a convex combination of at most $n +1$ extreme points.}  This is the reason we often talk about extreme points. Note that, without compactness, an element of a convex set sometimes cannot be written as a convex combination of extreme points.

\subsection{Quantum information facts}\label{app:QIbasic}
We shall use Greek letters, e.g., $\rho$, $\sigma$ for density matrices. Subsystems are specified in the subscript, e.g. $\rho_{A}$, $\sigma_{B}$. We frequently use $AB$ as a shorthand notation of the disjoint union of $A$ and $B$ (i.e., $A\sqcup B$) when $A\cap B=\emptyset$. We will sometimes call a (reduced) density matrix as a state for short. Also, we will sometimes refer to the reduced density matrix of a state as its marginal. We use $\calS(\mathcal{H})$ to denote the \emph{state space} of a Hilbert space $\mathcal{H}$. It is the set of all density matrices on $\calH$. A direct sum is denoted as $\bigoplus$. It is a sum with objects living on orthogonal supports. 

In order to quantify the distance between two quantum states, we use the trace distance, defined as 
\begin{equation*}
\| \rho - \sigma \|_{1}\equiv \Tr{\sqrt{(\rho-\sigma)^2}}
\end{equation*}
for any pair of density matrices $\rho$ and $\sigma$. This is a reasonable notion of distance because two states close in trace distance cannot be distinguished well by any measurement.

We shall use a variety of quantum mechanical entropies in our discussion. These are all defined in terms of the von Neumann entropy of a state, which is defined as 
\begin{equation*}
S(\rho)\equiv -\Tr(\rho\ln \rho).
\end{equation*}
Depending on the context, we shall use the following shorthand notations to denote the von Neumann entropy of the reduced density matrix over some subsystem: $S_A$, $(S_A+S_B)_{\rho}$. In the first case, the global state should be obvious from the context. In the second case, the global state is $\rho$.

There are two information-theoretic quantities that will play an important role in this paper:
\begin{equation*}
\begin{aligned}
I(A:B) &\equiv S_A + S_B - S_{AB}, \\
I(A:C|B) &\equiv S_{AB} +S_{BC} -S_B -S_{ABC}.
\end{aligned}
\end{equation*}
The first quantity, known as the mutual information between $A$ and $B$, quantifies a correlation between $A$ and $B$. The second quantity, known as the conditional mutual information between $A$ and $C$ conditioned on $B$, quantifies the correlation between $A$ and $C$ given a knowledge on $B$. By the strong subadditivity of entropy~\cite{1973JMP....14.1938L}, $I(A:C\vert B) \geq 0$ for any quantum state.

Below, we summarize the basic facts about quantum states and entropies. Most of these statements can be found in~\cite{Nielsen2011}.

\subsubsection{Fidelity} 

Let us begin with the fidelity between two quantum states, which is defined as 
\begin{equation}
F(\rho,\sigma)=\left( \Tr \sqrt{{\rho^{\frac{1}{2}}}\sigma \rho^{\frac{1}{2}}} \right)^2.
\end{equation}
This is a natural generalization of the absolute value of inner product into mixed state (and from the whole system to subsystems). Indeed, one can easily verify $F(\rho,\sigma)= \vert \langle \psi\vert \varphi\rangle \vert^2$ when $\rho=\vert\psi\rangle \langle \psi\vert$ and $\sigma= \vert \varphi\rangle \langle \varphi\vert$. Furthermore, $F(\rho, \sigma) \in [0,1]$ and $F(\rho,\sigma)= F(\sigma,\rho)$.

Why do we care about fidelity? It is because fidelity enjoys several useful properties. As a starter, $F(\rho,\sigma)$ can give a reasonably tight upper and lower bound on $\|\rho - \sigma \|_1$. In particular, $F(\rho,\sigma)=1$ if and only if $\| \rho - \sigma\|_1=0$ and $F(\rho,\sigma)=0$ if and only if $\|\rho-\sigma \|_1=2$. Moreover, fidelity has a rather special property:
\begin{equation}
F(\rho_A\otimes \rho_B,\sigma_A\otimes \sigma_B) = F(\rho_A,\sigma_A)\cdot F(\rho_B,\sigma_B),
\end{equation}
which will play an important role in the proof of Theorem~\ref{Prop: structure_2}. There is an intuitive explanation for both of these facts. The first fact says that two density matrices have zero (unit) overlap if and only if the two states are orthogonal (identical). 

Lastly, if two quantum states have unit fidelity, their purifications are identical up to a unitary operator acting on the purifying space~\cite{Uhlmann1976}. In other words, two states $\vert \psi_{AB}\rangle$ and $\vert \varphi_{AB}\rangle$ have the same reduced density matrix on subsystem $B$ if only if there is a unitary operator $U_A$ such that 
\begin{equation}
\vert \varphi_{AB}\rangle = U_{A}\otimes I_B \vert \psi_{AB}\rangle. \label{eq:U_A}
\end{equation}

\subsubsection{Quantum channels}

Quantum channel, also known as complete-positive trace-preserving (CPTP) map, is the most general form of physical operation that can be applied to a quantum state. It is a linear map from bounded operators  on $\mathcal{H}_{A}$ to bounded operators  on $\mathcal{H}_{A'}$. It preserves positivity, even in the presence of any ancillary system, and also preserves trace and hermiticity.\footnote{That we require the map to preserve positivity in the presence of any ancillary system is important. Otherwise, there are operations, e.g., transpose of a matrix, that preserves the positivity in the absence of an ancillary system but may not if the ancillary system is entangled with the system of interest.} In particular, it maps density matrices to density matrices. It can be written in an explicit form using a set of Kraus operators $\{ M_a \}$:
\begin{equation}
\mathcal{E}_{A\to A'}(X_{A})= \sum_a M_a X_A M_a^{\dagger},
\end{equation}
where $\sum_a M_a^{\dagger} M_a = I_{A}$ and $I_A$ is the identity operator on $\mathcal{H}_A$. With the definition of CPTP map, we could discuss some additional properties.

Quantum channels do not make two quantum states more distinguishable than they already are. This means that the trace distance is nonincreasing under quantum channels. More relevant to us is the fact that the fidelity is nondecreasing under quantum channels.
\begin{equation}
F(\mathcal{E}(\rho),\mathcal{E}(\sigma)) \ge F(\rho,\sigma),
\end{equation}
for any quantum channel $\mathcal{E}$. Since partial trace is also a quantum channel, we have
\begin{equation}
F(\rho_A,\sigma_A) \ge F(\rho_{AB},\sigma_{AB}). \label{eq: monotonicity of fidelity} 
\end{equation}

\subsubsection{Properties of entropies}

Let us begin with a few elementary facts about entropy. First, $S_{A}=S_{B}$ for an arbitrary pure state $\vert \varphi_{AB}\rangle$. Secondly, suppose a set of density matrices $\{ \rho^i \}$ has mutually orthogonal supports, i.e. $\rho^i\perp \rho^j,\forall i\ne j$, then
\begin{equation}
S(\sum_i p_i \rho^i)= \sum_i p_i (S(\rho^i)-\ln p_i),
\end{equation}
where $\{ p_i \}$ is a probability distribution.

For a bipartite quantum state, we have the following set of well-known inequalities:
\begin{equation}
\begin{aligned}
I(A:B)_{\rho} &\ge 0 \\
S_{BC} + S_C - S_B &\ge 0.
\end{aligned}
\end{equation}
The first inequality is known as the subadditivity of entropy, and the second inequality is known as the Araki-Lieb inequality~\cite{1970CMaPh..18..160A}. It is interesting to study the conditions under which these inequalities are satisfied with equality. The mutual information is $0$ if and only if the underlying state is a product state over $A$ and $B$. The condition for saturating the Araki-Lieb inequality is more subtle and interesting. We will revisit this condition later after we discuss inequalities for tripartite quantum states.

The most important inequality involving a tripartite quantum state is the strong subadditivity (SSA) of entropy~\cite{1973JMP....14.1938L}:
\begin{equation}
I(A:C|B) \ge 0.
\end{equation}
This inequality is surprisingly powerful in that inequalities that may look ``stronger" than this inequality are in fact implied by SSA. Here is a list of such inequalities:
\begin{equation*}
\begin{aligned}
&I(AA':BB') \ge I(A:B) \\
&I(AA':CC'\vert B) \ge I(A:C\vert B)\\
&I(AA':CC'\vert B) \ge I(A:C\vert A'BC')\\
&S_{BC} + S_{C} - S_B \ge I(A:C)\\
&S_{BC} + S_{C} - S_B \ge I(A:C\vert B)\\
&S_{BC} + S_{C} - S_B \ge S_{BB'C} + S_C -S_{BB'}\\
&S_{BC} + S_{CD} - S_B - S_D\ge I(A:C\vert B)\\
&S_{BC} + S_{CD} - S_B - S_D\ge S_{BB'C} + S_{CDD'} - S_{BB'} - S_{DD'}.
\end{aligned}
\end{equation*}
Also, let $\{ \rho^i_{AB} \}$ be a set of density matrices and $\{ p^i \}$ is a probability distribution, then
\begin{equation}
\sum_i p_i(S_{AB}-S_B)_{\rho^i} \le (S_{AB}-S_B)_{\sum_i p_i \rho^i}.
\label{eq: SSA_combination}
\end{equation}
To see why, let us introduce an auxiliary system $C$ with an orthonormal basis $\{ \vert i_C \rangle \}$. Let $\rho_{ABC}\equiv \sum_i p_i \, \rho^i_{AB}\otimes \vert i_C \rangle \langle i_C \vert$
and notice $I(A:C\vert B)_{\rho}\ge 0$.

\subsubsection{The structure of quantum Markov states}
If a tripartite quantum state satisfies SSA with equality, then such a state has a rather special property. Such a state is referred to as a quantum Markov state~\cite{Petz1987,2004CMaPh.246..359H}. Let $\rho$ be a quantum Markov state such that it has a vanishing conditional mutual information $I(A:C\vert B)_{\rho}=0$. Then the following facts hold.
\begin{enumerate}
	\item $\rho_{ABC}$ is uniquely determined by its marginals $\rho_{AB}$ and $\rho_{BC}$. The recovery can be done with a quantum channel, see Lemma~\ref{lemma_growth} and \ref{lemma_Petz} of the main text.
	
	\item There is a decomposition of the Hilbert space $\mathcal{H}_B$ into a direct sum of tensor products
	$	\mathcal{H}_B =\bigoplus_{j}\mathcal{H}_{b^{L}_j} \otimes \mathcal{H}_{b^{R}_j}$
	such that 
	\begin{equation}
	\rho_{ABC}= \bigoplus_{j} p_{j} \,\rho_{Ab^{L}_{j}}\otimes \rho_{b^{R}_{j}C}, \label{eq: QMS_explicit}
	\end{equation}
	where $\{p_j\}$ is a probability distribution, $\rho_{Ab^{L}_{j}}$ is a density matrix on $\mathcal{H}_A \otimes \mathcal{H}_{b^{L}_{j}}$ and  $\rho_{Ab^{R}_{j}}$ is a density matrix on $\mathcal{H}_{b^{R}_{i}}\otimes \mathcal{H}_C$.
	\item Eq.~(\ref{eq: QMS_explicit}) implies that 
	\begin{equation}
	\Tr_{B} \, \rho_{ABC}= \sum_j p_j \rho^j_A \otimes \rho^j_C.  \label{eq: QMS_explicit_mutual}
	\end{equation}
	Note that it is separable and therefore subsystem $A$ and $C$ have only classical correlations (no quantum correlation). 
\end{enumerate}

Now, we can make an important connection between the states saturating the Araki-Lieb inequality and quantum Markov states. The density matrices which saturate Araki-Lieb inequality have the properties summarized in  Lemma~\ref{Prop: the isolation technique}.
\begin{lemma}\label{Prop: the isolation technique}
	The following conditions about density matrix $\rho_{BC}$ are equivalent.
	\begin{enumerate}
		\item[(1)] $(S_{BC}+S_C -S_B)_{\rho}=0$, (saturated Araki-Lieb).
		\item[(2)] Any state $\rho_{ABC}$ which reduces to $\rho_{BC}$ on $BC$ has $I(A:C)_{\rho}=0$ and $I(A:C\vert B)_{\rho}=0$.
		\item [(3)] For any expression of the form $\rho_{BC}=\sum_i q_i \, \rho^i_{BC}$, where  $\{q_i \}$ is a probability distribution with $q_i>0,\,\forall\, i$ and $\{ \rho^i_{BC} \}$ is a set of density matrices, we have
		\begin{equation}
		\rho_{C}=\Tr_{B}\, \rho^i_{BC},  \quad \forall \, i. \label{eq:trac_B_constraint}
		\end{equation} 
		\item [(4)] Let $\rho_{BC}= \sum_{i} q_i \vert i_{BC}\rangle \langle i_{BC} \vert$, with $q_i>0$, $\forall \, i$ and $\langle i_{BC}\vert j_{BC}\rangle =\delta_{i,j}$, $\forall \, i,j$, we have 
		\begin{equation}
		\Tr_B \,\vert i_{BC} \rangle \langle j_{BC} \vert =\delta_{i,j}\, \rho_{C},\quad \forall\, i,j. \label{eq: off_diagonal}
		\end{equation}
	\end{enumerate}
\end{lemma}

\begin{proof}
	$(1)\Rightarrow(2)$. Let us purify $\rho_{ABC}$ and obtain $\vert\Psi_{A'ABC}\rangle$.  Condition $(1)$ implies $I(A'A:C)_{\vert\Psi\rangle \langle \Psi\vert}=0$ and  $I(A'A:C\vert B)_{\vert\Psi\rangle \langle \Psi\vert}=0$. Then, we use the facts $I(A'A:C)\ge I(A:C)$ and $I(A'A:C\vert B)\ge I(A:C\vert B)$.
	
	$(2)\Rightarrow(1)$.  Simply consider a pure $\rho_{ABC}$. 
	
	$(1),(2)\Rightarrow(3)$. Let $\rho_{ABC}=\sum_i q_i \vert i_A\rangle \langle i_A\vert \otimes \rho_{BC}^i$ with an orthonormal set of vectors $\{ \vert i_A\rangle \}$.  Since $I(A:C)_{\rho}=0$, a projective measurement in $\{ \vert i_A\rangle \}$  will not change the density matrix on $C$. It follows that $\rho_{C}=\Tr_{B}\, \rho^i_{BC},  \forall \, i$.

	$(3)\Rightarrow(4)$. Let us introduce an auxiliary system $A$ to purify $\rho_{BC}= \sum_i q_i \vert i_{BC}\rangle \langle i_{BC}\vert$ into
	$\vert \Psi_{ABC}\rangle=\sum_i\sqrt{q_i} \, \vert i_{A}\rangle \otimes \vert i_{BC}\rangle$. Here $\{ \vert i_{A}\rangle \}$ is an orthonormal basis of the Hilbert space $\mathcal{H}_A$. Therefore, one may obtain $\rho_{BC}$ from $\vert \Psi_{ABC}\rangle \langle \Psi_{ABC}\vert$ by taking a partial trace $\Tr_{A}$. For an arbitrary orthonormal basis $\{ \vert  \phi^i_A \rangle \}$ we could define $p_i \,\rho^i_{BC}\equiv \langle \phi^i_A\vert  \Psi_{ABC}\rangle \langle \Psi_{ABC}\vert \phi^i_A \rangle $. Recall that the condition Eq.~(\ref{eq:trac_B_constraint}) applies to any basis. For the ``diagonal" basis $\{ \vert i_A\rangle \}$, one derives $\Tr_B \,\vert i_{BC} \rangle \langle i_{BC} \vert =\rho_{C}$, $\forall\, i$. Then, consider an ``off-diagonal" basis which contains a basis vector $\vert \varphi_{ij}^{\theta}\rangle =\frac{1}{\sqrt{2}}(\vert i_A\rangle + e^{i\theta} \vert j_A\rangle )$ with $i\ne j$ and $\theta\in [0,2\pi]$. One derives that,  for $\forall \,i\ne j$ and $\forall \,\theta\in [0,2\pi]$,
	\begin{equation}
	e^{i\theta}\Tr_B \,\vert i_{BC} \rangle \langle j_{BC} \vert + e^{-i\theta}\Tr_B \,\vert j_{BC} \rangle \langle i_{BC} \vert=0. \nonumber
	\end{equation}
	Therefore, $\Tr_B \,\vert i_{BC} \rangle \langle j_{BC} \vert =0$ for $i\ne j$. Thus, Eq.~(\ref{eq: off_diagonal}) holds.
	
	$(4)\Rightarrow(1)$. Let us introduce an auxiliary system $A$ to purify $\rho_{BC}$ into
	$\vert \Psi_{ABC}\rangle=\sum_i\sqrt{q_i} \vert i_{A}\rangle \otimes \vert i_{BC}\rangle$. Here $\{ \vert i_{A}\rangle \}$ is an orthonormal basis of the Hilbert space $\mathcal{H}_A$. Then, it follows from Eq.~(\ref{eq: off_diagonal}) that $I(A:C)_{\vert\Psi\rangle \langle \Psi\vert}=0$ which  implies condition (1).
\end{proof}

\section{Information convex set on a closed manifold}
In this section, we show that the information convex set of a closed manifold is isomorphic to a (state space of a) finite-dimensional Hilbert space; see Theorem~\ref{Prop: Sigma(M)}.

\begin{theorem}\label{Prop: Sigma(M)}
	Let $M$ be a closed 2D manifold. Let $\sigma$ be a reference state on this manifold satisfying axiom \ref{as:A0'} and \ref{as:A1}. With respect to this reference state, 
	\begin{equation}
	\Sigma(M)=\mathcal{S}(\mathbb{V})
	\end{equation}
	for some finite dimensional Hilbert space $\cal\mathbb{V}\subseteq \calH$. Moreover, $\mathbb{V}$ is nonempty.
\end{theorem}

\begin{proof}
Because $\Sigma(M)$ is a subset of the state space of a finite-dimensional Hilbert space, 
one can represent an element $\rho \in \Sigma(M)$ as 
\begin{equation*}
    \rho = \sum_i p_i |i\rangle\langle i|
\end{equation*}
for some $p_i>0$ and an orthonormal set of vectors $\{|i\rangle \}$. The state $\rho$ is consistent with $\sigma$ on any disk. This fact follows from Proposition~\ref{Prop:disk_in_Omega}.

We show below that the reduced density matrices of $|i\rangle\langle i|$ on the $\mu$-disks are equal to those obtained from the reference state $\sigma$. Without loss of generality, consider a $\mu$-disk, say $c$. Consider an annulus that surrounds $c$, which we denote as $b$. This annulus is chosen so that $bc$ is again a disk. By Proposition~\ref{Prop:axiom_large}, $(S_{bc} + S_c - S_b)_{\rho}=0$. By Lemma~\ref{Prop: the isolation technique} $(1)\Rightarrow(3)$, for $\rho_{bc}=\sigma_{bc} = \sum_i p_i \Tr_{M\setminus bc}(|i\rangle\langle i|)$, we find that
\begin{equation}
    \sigma_c = \Tr_b(\Tr_{M\setminus bc}(|i\rangle\langle i|)).
\end{equation}
Therefore, we conclude
\begin{equation}
    \Tr_{M\setminus c}(|i\rangle \langle i|) = \sigma_c \label{eq:bc_1}
\end{equation}
for all $\vert i\rangle$.

Now we show that any state in the span of $\{|i\rangle \}$ lies in $\Sigma(M)$. Let $|\{|i\rangle \}|=N$. Note that $\frac{1}{N} \sum_{i} |i\rangle \langle i|\in \Sigma(M)$. Moreover, this state is equal to a uniform mixture over $\{ U\vert i\rangle \}$ for any unitary $U$ acting on the span of $\{ |i\rangle \}$. 
Therefore, for any such $U$, we find
\begin{equation}
    \Tr_{M\setminus c}(U|i\rangle \langle i|U^{\dagger}) = \sigma_c
\end{equation}
for all $\vert i\rangle$. This follows from the same logic that leads to Eq.~(\ref{eq:bc_1}). The fact that $\Sigma(M)$ is nonempty follows from the existence of $\sigma \in \Sigma(M)$. This completes the proof.
\end{proof}

\section{Elementary step of the isomorphism theorem} \label{sec:facts_isomorphism_theorem}
Our goal here is to prove the elementary step of the isomorphism theorem; see Proposition~\ref{Prop: iso_ABCD}. The proof involves an alternative formulation of the information convex set, which we denote as $\hat{\Sigma}(\Omega)$. Under Axiom \ref{as:A0'} and \ref{as:A1}, $\Sigma(\Omega)$ becomes equivalent to $\hat{\Sigma}(\Omega).$ However, generally, we do not expect these two sets to be equivalent.

\subsection{An alternative formulation of the information convex set}\label{sec:info_convex_alternative}
The set $\hat{\Sigma}(\Omega)$, which we define in Definition~\ref{def: sigma_hat}, enjoys a number of properties which are not evident from the definition of $\Sigma(\Omega)$. We use these facts to prove Proposition~\ref{Prop: iso_ABCD}. The entire proof is admittedly lengthy and circuitous. A more succinct proof is left for the readers to work on.

The key difference between the definition of $\hat{\Sigma}(\Omega)$ and $\Sigma(\Omega)$ is that the latter involves an extended subsystem\footnote{In the main text, we referred to the extended subsystems as $\Omega'$.} whereas the former does not. 
\begin{definition} \label{def: sigma_hat}
	Let $\hat{\Sigma}(\Omega)$ be a set such that $\forall \rho \in \hat{\Sigma}(\Omega)$
	\begin{enumerate}
		\item $\rho\overset{c}{=} \sigma_{b}$  $\forall \, \sigma_{b}\in \mu$.
		\item $I(A:C)_{\rho}=0$ for the partition in Fig.~\ref{AC_hat}.
		\item $I(A:C\vert B)_{\rho}=0$ for the partitions in Fig.~\ref{ABC_hat}.
	\end{enumerate}
\end{definition}

\begin{figure}[h]
	\centering
\includegraphics[scale=1.1]{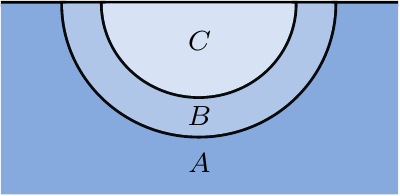}
	\caption{A partition of the subsystem $\Omega$ for defining $\hat{\Sigma}(\Omega)$; see the second condition of Definition~\ref{def: sigma_hat}. Let $\Omega = ABC$ where $BC$ is a subsystem contained in a $\mu$-disk. The horizontal line is the boundary of $\Omega$. Only part of $A$ is shown for illustration purposes.}
	\label{AC_hat}
\end{figure}

\begin{figure}[h]
	\centering
\includegraphics[scale=1.1]{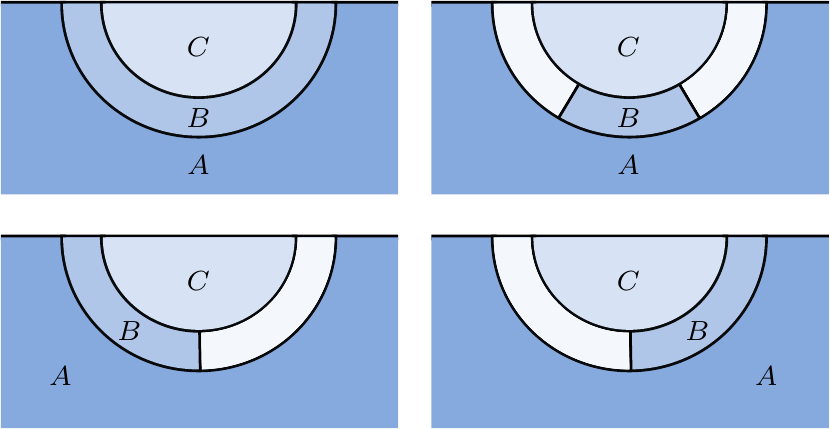}
	\caption{A partition of the subsystem $\Omega$ for defining $\hat{\Sigma}(\Omega)$; see the third condition of Definition~\ref{def: sigma_hat}. Let $\Omega \supseteq ABC$ where $BC$ is a  subsystem contained in a $\mu$-disk. The horizontal line is the boundary of $\Omega$. Only part of $A$ is shown for illustration purposes.}
	\label{ABC_hat}
\end{figure}

This definition does not make it clear why $\hat{\Sigma}(\Omega)$ is convex. The following proposition establishes this fact.
\begin{Proposition}
	$\hat{\Sigma}(\Omega)$ is convex. 
\end{Proposition}
\begin{proof}
Let $\rho_{\Omega}, \lambda_{\Omega}\in \hat{\Sigma}(\Omega)$. We wish to show that their convex combination also satisfies the three conditions in Definition~\ref{def: sigma_hat}. Without loss of generality, consider the convex combination  $p\rho_{\Omega} + (1-p)\lambda_{\Omega}$ where $p\in [0,1]$. 

The first condition in Definition~\ref{def: sigma_hat} is trivially true. For the second condition, note that  $\rho_{AC}= \rho_{A}\otimes \rho_{C}$ and $\lambda_{AC}= \lambda_{A}\otimes \lambda_{C}$. Moreover, because $C$ is in a $\mu$-disk, $\rho_C = \lambda_C = \sigma_C$. Therefore, any convex combination of $\rho$ and $\lambda$ is also factorized over $AC$. Lastly, for the third condition, let us consider the conditional mutual information of the convex combination. 
\begin{equation}
\begin{aligned}
	I(A:C\vert B)_{p\rho + (1-p)\lambda} &=(S_{BC}-S_B)_{p\rho + (1-p)\lambda} + (S_{AB}-S_{ABC})_{p\rho + (1-p)\lambda}\\
	&\le p(S_{BC}-S_B)_{\rho} + (1-p)(S_{BC}-S_B)_{\lambda} \\
	&\quad+ p(S_{AB}-S_{ABC})_{\rho} + (1-p)(S_{AB}-S_{ABC})_{\lambda}\\
	&= p\,I(A:C\vert B)_{\rho} +(1-p)I(A:C\vert B)_{\lambda}\\
	&=0.
\end{aligned}
\end{equation}
To derive this bound, we used the following two facts. First, $BC$ is in a $\mu$-disk. Therefore, $S_{BC} - S_B$ is the same for $\rho, \lambda$, as well as their convex combinations. Second, conditional entropy is concave; see Eq.(\ref{eq: SSA_combination}).

Therefore, $I(A:C\vert B)=0$ for the state $p\rho + (1-p)\lambda$, for any $p\in [0, 1]$. This completes the proof.
\end{proof}

\begin{figure}[h]
	\centering
	\includegraphics[scale=1.2]{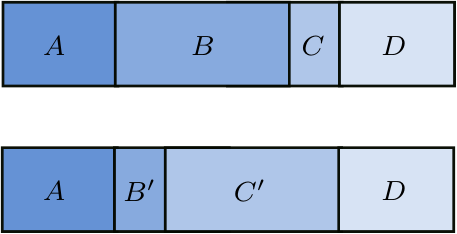}
	\caption{ A schematic depiction of subsystem $ABCD$. The partition $B'C'=BC$ is chosen such that no $\mu$-disk overlaps with both $AB'$ and $CD$. Note that, the subsystems $A,B,C,D$ are allowed to take a variety of topologies.}
	\label{fig: divide BC}
\end{figure}

In a variety of circumstances, elements in $\hat{\Sigma}(ABC)$ and $\hat{\Sigma}(BCD)$ can be merged, and the merging result is an element of $\hat{\Sigma}(ABCD)$. This is the content of the following proposition. 

\begin{Proposition}\label{Prop:mergin_hat}
	Consider two density matrices  $\rho_{ABC}\in \hat{\Sigma}(ABC)$ and  $\lambda_{BCD}\in \hat{\Sigma}(BCD)$. If the following conditions hold, $\rho_{ABC}$ and $\lambda_{BCD}$ can be merged. Moreover, the resulting density matrix is an element of $\hat{\Sigma}(ABCD)$.
	\begin{enumerate}
		\item There exists a partition $B'C'=BC$, such that no $\mu$-disk overlaps with both $AB'$ and $CD$; see Fig.~\ref{fig: divide BC}.
		\item $\rho\overset{c}{=}\lambda$.
		\item $I(A:C\vert B)_{\rho}= I(B:D\vert C)_{\lambda}=0$.
		\item $I(A:C'\vert B')_{\rho}= I(B':D\vert C')_{\lambda}=0$.
	\end{enumerate} 
\end{Proposition}
\begin{proof}
It follows from the conditions of the proposition that $\rho_{ABC}$ and $\lambda_{BCD}$ can be merged. In fact, there are two different ways to merge these density matrices. Using the third condition, we get $\tau_{ABCD}\equiv \mathcal{E}^{\lambda}_{C\to CD}(\rho_{ABC}) $ and using the fourth condition, we get $\tau'_{ABCD}\equiv \mathcal{E}^{\rho }_{B'\to AB'}(\lambda_{BCD})$. 

Note that $\tau_{ABCD}=\tau'_{ABCD}$. This is because both of them satisfy $I(A:D\vert BC)$ and they have the same marginal on $ABC$ and $BCD$. Therefore, Lemma~\ref{lemma_growth} implies that $\tau_{ABCD}=\tau'_{ABCD}$.

Below, we show that $\tau_{ABCD}$ is an element of $\hat{\Sigma}(ABCD)$. In order to prove this claim, we need to show that $\tau_{ABCD}$ satisfies the condition 1, 2, and 3 in Definition~\ref{def: sigma_hat}. Condition 1 is easy to check. Because no $\mu$-disk overlaps with both $AB'$ and $CD$, the overlap between a $\mu$-disk and $ABCD$ is either contained in $ABC$ or $BCD$. 

For the second and the third condition, the key observation is that we have the freedom to choose the quantum channel:
	\begin{equation}
	\tau_{ABCD}\equiv \mathcal{E}^{\lambda}_{C\to CD}(\rho_{ABC}) =\mathcal{E}^{\rho }_{B'\to AB'}(\lambda_{BCD}).
	\end{equation}
For every $\mu$-disk near the boundary of subsystem $ABCD$, we could pick a suitable quantum channel, (either $\mathcal{E}^{\lambda}_{C\to CD}$ or $\mathcal{E}^{\rho}_{B'\to AB'}$), which has no overlap with the $\mu$-disk. 

For the second condition, we use the following fact. Suppose, without loss of generality, we are given a mutual information between two subsystems say $X$ and $Y$. The mutual information does not increase under a quantum channel acting only on either $X$ or $Y$. Indeed, the channels we described above are instances of such quantum channels. Therefore, the mutual information in the second condition is upper bounded by $0$. This subsequently implies that the second condition holds.

For the third condition, we use a similar fact. Now, suppose we are given a conditional mutual information $I(X:Y|Z)$. This also does not increase under a quantum channel acting only on either $X$ or $Y$. Therefore, with the same logic, the third condition holds as well.
\end{proof}

\begin{Proposition}\label{Prop: extension_0}
Consider $\Omega=ABC$ and $\Omega'=ABCD$ whose partitions are depicted in Fig.~\ref{fig:elementary step}, (note that $A$ and $D$ are assumed to be separated by at least $2r+1$).  Let $\calE^{\sigma}_{C\to CD}$ be the Petz map constructed from the reference state density matrix $\sigma_{BCD}$. Then, 
	\begin{eqnarray}
	\Tr_D\circ \calE^{\sigma}_{C\to CD}(\rho_{\Omega}) &=& \rho_{\Omega}, \quad \forall\, \rho_{\Omega}\in \hat{\Sigma}(\Omega),\label{eq:hat_1}\\
	\calE^{\sigma}_{C\to CD}(\rho_{\Omega}) &\in& \hat{\Sigma}(\Omega') ,\quad \forall\, \rho_{\Omega}\in \hat{\Sigma}(\Omega).\label{eq:hat_2}
	\end{eqnarray}
\end{Proposition}

\begin{proof}
Eq.~(\ref{eq:hat_1}) follows directly from Lemma~\ref{Prop: merging technique}. To see why, let us observe that $\rho_{\Omega}$ and $\sigma_{BCD}$ can be merged using Lemma~\ref{Prop: merging technique}. With condition 1, 2, 3 in Definition~\ref{def: sigma_hat}, one could verify the two conditions required in Lemma~\ref{Prop: merging technique}.  First, $\rho_{\Omega}$ is consistent with the global state on any disk $\omega\subseteq \Omega$. In particular,  $\rho_{\Omega}\overset{c}{=}\sigma_{BCD}$. Second, $I(A:C\vert B)_{\rho}= I(B:D\vert C)_{\sigma}=0$.

Eq.~(\ref{eq:hat_2}) is essentially a corollary of Proposition~\ref{Prop:mergin_hat}. It is straightforward to construct the $B'C'$ required in Proposition~\ref{Prop:mergin_hat} and check all the conditions. The $2r+1$ separation between $A$ and $D$ is large enough for the construction. This completes the proof.
\end{proof}

\subsection{Equivalence of the definitions}
Now we can show that $\hat{\Sigma}(\Omega)$ is equivalent to $\Sigma(\Omega)$. This justifies our choice of calling $\hat{\Sigma}(\Omega)$ as the information convex set.
\begin{Proposition} \label{Prop: equivalent_info_convex}
	$	\Sigma(\Omega)= \hat{\Sigma}(\Omega),\,\forall\, \Omega$.
\end{Proposition}

\begin{proof}
	If $\Omega$ is a closed manifold, then it is obvious that $\Sigma(\Omega)= \hat{\Sigma}(\Omega)$. 
	If $\Omega$ has boundaries,  it is easy to show that $\Sigma(\Omega)\subseteq \hat{\Sigma}(\Omega)$ from the assumptions. 
	On the other hand, Proposition~\ref{Prop: extension_0} implies
	that any $\rho_{\Omega}\in \hat{\Sigma}(\Omega)$ can be written as $\rho_{\Omega}=\Tr_{\Omega_{\epsilon}\backslash \Omega } \,\rho_{\Omega_{\epsilon}}$ for some element $\rho_{\Omega_{\epsilon}} \in \hat{\Sigma}(\Omega_{\epsilon})$. This is because of
	Eq.~(\ref{eq:hat_2}) and that $\Omega$ and $\Omega_{\epsilon}$ are connected by a path which consists of a sequence of elementary extensions. 
	It follows that $\rho_{\Omega}\in \Sigma(\Omega)$ and therefore $\Sigma(\Omega) \supseteq \hat{\Sigma}(\Omega)$. Thus, $\Sigma(\Omega) = \hat{\Sigma}(\Omega)$, $\forall\,\Omega$.
	This completes the proof.		
\end{proof}

\subsection{Merging with $\Sigma(\Omega)$}

In a variety of circumstances, we can merge elements in multiple information convex sets into an element of yet another information convex set. This follows from Proposition~\ref{Prop:mergin_hat} and the established equivalence between $\Sigma(\Omega)$ and $\hat{\Sigma}(\Omega)$ (Proposition~\ref{Prop: equivalent_info_convex}).

\begin{Proposition}\label{Prop:merging_Sigma}
	Consider two density matrices  $\rho_{ABC}\in {\Sigma}(ABC)$ and  $\lambda_{BCD}\in {\Sigma}(BCD)$. If the following conditions hold, $\rho_{ABC}$ and $\lambda_{BCD}$ can be merged. Moreover, the resulting density matrix belongs to ${\Sigma(ABCD)}$.
	\begin{enumerate}
		\item There exists a partition $B'C'=BC$, such that no $\mu$-disk overlaps with both $AB'$ and $CD$; see Fig.~\ref{fig: divide BC}.
		\item $\rho\overset{c}{=}\lambda$.
		\item $I(A:C\vert B)_{\rho}= I(B:D\vert C)_{\lambda}=0$.
		\item $I(A:C'\vert B')_{\rho}= I(B':D\vert C')_{\lambda}=0$.
	\end{enumerate} 
\end{Proposition}

\begin{proof}
	The proof directly follows from Proposition~\ref{Prop:mergin_hat} and \ref{Prop: equivalent_info_convex}.
\end{proof}

\begin{remark}
	Importantly, Proposition~\ref{Prop:merging_Sigma} implies that the merged results in Fig.~\ref{Merging_change_patch_III}, \ref{Merging_a_b}, \ref{Merging_1_a}, \ref{Last_fusion_1}, \ref{Merging_3_to_4_III}, \ref{Merging_3_annulus_III}, \ref{Merging_change_1_2_III}  are elements of some information convex sets.
\end{remark}

A special case relevant to the proof of isomorphism theorem is the following corollary.  
\begin{corollary}\label{Prop:Merging_in_Sigma}
	The merging process in Fig.~\ref{fig:elementary step} generates a map from $\Sigma(\Omega)$ to $\Sigma(\Omega')$, i.e.
	\begin{equation}
	\calE^{\sigma}_{C\to CD}(\rho_{\Omega}) \in {\Sigma}(\Omega') ,\quad \forall\, \rho_{\Omega}\in {\Sigma}(\Omega),
	\end{equation}
	where $\calE^{\sigma}_{C\to CD}$ is the Petz map constructed from the reference state density matrix $\sigma_{BCD}$.
\end{corollary}
\begin{proof}
	The proof follows from Proposition~\ref{Prop: extension_0} and \ref{Prop: equivalent_info_convex}.
\end{proof}

\subsection{Proof of Proposition~\ref{Prop: iso_ABCD}}
Now, we are in a position to prove Proposition~\ref{Prop: iso_ABCD}.
\isoabcd*

\begin{proof}
	We have shown that $\Tr_{D}$ is a linear map which maps elements of $\Sigma(\Omega')$ into elements of $\Sigma(\Omega)$ (Proposition~\ref{Prop.Info_convex_basic}). 
	Conversely, upon applying $\calE^\sigma_{C\to CD}$ to $\rho_{\Omega}$ we obtain a merged state $\tau^\rho_{ABCD}$ of $\rho_{\Omega}$ and $\sigma_{BCD}$. (Here $\sigma_{BCD}$ is the reduced density matrix of the reference state.) We can merge them using the merging lemma (Lemma~\ref{Prop: merging technique}) because (1) $\rho_{BC}=\sigma_{BC}$ follows from Proposition~\ref{Prop:disk_in_Omega} and (2) the requisite conditional independence conditions follow from \ref{as:A1}. 
	While the merging lemma guarantees the existence of $\tau^\rho_{\Omega'}$, it remains to show that $\tau^\rho_{\Omega'}$ is an element of $\Sigma(\Omega')$. This fact follows from Corollary~\ref{Prop:Merging_in_Sigma}.
	This step requires that $BC$ is large enough, or equivalently, $A$ and $D$ are separated by enough distance.

Now, it remains to prove Eqs.~\eqref{eq:iso_1} and \eqref{eq:iso_2}. Eq.~\eqref{eq:iso_1} holds because the merged state is consistent with the given marginals; see Eq.~\eqref{eq:marginal}. Eq.~(\ref{eq:iso_2}) follows from the fact that the state on both sides of the equation obey $I(AB:D\vert C)=0$ and that they have the same reduced density matrices over $ABC$ and $CD$; they are equal to $\rho_{ABC}$ and $\sigma_{CD}$. According to Lemma~\ref{lemma_growth}, the two global states must be identical. 
\end{proof}

\section{Extreme points}\label{sec:extreme_points}
In this section, we prove various properties of the extreme points of the information convex set. Throughout this section, we shall often consider a slight ``thickening" of a subsystem. Like the convention we used in the main text, a thickening of a subsystem $\Omega$ is an enlarged subsystem $\Omega'$ which is obtained by expanding the boundaries of $\Omega$. 

If the boundary of $\Omega$ is expanded by a thickness of $\delta$, we shall refer to that subsystem as $\Omega_{\delta}$. For the ensuing analysis, it will be convenient to consider a length scale $\epsilon$, which is comparable to a single lattice spacing for the convention used in the main text. 

Let us begin with the following lemma.
\begin{lemma}
	\label{C2}
	Suppose $\rho_{\Omega_{2\epsilon}}\in \Sigma(\Omega_{2\epsilon})$ can be written as 
	\begin{equation}
	\rho_{\Omega_{2\epsilon}}=\sum_i q_i\, \rho^i_{\Omega_{2\epsilon}}, \label{eq:basis}
	\end{equation}
	where $\{q_i\}$ is a probability distribution with $q_i>0,\, \forall i$ and $\{ \rho^i_{\Omega_{2\epsilon}} \}$ is a set of density matrices. Then,
	\begin{equation}
	\Tr_{\Omega_{2\epsilon}\backslash \Omega}  \, \rho^i_{\Omega_{2\epsilon}} \in \Sigma(\Omega). \label{eq: iso_2_epsilon}
	\end{equation}
\end{lemma}

\begin{proof}
It suffices to show that every $\rho^i_{\Omega_{2\epsilon}}$ reduces to $\sigma_{b}\in \mu$ on any $\mu$-disk $b\subseteq \Omega_{\epsilon}$. In order to show this fact, consider $b_{\epsilon} \subseteq \Omega_{2\epsilon}$. Let $C=b$ and choose $B$ such that $b_{\epsilon} =BC$. Then, the topology of $BC$ is identical to the one shown in Fig.~\ref{A0}. 

Because our axioms hold at a scale larger than the $\mu$-disks(Proposition~\ref{Prop:axiom_large}), and that $\rho_{BC}=\sigma_{BC}$(Proposition~\ref{Prop:disk_in_Omega}), we conclude that $(S_{BC}+S_{C}-S_{B})_{\rho}=0$.
Now, apply Lemma~\ref{Prop: the isolation technique} $(1)\Rightarrow(3)$. We conclude
	\begin{equation}
	\rho^i_b = \rho_b=\sigma_b     \nonumber
	\end{equation} 
for any $i$ and  any $\mu$-disk $b\in \Omega_{\epsilon}$. This completes the proof.
\end{proof}

The following lemma shows that any element in the information convex set has zero conditional mutual information for onion-like partitions; see Fig.~\ref{Growth_around_an_entanglement_cut}.
\begin{lemma}
	\label{Prop: QMS_around_an_entangle_cut}
Let	$\Omega=ABC$. Suppose  $B$ and $C$ are concentric annuli described in Fig.~\ref{Growth_around_an_entanglement_cut}. Then,
	\begin{equation}
	I(A:C\vert B)_{\rho}=0,\quad \forall \, \rho_{ABC}\in \Sigma(ABC). \label{eq: ex_entanglement_cut}
	\end{equation}
\end{lemma} 

\begin{proof}
Let us consider a sequence of regions $C_m$ with $m=0,1,2,3,\cdots M$ with $C_0=\emptyset$, $C_M=C$, and $C_i\subseteq C_{i+1}$. Here $BC_{i+1}$ is an ``infinitesimal deformation" of $BC_i$ as in Fig.~\ref{Grow_a_disk_12_PDF}.

We want to show that
\begin{equation}
(S_{ABC_i}- S_{BC_i})_{\rho}  =   (S_{ABC_{i+1}}- S_{BC_{i+1}})_{\rho}, \label{eq: i_to_i+1}
\end{equation}
for $ i=0,\cdots, M-1$. Equivalently, we can show that
\begin{equation}
I(A:\delta C_{i+1}|BC_{i}) = 0, \label{eq:cmi_infinitesimal}
\end{equation}
where $\delta C_{i+1} = C_{i+1} \backslash C_{i}$. One can upper bound the left hand side of Eq.~(\ref{eq:cmi_infinitesimal}) by $0$. To see why, first use $I(AA':C\vert B)\ge I(A:C\vert A'B)$ so that
\begin{equation}
    I(A:\delta C_{i+1}|BC_i) \leq I(ABC_i\backslash D_i:\delta C_{i+1}\vert D_i), \label{eq:temp_upper}
\end{equation}
where $D_i \subseteq BC_i$ is a disk-like subsystem that separates $\delta C_{i+1}$ from $ABC_i\backslash D_i$. In particular, we choose $D_i$ such that $D_i \delta C_{i+1}$ is contained in a $\mu$-disk. Then, one can upper bound the right-hand side of Eq.~\eqref{eq:temp_upper} by $0$, by using Axiom~\ref{as:A1}.
Therefore,
	\begin{equation}
	(S_{ABC_0}- S_{BC_0})_{\rho}  =   (S_{ABC_{M}}- S_{BC_{M}})_{\rho},\nonumber
	\end{equation}
	which justifies Eq. (\ref{eq: ex_entanglement_cut}).
\end{proof}

\begin{figure}[h]
	\centering
\includegraphics[scale=1.15]{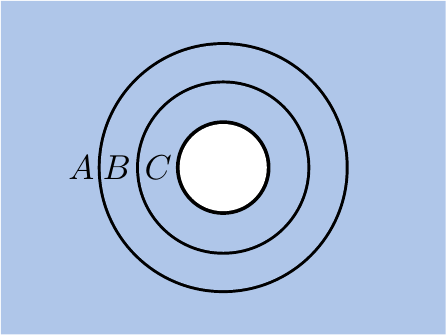}
	\caption{Here $\Omega=ABC$ is an arbitrary subsystem with a boundary. $B$ and $C$ are concentric annuli covering the boundary.}
	\label{Growth_around_an_entanglement_cut}
\end{figure}

\begin{lemma} \label{Prop: purification}
	Consider an extreme point $\sigma^{\langle e\rangle}_{\Omega_{2\epsilon}}\in \Sigma(\Omega_{2\epsilon})$ written as
	\begin{equation}
	\sigma^{\langle e\rangle }_{\Omega_{2\epsilon}}=\sum_i q_i\, \rho^i_{\Omega_{2\epsilon}} \label{eq:basis_extreme}
	\end{equation}
	where $\{q_i\}$ is a probability distribution with $q_i>0, \forall i$ and $\{ \rho^i_{\Omega_{2\epsilon}} \}$ is a set of density matrices. Then, 
	\begin{equation}
	\Tr_{\Omega_{2\epsilon}\backslash \Omega} \,\rho^i_{\Omega_{2\epsilon}} =\Tr_{\Omega_{2\epsilon}\backslash \Omega} \,\sigma^{\langle e\rangle}_{\Omega_{2\epsilon}} ,\quad \forall \, i \nonumber
	\end{equation}
	is the same extreme point of $\Sigma(\Omega)$. 
	
\end{lemma}
\begin{proof}
It follows from Lemma~\ref{C2} that $\Tr_{\Omega_{2\epsilon}\backslash \Omega} \,\rho^i_{\Omega_{2\epsilon}} $ is an element of $\Sigma(\Omega)$ for all $i$. The nontrivial statement is that the reduced state on $\Omega$ is an extreme point and that the reduced state does not depend on $i$.

Suppose there is a dependences on $i$. Then $\Tr_{\Omega_{2\epsilon}\backslash \Omega} \,\sigma^{\langle e\rangle }_{\Omega_{2\epsilon}}$ must be a convex combination of these distinct elements, so this density matrix cannot be an extreme point. This contradicts with the isomorphism theorem (Theorem~\ref{thm: the isomorphism theorem}), which implies that $\Tr_{\Omega_{2\epsilon}\backslash \Omega} \,\sigma^{\langle e\rangle }_{\Omega_{2\epsilon}}$ is an extreme point of $\Sigma(\Omega)$.
(Recall that any linear bijective map between two convex sets must map extreme points to extreme points.) Therefore, the density matrix $\Tr_{\Omega_{2\epsilon}\backslash \Omega} \,\rho^i_{\Omega_{2\epsilon}} $ is independent of $i$ and it follows that it is an extreme point of $\Sigma(\Omega)$.
\end{proof}

\begin{lemma}\label{Lemma_similar_to_A0}
	Consider an extreme point $\sigma^{\langle e\rangle}_{\Omega_{2\epsilon}} \in \Sigma(\Omega_{2\epsilon})$ and let $B=\Omega_{2\epsilon}\backslash\Omega$, then
	\begin{equation}
	(S_{B\Omega } + S_{\Omega}- S_{B})_{\sigma^{\langle e\rangle}}=0. \label{eq:generalized_A0}
	\end{equation}
\end{lemma}

\begin{proof}
	This result follows from Lemma~\ref{Prop: purification} and $(3)\Rightarrow (1)$ of Lemma~\ref{Prop: the isolation technique}.
\end{proof}

The significance of Lemma~\ref{Lemma_similar_to_A0} is that it applies to any subsystem. In particular, for 2-hole disks, we obtain the following corollary.
\begin{corollary}\label{coro:generalized_A0}
	Let $Y$ be a 2-hole disk divided according to Fig.~\ref{Y_BC}, i.e. $Y=BY'$. Let $\sigma^{\langle e\rangle}_{Y}$ be an extreme point of $\Sigma(Y)$, then
	\begin{equation}
	(S_{BY'} + S_{Y'} - S_B)_{\sigma^{\langle e\rangle}}=0. \label{eq: isolation_BCY}
	\end{equation}
\end{corollary}

\begin{lemma} \label{Prop: product}
Let $\Omega=ABC$ with a choice of subsystems described in Fig.~\ref{Growth_around_an_entanglement_cut}. If $\sigma^{\langle e\rangle}_{\Omega}$ is an extreme point of $\Sigma(\Omega)$:
	\begin{equation}
	I(A:C)_{\sigma^{\langle e\rangle}}=0. \label{eq:vanishing_mutual_general}
	\end{equation}
\end{lemma}

\begin{proof}
	Since $\sigma^{\langle e\rangle}_{\Omega}\in \Sigma(\Omega)$, from Lemma~\ref{Prop: QMS_around_an_entangle_cut}, 
	\begin{equation}
	I(A:C\vert B)_{\sigma^{\langle e\rangle}}=0.\label{eq: QMS_1}
	\end{equation}
	Then, it follows from the explicit structure of quantum Markov state Eq.~(\ref{eq: QMS_explicit}) and (\ref{eq: QMS_explicit_mutual}) that 
	\begin{equation}
	\Tr_B \,\sigma^{\langle e\rangle}_{\Omega} = \sum_i p_i\, \rho^i_{A}\otimes \rho^i_{C},
	\end{equation}
	where $\rho^i_{A}$ and $\rho^i_{C}$ are density matrices (which may or may not belong to $\Sigma(A)$ and  $\Sigma(C)$). $\{p_i\}$ is a probability distribution. We know from the isomorphism theorem~\ref{thm: the isomorphism theorem} that $\sigma^{\langle e\rangle}_{A}\equiv \sum_i p_i \,\rho^i_{A} $ is an extreme point of $\Sigma(A)$.
	
	Since  $A$ is thick enough, let $A = A'_{2\epsilon}$. Here  $A'$ has the same topology as $A$ but it is thinner. From Lemma~\ref{Prop: purification} we know that  $\Tr_{A\backslash A'} \,\rho^i_{A} = \sigma^{\langle e\rangle}_{A'}$, $\forall \,i$. Therefore,  $I(A':C)_{\sigma^{\langle e\rangle}}=0$. Since we could enlarge  $A'$ (as that in Fig.~\ref{Grow_a_disk_12_PDF}) until it recovers $A$ without changing the mutual information, we conclude that
	$I(A:C)_{\sigma^{\langle e\rangle}}=0.$
\end{proof}

As an immediate application of Lemma~\ref{Prop: product}, we can prove the following factorization property between subsystem $L$ and $R$ described in Fig.~\ref{Annulus_LMR}. This plays an important role in the proof of the orthogonality of the extreme points.

\begin{corollary}\label{coro:product_X}
	For the annulus $X=LMR$  in Fig.~\ref{Annulus_LMR}(a), for any extreme point $\sigma^a_{X}\in \Sigma(X)$,
	\begin{equation}
	\Tr_{M}\, \sigma^a_{X} = \sigma^a_{L}\otimes \sigma^a_{R}, \label{eq: product}
	\end{equation}
	where $\sigma^a_{L}$ and  $\sigma^a_{R}$ are the reduced density matrices of $\sigma^a_{X}$ on $L$ and $R$ respectively.
\end{corollary}
\begin{corollary}\label{coro:entropy_B_B1B2B3}
	Consider the partition of a 2-hole disk $Y$ in Fig.~\ref{Y_BC}, i.e. $Y=Y'B$ and $B=B_1B_2B_3$.  Let $\sigma_Y^{\langle e\rangle}$ be an extreme point of $\Sigma(Y)$, then
	\begin{equation}
	(S_{B_1}+S_{B_2} +S_{B_3}-S_{B})_{\sigma^{\langle e\rangle}}=0. \label{eq:entropy_B_B1B2B3}
	\end{equation}
\end{corollary}
Note that Eq.~\eqref{eq:entropy_B_B1B2B3} is equivalent to saying that $\sigma_{B_1B_2B_3}^{\langle e\rangle}$ is a tripartite product state.

\subsection{Orthogonality of extreme points}\label{Proof_of_Prop: structure_2}
Below, we present the proof of Theorem~\ref{Prop: structure_2}, which establishes the orthogonality of extreme points.

\simplex*

\begin{proof}
	Let us divide an annulus $X$  according to Fig.~\ref{Annulus_LMR}(a) and consider the path in Fig.~\ref{Annulus_LMR}(b) which defines an isomorphism $\Sigma(L)\cong \Sigma(X)\cong \Sigma(R)$. Let us consider an extreme point $\sigma^{a}_{X}$ and let its image in $\Sigma(L)$  and $\Sigma(R)$ be $\sigma^a_{L}$  and $\sigma^a_{R}$. It follows from the isomorphism theorem (Theorem~\ref{thm: the isomorphism theorem}) that (1) they are extreme points; (2) they are the reduced elements of $\sigma^a_{X}$.
	For a pair of extreme points $\sigma^a_X$ and $\sigma^b_{X}$,
	\begin{equation}
	F(\sigma^a_X,\sigma^b_X) = F(\sigma^a_L,\sigma^b_L)= F(\sigma^a_R,\sigma^b_R), \label{eq: fidelity_equal}
	\end{equation}
	since the isomorphism preserves fidelity.
	According to Corollary~\ref{coro:product_X},  reducing each extreme point ($\sigma^a_{X}$) to $LR$ gives a tensor product structure ($\sigma^a_L\otimes \sigma^a_R$). Thus,
	\begin{equation}
	\begin{aligned}
	F(\sigma^a_{X},\sigma^b_{X})&\le F(\sigma^a_L \otimes\sigma^a_{R} , \sigma^b_L\otimes \sigma^b_{R})\\
	&= F(\sigma^a_{L}, \sigma^b_{L})\cdot  F(\sigma^a_{R}, \sigma^b_{R}). \label{eq: fidelity_le}
	\end{aligned}
	\end{equation}
	The first line follows from the monotonicity of fidelity, namely, the fact that fidelity is nondecreasing when restricted to a smaller region ($LR\subseteq X$).
	Eqs.~(\ref{eq: fidelity_equal}) and (\ref{eq: fidelity_le}) imply that $F(\sigma^a_{X},\sigma^b_{X})$ is either 0 or 1. If $F(\sigma^a_{X},\sigma^b_{X})=1$ then $\sigma^a_X=\sigma^b_X$, so $a=b$. If $F(\sigma^a_{X},\sigma^b_{X})=0$ then   $\sigma^a_X\cdot \sigma^b_X=0$, i.e. $\sigma^a_X$ and $\sigma^b_X$ live on orthogonal subspaces. This justifies the direct sum structure. Since a finite dimensional Hilbert space could only accommodate a finite number of orthogonal subspaces, the extreme points form a finite set.
\end{proof}

\subsection{Implication of the orthogonality}
\label{sec:orthogonality}

Based on the orthogonality of the extreme points, we can prove several new facts about the elements of the information convex set. In the remainder of this section, we use both the isomorphism theorem (Theorem~\ref{thm: the isomorphism theorem}) and  the simplex theorem (Theorem~\ref{Prop: structure_2}). Let us begin with a succinct formula for the mutual information.
\begin{Proposition}\label{Lemma_mutual_annulus}
	Let $\rho_{X}= \sum_a p_a \sigma^a_{X}$ be an element of $\Sigma(X)$, written in terms of the orthogonal extreme points. Let $X=LMR$ be a subsystem described in Fig.~\ref{Annulus_LMR}(a). Then,
	\begin{equation}
	I(L:R)_{\rho} = -\sum_a p_a \ln p_a.\label{eq.mutual}
	\end{equation}  
\end{Proposition}
A similar result has been obtained in \cite{2015arXiv150807006J} using Chern-Simons theory. We obtained the same result as a consequence of \ref{as:A0'} and \ref{as:A1}.
\begin{proof}
	From Theorem~\ref{thm: the isomorphism theorem} we know that the reduced elements of $\sigma^a_{X}$ on $L$ and $R$, which we call as $\sigma_{L}^a$ and $\sigma_R^a$, are extreme points of $\Sigma(L)$ and $\Sigma(R)$ respectively. Moreover, it follows from Corollary \ref{coro:product_X} that $\Tr_M \sigma_X^a =\sigma_L^a\otimes \sigma_R^a$. From the orthogonality relation (Eq.~(\ref{Eq:directsumconv})) of Theorem~\ref{Prop: structure_2}. We obtain
	\begin{equation}
	\begin{aligned}
	(S_L)_{\rho} &= H(p) + \sum_{a} p_a S(\sigma^a_L),\\
	(S_R)_{\rho} &=H(p) + \sum_{a} p_a S(\sigma^a_R),\\
	(S_{LR})_{\rho} &= H(p)+ \sum_{a} p_a (S(\sigma^a_L)+S(\sigma^a_R)),\nonumber
	\end{aligned}
	\end{equation} 
	where $H(p)\equiv -\sum_a p_a \ln p_a$ is the Shannon entropy of the probability distribution $\{p_a\}$. The claim follows straightforwardly from these equations.
\end{proof}

As a consequence, we can identify the extreme points of $\Sigma(X)$ as the elements with zero correlation between $L$ and $R$.
\begin{corollary}\label{Coro_1}
	An element  $\rho_{X}\in \Sigma(X)$ is an extreme point if  only if 
	$I(L:R)_{\rho} =0$ for the partition in Fig.~\ref{Annulus_LMR}(a).
\end{corollary}

\begin{lemma} \label{Prop: product_entanglement_cut}
Let $\Omega = ABC$ with a choice of subsystems described in Fig.~\ref{Growth_around_an_entanglement_cut}. If $\sigma^{\langle e\rangle}_{\Omega}$ is an extreme point of $\Sigma(\Omega)$,	$\Tr_{\Omega \backslash C}\,\sigma^{\langle e\rangle }_{\Omega}$ is an extreme point of $\Sigma(C)$.
\end{lemma}
\begin{proof}
Let us consider an annulus subset $A'\subseteq A$ next to $BC$. $I(A':C)_{\sigma^{\langle e\rangle}}=0$ follows from Lemma~\ref{Prop: product}. Then, because $A'BC$ is a partition of annulus similar to $LMR$ in Fig.~\ref{Annulus_LMR}(a),  $I(A':C)=0$ guarantees that $\Tr_{\Omega \backslash C}\,\sigma^{\langle e\rangle }_{\Omega}$ is an extreme point of $\Sigma(C)$. This step follows from Corollary \ref{Coro_1}. 
\end{proof}

The proof of Proposition~\ref{Prop. 1 is extreme} and Lemma~\ref{prop:topological_deformation} are also discussed below.

\propextreme*
\begin{proof}
	Without loss of generality, assume there is a disk $D$ surrounded by  $R$ (in Fig.~\ref{Annulus_LMR}(a)). Let us consider the disk $\omega=LMRD$. There is a constraint $S_{MRD}+ S_{RD}-S_{M}=0$ for the unique element of $\Sigma(\omega)$, because of the enlarged version of \ref{as:A0'} (Proposition~\ref{Prop:axiom_large}).
	This implies that $I(L:R)_{\sigma^1}=0$. Then it follows from Corollary \ref{Coro_1} that $\sigma^1_X$ is an extreme point of $\Sigma(X)$.  
\end{proof}

\deformation*
\begin{proof}
The key of the proof is the existence of an extension $X^t_{(i)}\to \tilde{X}^{t}_{(i)}$, where $B\subseteq\tilde{X}^{t}_{(i)}$ for each path label $i=1,2$ and for any time step $t$;  see Fig.~\ref{Topo_deformation}. For both $i=1, 2$, the sequence of configurations $\{ \tilde{X}^t_{(i)} \}$ forms a path.

Let us show that two isomorphisms on the ``extended paths", $\Phi_{\{ \tilde{X}^t_{(1)} \}}$ and $\Phi_{\{ \tilde{X}^t_{(2)} \}}$, are identical. (Each of them maps $\Sigma(\tilde{X}^0)$ to $\Sigma(\tilde{X}^1)$.) Because every subsystem in the paths contains annulus $C$, the reduced density matrix on $C$ is unchanged during the process, (for any $t$). Since an element in $\Sigma(\tilde{X}^0)$ (or $\Sigma(\tilde{X}^1)$) is uniquely determined from its reduced density matrix on $C$ independent of the chosen path, the two isomorphisms $\Phi_{\{ \tilde{X}^t_{(1)} \}}$ and $\Phi_{\{ \tilde{X}^t_{(2)} \}}$ must be identical. 

To complete the proof, note that $\Phi_{\{ {X}^t_{(1)} \}}$ and $\Phi_{\{ {X}^t_{(2)} \}}$ are determined from $\Phi_{\{ \tilde{X}^t_{(1)} \}}$ and $\Phi_{\{ \tilde{X}^t_{(2)} \}}$ by taking a partial trace. This completes the proof.
\end{proof}

\section{Fusion space}
\label{section:fusion_space_proofs}
Here, we prove that the convex set $\Sigma_{ab}^c(Y)$ is isomorphic to the state space of a finite-dimensional fusion space $\mathbb{V}_{ab}^c$. Recall that $a,b,$ and $c$ are superselection sectors and $Y$ is a $2$-hole disk. We begin with the following lemma, which characterizes the extreme points of $\Sigma_{ab}^c(Y)$.

\begin{lemma}\label{lemma_identical_entropy}
Let $Y$ be a 2-hole disk. Then, every extreme point of $\Sigma_{ab}^c(Y)$ has the same von Neumann entropy.
\end{lemma}
\begin{proof}
Consider a partition of $Y$ into $Y=BY'$ as described in Fig.~\ref{Y_BC}. Let $\lambda_Y$ and $\rho_Y$ be two extreme points of $\Sigma_{ab}^c(Y)$ and $\delta\equiv S(\lambda_Y)- S(\rho_Y)$ be the entropy difference between the two states. Eq.~(\ref{eq: isolation_BCY}) and (\ref{eq:entropy_B_B1B2B3}) imply that $(S_Y + S_{Y'})_{\lambda}$ and $(S_Y + S_{Y'})_{\rho}$ are identical. Also, by the isomorphism theorem $(S_Y + S_{Y'})_{\lambda}-(S_Y + S_{Y'})_{\rho}= 2\delta$. Thus, $\delta=0$.
\end{proof}

We have seen that all the extreme points of $\Sigma_{ab}^c(Y)$ have the same entropy. It follows that a non-extreme point, which is a convex combination of multiple extreme points, must have higher entropy. This fact follows from the general property of von Neumann entropy, $S(\sum_i p_i \rho^i)\ge \sum_i p_i S(\rho^i)$, where $\{p_i\}$ is a probability distribution with $p_i>0$. The equality is achieved if and only if all $\rho^i$ are identical. Thus, we have the following corollary.

\begin{corollary}\label{corollary_optional}
    If the entropy of an element $\rho_Y\in \Sigma_{ab}^c(Y)$ is identical to that of an extreme point of $\Sigma_{ab}^c(Y)$, then $\rho_Y$ itself is an extreme point of $\Sigma_{ab}^c(Y)$. 
\end{corollary}

\twohole*

\begin{proof}
Recall that we have already partitioned $Y$ into $BY'$; see Fig.~\ref{Y_BC}. We shall consider two different partitions of $B$. In Fig.~\ref{Y_BC}, we have already considered a partition of $B$ into $B=B_1B_2B_3$, which is a (disjoint) union of three annuli ($B_1,B_2,$ and $B_3$). We shall also consider a different partition of $B=B_LB_MB_R$. Here $B_L$ is a (disjoint) union of three ``outermost" annuli and $B_R$ is a (disjoint) union of three ``innermost" annuli; see Fig.~\ref{5_subsystems}. 

In total, we are considering $9$ disjoint subsets of $B$, $Y',$ and $E$; see Table~\ref{tab:B_partition} for a detailed discussion on the partition of $B$. Here $E$ is an auxiliary Hilbert space used to purify a density matrix supported on $Y=BY'$.
\begin{table}[h]
    \centering
    \begin{tabular}{c|c|c|c}
    Partitions & $B_1$ & $B_2$ & $B_3$ \\
    \hline 
    $B_R$ (Inner) & $B_{1R}$ & $B_{2R}$ & $B_{3R}$ \\
    $B_M$ (Middle) & $B_{1M}$ & $B_{2M}$ & $B_{3M}$ \\
    $B_L$ (Outer) & $B_{1L}$ & $B_{2L}$ & $B_{3L}$
    \end{tabular}
    \caption{A partition of $B$ used in the proof of Theorem~\ref{Thm:}.}
    \label{tab:B_partition}
\end{table}

\begin{figure}[h]
	\centering
\includegraphics[scale=1.1]{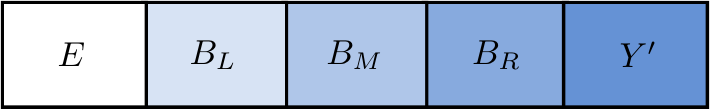}
\caption{A partition of $B$ used in the proof of Theorem~\ref{Thm:}. This figure does not represent the actual underlying geometry. Rather, it represents the relative distance between the ``inner" part of $Y$(i.e., $Y'$) and the annuli surrounding $Y'$(i.e., $B_{R}, B_M$, and $B_L$). Auxiliary system $E$ is introduced to purify the extreme points. Here $B_R$ is the ``innermost" part of $B$ that is directly in contact with $Y'$, $B_M$ is a disjoint union of annuli surrounding $B_R$, and $B_L$ surrounds $B_M$.}
	\label{5_subsystems}
\end{figure}

The statement is trivially true if $\Sigma_{ab}^c(Y)$ is empty. In this case $\dim \mathbb{V}_{ab}^c=0$. For a nonempty $\Sigma_{ab}^c(Y)$, we use $\{ \sigma^{\langle e\rangle_x}_Y$\} to denote the set of extreme points of $\Sigma_{ab}^c(Y)$ and use $\lambda_Y$ for a generic element of $\Sigma_{ab}^c(Y)$. For the extreme points, the alphabet $e$ signifies that the density matrix is an extreme point. They are labeled by $x,y,$ and $z$ in this proof.
	
Pick an extreme point, say $\sigma^{\langle e\rangle_x}_{Y}$. Let us purify $\sigma^{\langle e\rangle_x}_{Y}$ into $\vert\varphi^x_{EY}\rangle$ and let $\rho_{EB}$ be its reduced density matrix on $EB$. From Corollary \ref{coro:generalized_A0}, one can verify  $I(E:B_M B_R\vert B_L)_{\rho}=0$. Moreover, for $\forall\,\lambda_Y \in \Sigma^{c}_{ab}(Y)$ we have $I(B_L: Y'\vert B_M B_R)_{\lambda}=0$   and $\rho_{EB} \overset{c}{=}\lambda_Y$. From the merging lemma (Lemma~\ref{Prop: merging technique}), there is a quantum channel $\calE^{\rho}_{B_L\to EB_L}$ which defines a set of states 
\begin{equation}
	\calS_{EY}\equiv \{ \calE^{\rho}_{B_L\to EB_L}(\lambda_Y) \,\vert \, \lambda_Y\in \Sigma_{ab}^c(Y) \} \label{eq:S_EY}
\end{equation} 
obtained from merging $\rho_{EB}$ with  $\lambda_Y$. It follows that $\Sigma_{ab}^c(Y)\cong \calS_{EY}$. Below, we will determine the structure of $\calS_{EY}$.

Let us first show that the extreme points of $\calS_{EY}$ are pure states. For the particular extreme point we have already considered, i.e., $\sigma^{\langle e\rangle_x}_Y$, the merged state is obviously the pure state $\vert \varphi^x_{EY}\rangle$. Because all the extreme points in $\Sigma_{ab}^c(Y)$ have the same entropy (Lemma~\ref{lemma_identical_entropy}) and because the entropy difference is preserved under the map $\mathcal{E}_{B_L\to EB_L}^{\rho}$ (property (3) of Lemma~\ref{Prop: merging technique}), all the other extreme points are also mapped to pure states.  This means the quantum channel $\calE^{\rho}_{B_L\to EB_L}$ purifies all the extreme points of $\Sigma_{ab}^c(Y)$ simultaneously.

Now, we show $\calS_{EY}$ is the state space of a finite dimensional Hilbert space. The nontrivial statement is that any superposition of pure states in  $\calS_{EY}$ is again in $\calS_{EY}$. Once  this statement is verified, the finiteness of dimension follows straightforwardly from the fact that $\Sigma_{ab}^c(Y)$ is finite dimensional. To prove this claim, we consider a normalized state
	\begin{equation}
	\vert \varphi^z_{EY}\rangle \equiv  \sum_{i} z_i \vert \varphi^{y_i}_{EY}\rangle, \label{eq:varphi_z}
	\end{equation}
where $\{ \vert \varphi^{y_i}_{EY}\rangle \}$ is the set of purifications (in  $\calS_{EY}$) of a finite (sub)set of extreme points $\{ \sigma_Y^{\langle e\rangle_{y_i}}\in \Sigma_{ab}^c(Y) \}$. The complex numbers $z_i$ can be arbitrary as long as $\vert \varphi^z_{EY}\rangle$ is normalized. 	It is sufficient to show $\vert \varphi^{z}_{EY}\rangle\langle \varphi^{z}_{EY}\vert \in \calS_{EY}$. A proof is done by the following steps.
	\begin{enumerate}
		\item The reduced density matrix of $\vert \varphi^{z}_{EY}\rangle$ on $EB_LB_M$ is $\rho_{EB_LB_M}$. 
		This is because (1) $I(EB_LB_M:Y')=0$ on $\vert \varphi^{x}_{EY}\rangle$ and (2) $\vert \varphi^{z}_{EY}\rangle= O_{Y'}\vert \varphi^{x}_{EY}\rangle$ for some operator $O_{Y'}$ supported on $Y'$. The first equation follows from Corollary~\ref{coro:generalized_A0}.
		The second equation follows from the fact that $ \vert \varphi^{y_i}_{EY}\rangle $ and $ \vert \varphi^{x}_{EY}\rangle $ have the same reduced density matrix $\rho_{EB}$  $\forall\, y_i$. In fact, Eq.~(\ref{eq:U_A}) implies an explicit choice $O_{Y'}= \sum_i z_i U^i_{Y'}$, where $U^i_{Y'}$ are unitary operators.

		\item The reduced density matrix of $\vert \varphi^{z}_{EY}\rangle$ on $B_MB_RY'$, which we denote as $\sigma^{\langle e\rangle_z}_{B_MB_R Y'}$, is an extreme point of $\Sigma_{ab}^c(B_MB_R Y')$. 
		To see this, we first observe that $\vert \varphi^{z}_{EY}\rangle\langle \varphi^{z}_{EY}\vert \overset{c}{=}\sigma_{b}$ for any $\mu$-disk $b\subseteq (B_MB_RY')_{\epsilon}$. The logic to establish this fact is similar to that leads to the point made above: (1) $\vert \varphi^{x}_{EY}\rangle$ has vanishing correlation between $b$ and $EY\backslash b_{\epsilon}$ and (2) $\vert \varphi^{z}_{EY}\rangle= O_{EY\backslash b_{\epsilon}}\vert \varphi^{x}_{EY}\rangle$  for some $O_{EY\backslash b_{\epsilon}}$.
		Thus,  $\sigma^{\langle e\rangle_z}_{B_MB_R Y'}\in \Sigma(B_MB_R Y')$. 
		Its reduced density matrix $\rho_{B_M}$ determines the charge sectors $a,b,c$. The entropy of $\sigma^{\langle e\rangle_z}_{B_MB_R Y'}$ is identical to that of known extreme points of $\Sigma_{ab}^c(B_M B_R Y')$, e.g., $\sigma^{\langle e\rangle_x}_{B_MB_R Y'}$. Therefore, according to Corollary~\ref{corollary_optional}, $\sigma^{\langle e\rangle_z}_{B_MB_R Y'}$ is an extreme point of $\Sigma_{ab}^c(B_M B_R Y')$.

		\item The state $\vert \varphi^{z}_{EY}\rangle$ has vanishing conditional mutual information $I(B_L: B_R Y' \vert B_M)=0$. Therefore, its reduced density matrix on $Y$ is uniquely determined from its reduced density matrices $\rho_{B_L B_M}$ and $\sigma^{\langle e\rangle_z}_{B_MB_R Y'}$ (by Lemma~\ref{lemma_growth}). Therefore, $\Tr_E \vert \varphi^{z}_{EY}\rangle\langle  \varphi^{z}_{EY}\vert$  is the extreme point of $\Sigma_{ab}^c(Y)$ obtained from an extension of $\sigma^{\langle e\rangle_z}_{B_MB_R Y'}$. We denote this extreme point as  $\sigma^{\langle e\rangle_z}_{Y}$.
		
		\item From the discussion above one can see, for any $\vert \varphi_{EY}^z\rangle$ of the form (\ref{eq:varphi_z}), there exists an extreme point $\sigma^{\langle e\rangle_z}_{Y}$ of $\Sigma_{ab}^c(Y)$ such that
		\begin{equation}
		\vert \varphi^{z}_{EY}\rangle \langle \varphi^{z}_{EY}\vert  = \calE^{\rho}_{B_L\to EB_L}(\sigma^{\langle e\rangle_z}_Y).
		\end{equation} 
		 Thus,  $\vert \varphi^{z}_{EY}\rangle\langle \varphi^{z}_{EY}\vert\in \calS_{EY}$.
	\end{enumerate}
	We have proved that the set $\calS_{EY}$ in Eq.~(\ref{eq:S_EY}) is the state space of some finite dimensional Hilbert space.
	The Hilbert space depends on the purification, but its dimension cannot depend on this detail. The reason is that the state spaces of two finite-dimensional Hilbert spaces are isomorphic if and only if the dimension of the Hilbert spaces are the same and that $\calS_{EY}\cong \Sigma_{ab}^c(Y)$. 
	Therefore, we can assign an abstract finite dimensional Hilbert space $\mathbb{V}_{ab}^c$ with $\dim \mathbb{V}_{ab}^c =N_{ab}^c\in \mathbb{Z}_{\ge0}$, such that
	\begin{equation}
	\Sigma_{ab}^c(Y)\cong \calS_{EY}\cong \calS(\mathbb{V}_{ab}^c).
	\end{equation}
	Here $\calS(\mathbb{V}_{ab}^c)$ is the state space of $\mathbb{V}_{ab}^c$. This completes the proof.
\end{proof}

\lemmafabc*
\begin{proof}
	From Lemma~\ref{lemma_identical_entropy}, we know that all the extreme points of $\Sigma_{ab}^c(Y)$ have the same entropy. It follows that $f(a,b,c)\equiv S(\rho_Y)-S(\sigma^1_Y)$  depends only on the sector $(a,b,c)$. Let us determine $f(a,b,c)$ in terms of the universal contributions $f(\cdot)$ in Definition~\ref{def:f(a)}. Let $Y=BY'$ according to Fig.~\ref{Y_BC}.
	\begin{equation}
	\begin{aligned}
	2f(a,b,c)&= (S_Y + S_{Y'})_{\rho}-(S_Y + S_{Y'})_{\sigma^1}\\
	&=(S_B)_{\rho} - (S_B)_{\sigma^1}\\
	&=(S(\sigma^a_{B_1}) + S(\sigma^b_{B_2}) +S(\sigma^c_{B_3}))\\
	&\,\,\,\,\,\,\,\,- (S(\sigma^1_{B_1}) + S(\sigma^1_{B_2}) +S(\sigma^1_{B_3}))\\
	&=2(f(a)+f(b)+f(c)).
	\end{aligned}
	\end{equation}
	In the first line, we used the isomorphism theorem. Recall that the isomorphism preserves the entropy difference. In the second line, we applied Eq.~(\ref{eq: isolation_BCY}). In the third line, we applied Eq.~(\ref{eq:entropy_B_B1B2B3}). The fifth line follows from Definition~\ref{def:f(a)}.
	This completes the proof.
\end{proof}

\section{Topological entanglement entropy for the Kitaev-Preskill partition}\label{P_Sec. Proof of {Prop: KP_TEE}}

In this section, we provide a proof of Proposition~\ref{Prop: KP_TEE}.

\tee*
\begin{proof}
	In the following, all the von Neumann entropies are calculated for the reference state $\sigma$.  
	By deforming the subsystems using the idea in Fig.~\ref{Grow_a_disk_12_PDF}, one shows that the value $\gamma$ in Eq.~(\ref{eq: KP_TEE}) is invariant under small deformations of subsystem $A$, $B$, $C$. Since large deformations can be built up from small ones, $\gamma$ is a topological invariant. In the following, we calculate its value.
	
	\begin{figure}[h]
	\centering
\includegraphics[scale=1.1]{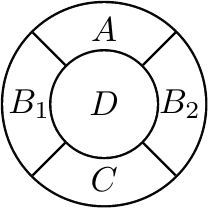}
		\caption{The partition of a disk in the proof, where $B=B_1B_2$ and  that $ABC$ is the Levin-Wen partition.}
		\label{TEE_appendix}
	\end{figure}

	First, let us consider subsystem choice $A,B_1,D$ in Fig.~\ref{TEE_appendix}, and find 
	\begin{equation}
	\gamma = S_{AD} + S_{DB_1} + S_{AB_1} - S_{A} - S_{D} - S_{B_1} -S_{ADB_1}\,. \label{eq: gamma_1}
	\end{equation}
	Second, we consider subsystem choice $AD, B_1, B_2 C$ in Fig.~\ref{TEE_appendix}, and find 
	\begin{equation}
	\begin{aligned}
	\gamma \,\,\,\,=& \,\,\,\,S_{BC} + S_{ADB_1} + S_{ADB_2C}\\
	&- S_{B_2C}  - S_{B_1} - S_{AD} -S_{ABCD}\,. \label{eq: gamma_2}
	\end{aligned}
	\end{equation}
	Both partitions are equivalent to the Kitaev-Preskill partition Fig.~\ref{TEE_LW_PK}(a). 
	Adding up Eq.~(\ref{eq: gamma_1}) and Eq.~(\ref{eq: gamma_2}), and using 
	\begin{eqnarray}
	S_{ABCD} + S_{D} &=&  S_{ABC}, \label{A0_a} \nonumber\\
	S_{DB_1} + S_{ADB_2C} &=& S_{B_1} + S_{AB_2C}, \label{A1_a} \nonumber\\
	S_{AB_2C} &=& S_{AB_2} + S_{B_2C} - S_{B_2}, \label{A1_b}\nonumber\\
	S_{B_1} + S_{B_2} &=& S_{B}, \label{A0_b} \nonumber \\
	S_{AB_1} + S_{AB_2} - S_{A} &=& S_{AB},  \label{A1_c} \nonumber
	\end{eqnarray}
	(these results follow from the quantum Markov chain structure of the global reference state, i.e. the fact that Eqs. (\ref{eq:tau_A0}) and (\ref{eq:tau_A1}) are satisfied on all length scales larger than a constant value), we find
	\begin{eqnarray}
	2\gamma &=& S_{AB} + S_{BC} - S_{B} - S_{ABC}\nonumber\\
	&=& 2 \ln \mathcal{D}. \nonumber
	\end{eqnarray}
	Therefore, $\gamma =\ln \mathcal{D}$.
\end{proof}

\section{Antiparticles and automorphism of annulus on a sphere}\label{Antiparticle_sphere}

In this appendix, we discuss a connection of antiparticle with the automorphisms of the information convex set of an annulus on a sphere. Intuitively, the connection comes from two facts.
First, the automorphism only depends on the topological class of the path that maps the annulus back to itself. This fact is established by Lemma~\ref{prop:topological_deformation}.	Second, on a sphere, the topological class of the paths is described by the braid group on a sphere  \footnote{https://homepages.warwick.ac.uk/~maseay/doc/braids.pdf}.  In general, we use $\mathcal{B}_n(V)$ to denote the $n$-string braid group of manifold $V$. Physically, this is related to the spacetime diagram of $n$ particles braiding on a manifold $V$. In our framework, it is related to the deformation of a subsystem ($V$ with $n$ holes) by a path and then goes back to itself. For our purpose, an annulus is a 2-hole sphere, and the relevant result is the $2$-string braid group on a sphere:
\begin{equation}
\mathcal{B}_2(S^2) = \mathbb{Z}_2.
\end{equation}

For an automorphism of $\Sigma(X)$ generated by a path $\{X^t\}$ with $X^0=X^1=X$, where $X$ is an annulus on a sphere, we could draw a spacetime diagram corresponding to the path. The spacetime diagram shows the braiding of two holes. The braiding belongs to one of the two classes in $\mathcal{B}_2(S^2)=\mathbb{Z}_2$, so does the path.

The  path in the trivial class could be smoothly deformed into the path $X^t=X$, $\forall\,t$. The corresponding automorphism of $\Sigma(X)$ preserves the superselection sectors, i.e., it maps each extreme point back to itself. 

On the other hand, a path in the nontrivial class  generates an automorphism of $\Sigma(X)$ which permutes the extreme points according to \begin{equation}
\Phi(\sigma^a_X) = \sigma^{\bar{a}}_{X},\quad \forall \,a\in \mathcal{C}.\nonumber
\end{equation} 
Intuitively, a nontrivial path switches the pair of holes. Furthermore, if one introduces an oriented loop to the annulus $X$, which deforms smoothly with $X$, then the loop will end up in the opposite orientation after $X$ is mapped back to itself according to the nontrivial path.

\section{String operators}\label{appendix:String}

\begin{Proposition} \label{Prop: U^a string}
	Given a pure reference state $\vert \psi\rangle$, two holes within a disk and $\forall\, a \in \calC$,  there exists a deformable unitary string operator  $U^{(a,\bar{a})}$ supported within the disk, such that the state
	\begin{equation}
	\vert \varphi^{(a,\bar{a})}\rangle \equiv U^{(a,\bar{a})} \vert\psi\rangle,
	\end{equation}
	has topological charges $a$ and $\bar{a}$ in the two holes.
\end{Proposition}

Here, deformable means the support of $U^{(a,\bar{a})}$ can be deformed smoothly while keeping its endpoints fixed.  It is easy to see, for $a=1$, the string can be chosen to be the identity operator while for $a\ne 1$, the string cannot break apart. While this result is not directly useful in the current work, we expect it to play a role in deriving more advanced fusion and braiding properties, e.g., the $S$-matrix.

\begin{figure}[h]
	\centering
	\includegraphics[scale=1.2]{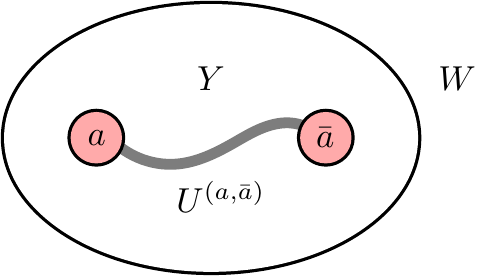}
	\caption{Disk $\omega$ is the union of 2-hole disk $Y$ and its two holes. $W$ is the complement of $\omega$. The topological charges $a$ and $\bar{a}$ within the two holes are created by unitary string operator $U^{(a,\bar{a})}$.
		The support of $U^{(a,\bar{a})}$ is the union of the deformable gray area and the two holes shown in red. }
	\label{string_operator}
\end{figure}

\begin{proof}
	Let $\omega$ be the disk required in the proposition and $W$ is its complement, see Fig.~\ref{string_operator}. One can verify that $\sigma_{W}\equiv \Tr_{\omega} \vert\psi\rangle \langle \psi\vert \in \Sigma(W)$ and that $\sigma_W$ is an extreme point. The 2-hole disk $Y\subseteq \omega$ is obtained by erasing the two holes from disk $\omega$.  From Proposition~\ref{Prop: a_bar}, $\Sigma_{a\bar{a}}^{1}(Y_{2\epsilon})$ contains a unique element which we denote as $\sigma^{a\bar{a}1}_{Y_{2\epsilon}}$, where $Y_{2\epsilon}$ is a thickening of $Y$. Here the subscript $2\epsilon$ means that $Y$ is expanded along its boundary by two lattice spacings; see the beginning of Appendix~\ref{sec:extreme_points} for a related discussion.
	
	The elements $\sigma_{W}$ and $\sigma^{a\bar{a}1}_{Y_{2\epsilon}}$, whose supports are overlapping around the boundary of $\omega_{\epsilon}$, can be merged and the resulting state is an extreme point of $\Sigma(WY_{2\epsilon})$, where $WY_{2\epsilon}$ is again the thickening of $WY$ by two unit lattice spacing.  Let 
	$\vert \varphi^{(a,\bar{a})}\rangle =  \vert \varphi^{(a,\bar{a})}_{WY_{2\epsilon}}\rangle \otimes \vert \varphi_{{V\backslash WY_{2\epsilon} } }\rangle, $
	where $\vert \varphi^{(a,\bar{a})}_{WY_{2\epsilon}}\rangle $ is an eigenvector (with nonzero eigenvalue) of the merged state and $\vert\varphi_{{V\backslash WY_{2\epsilon} } }\rangle$ is an arbitrary pure state.
	According to Lemma~\ref{Prop: purification}, the reduced density matrix of $\vert \varphi^{(a,\bar{a})}\rangle$ on $WY$ is identical to that of the merged state. Therefore, $\vert \varphi^{(a,\bar{a})}\rangle$ has topological charges $a$ and $\bar{a}$ within the two holes, and $\vert \varphi^{(a,\bar{a})}\rangle$ is identical to the reference state $\vert \psi\rangle$ on any subsystem $W'\subseteq WY$ which is connected to $W$ by a path, where the path is within $WY$. In particular, choose $W'$ to be the complement of disk $\omega'$, where $\omega'$ is the union of the gray string and the two holes in Fig.~\ref{string_operator}. Since $\vert \psi\rangle$ and $ \vert \varphi^{(a,\bar{a})}\rangle$ are identical on $W'$, by applying Eq.~(\ref{eq:U_A}), we have $\vert \varphi^{(a,\bar{a})}\rangle= U^{(a,\bar{a})} \vert\psi\rangle$ for a unitary operator $U^{(a,\bar{a})}$ supported on $\omega'$.
	Because we may deform $W'$ and $\omega'$, the support of $U^{(a,\bar{a})}$ can be deformed smoothly.  This completes the proof.
\end{proof}

\FloatBarrier

	\bibliography{ref}
\end{document}